%% file: covert_fading_manuscript.tex
\documentclass[10pt, onecolumn, romanappendices]{IEEEtran}
\usepackage{graphicx}
\usepackage{amsthm}
\usepackage{times}
\input{preamble}


\usepackage{stfloats}
\usepackage{url}
\usepackage{bbm}
\usepackage{cite}
\usepackage{enumerate}
\usepackage[colorlinks = true,
linkcolor = blue,
urlcolor  = blue, 
citecolor = red,
anchorcolor = green]{hyperref}

\title{Second-Order Asymptotics for Covert Communication over Quasi-Static Multiple-Antenna Fading Channels}
\author{Changhong Liu, Jingjing Wang, Qiaosheng Zhang, Jinpeng Xu and Lin Zhou

\thanks{C. Liu, J. Wang and J. Xu are with the School of Cyber Science and Technology, Beihang University (Emails: \{liuchanghong, drwangjj, xujinpeng\}@buaa.edu.cn). C. Liu is visiting the School of Automation and Intelligent Manufacturing, Southern University of Science and Technology.}
\thanks{Q. Zhang is with the Shanghai AI Laboratory (Email: zhangqiaosheng@pjlab.org.cn).}
\thanks{L. Zhou is with the School of Automation and Intelligent Manufacturing, Southern University of Science and Technology (Email: zhoul9@sustech.edu.cn).}
}

\begin{document}
\maketitle

\begin{abstract}
We study the second-order asymptotics of optimal codes for covert communication over quasi-static multi-antenna fading channels, under the covertness metric of Kullback--Leibler (KL) divergence.
In particular, we study all four cases regarding the availability of channel state information (CSI) for the legitimate transmitter and receiver, assume that the warden knows perfect CSI for the channel from the legitimate transmitter to itself, whereas the legitimate transmitter knows only that the warden's channel matrix belongs to a bounded deterministic uncertainty set.
Specifically, we show that, when the blocklength is $n$, the first-order covert rate satisfies the square root law, scaling as $\Theta(n^{-\frac{1}{2}})$ with the coefficient determined by the traces of the channel matrices of the legitimate users and the warden, and the second-order rate vanishes.
We also show that the availability of CSI at the transmitter enables optimal covert power allocation across the spatial sub-channels, which can increase the covert rate, relatively to equal power allocation, without changing the square root law.
Furthermore, we reveal the significant spatial diversity gain provided by multiple-antenna systems for covert communication and demonstrate the critical role of the number of antennas to achieve high throughput covert communication. For the covertness analysis, we extend the quasi-$\eta$-neighborhood framework to quasi-static fading channels.
For the reliability analysis, due to the vanishing power imposed by the covertness constraint, we refine the non-covert analysis by Yang \emph{et al.} (TIT, 2014), by carefully controlling higher-order terms and exploiting the properties of covert outage probability.
\end{abstract}


\section{Introduction}

For 6G and next-generation wireless systems, International Telecommunications Union (ITU) highlights that superlative ultra-reliable low-latency communication (SURLLC) scenarios require latency as small as 0.1 millisecond, roughly one hundredth of the requirement in 5G~\cite{letaief2019roadmap,yuan2025ground}.
Under such stringent delay constraints, quasi-static fading is a realistic channel model, in which the channel coefficients remain constant over an entire codeword~\cite[Section~5.4.1]{tse2005fundamentals}. For quasi-static multiple-antenna fading channels, Yang \emph{et al.} \cite{yang2014quasi} bounded the maximal achievable rate as a function of the blocklength and the tolerate error probability by deriving both non-asymptotic and second-order asymptotic bounds. Due to the outage nature of quasi-static fading, Shannon capacity becomes inapplicable. Instead, the \emph{outage capacity}~\cite[Page 2631]{biglieri2002fading} applies, characterizing the first-order maximal rate for reliable communication over quasi-static fading channels. Although outage capacity is inherently asymptotic, it remains an accurate benchmark for finite blocklength performance in low-latency communication, because the second-order rate vanishes for quasi-static fading~\cite{yang2014quasi}. The above result has been widely used as a theoretical benchmark for performance analyses in wireless communications~\cite{she2017cross,cheng2021adaptive,li2023joint,kallehauge2025prediction}.

Furthermore, ITU also highlights that strong security represents another core feature of 6G communications, which is crucial due to the open nature of wireless communication environments \cite{dang2020should}. Covert communication \cite{bash_limits_2013} has been proposed to confuse malicious warden in sensitive scenarios where the communication behavior needs to be protected \cite{chen2023covert}. The goal of covert communication is to protect the behavior of legitimate transmission so that any warden cannot reliably determine whether the legitimate users are communicating or not. Pioneering work investigated covert communication over an additive white Gaussian noise (AWGN) channel \cite{bash_limits_2013} and discrete memoryless channel (DMC) \cite{che2013reliable,wang_fundamental_2016,bloch_covert_2016}. The square root law was revealed, stating that the number of bits that can be covertly and reliably transmitted scales with the square root of the blocklength. To characterize the theoretical benchmark, the exact coefficient of the square root law is referred to as the first-order covert rate.
The square root law in \cite{bash_limits_2013,che2013reliable,wang_fundamental_2016,bloch_covert_2016} was later extended to various scenarios under different channel conditions \cite{zhang2021covert,wang2021covert,bouette2024covert,bullock2025fundamental} and to multi-user communication systems \cite{soltani2018covert,arumugam2019covert,tan2018time}. To account for the low-latency requirement as well, Tahmasbi and Bloch derived second-order asymptotics for covert communication over any DMC, which was generalized to an AWGN channel 
by Yu \emph{et al.}~\cite{yu_finite_2021,yu_second_2023} and recently further generalized to 
multiple-input multiple-output (MIMO) AWGN channels by Liu, Wang and Zhou~\cite{liu2026covertMIMO}. In particular, the authors of~\cite{liu2026covertMIMO} demonstrated the crucial role of MIMO systems in achieving high-rate covert communication.

However, all above studies of covert communication have focused primarily on non-fading channels, which fail to satisfy intertwined requirements of SURLLC and high security for 6G wireless communications. To bridge this gap, we study the second-order asymptotics of covert communication over a quasi-static MIMO fading channel. In particular, we answer three fundamental questions: \emph{1) how to accurately characterize the outage behavior in covert communication; 2) how the covertness constraint affects the second-order term; and 3) how the number of antennas influences the covert transmission rate}.

\subsection{Main Contributions}
Under the Kullback--Leibler (KL) divergence covertness metric, we bound the maximal achievable rate for covert communication over point-to-point (P2P) quasi-static MIMO fading channels as a function of the blocklength $n$ by deriving both non-asymptotic bounds and second-order asymptotic bounds. Analogous to covert communication results in the non-fading case~\cite{abdelaziz2017fundamental,wang_covert_2021,liu2026covertMIMO}, the first-order rate follows the square root law, scales in the order of $\Theta(n^{-\frac{1}{2}})$, and has the coefficient that is a function of the traces of the channel matrices involving legitimate users and the warden. Furthermore, analogous to the non-covert case \cite{yang2014quasi}, the second-order term vanishes due to quasi-static fading.
In contrast to the non-covert case of \cite{yang2014quasi}, conventional water-filling is generally infeasible under the covertness constraint. Nevertheless, the availability of channel state information (CSI) at legitimate users remains beneficial because it enables covert power allocation across the equivalent sub-channels, which can improve the first-order covert rate.

Our results generalize the analyses of \cite{yu_second_2023,liu2026covertMIMO} to quasi-static fading channels, and extend the second-order asymptotic analyses of non-covert communication in \cite{yang2014quasi} to covert communication.
On the one hand, to handle the covertness constraint, we apply the quasi-$\eta$-neighborhood framework in \cite{yu_second_2023,liu2026covertMIMO} to derive covert power bounds so that the higher-order approximation errors do not alter the second-order term, and optimize power allocation across the equivalent sub-channels.
On the other hand, the finite blocklength performance in~\cite{yang2014quasi} was established for non-vanishing power, which cannot be used when the power is vanishing as in covert communication. If one directly applies the analyses in \cite{yang2014quasi}, an unexpected $O(\frac{1}{\sqrt{n}})$ rate term arises in \cite[Eqs. (163), (266)]{yang2014quasi}, which makes it impossible to even derive the first-order rate for covert communication. 
To address this problem, we carefully verify the conditions of the Cram\'er--Esseen Theorem, analyze the high-order moments of the remainder terms and exploit the properties of covert outage probability (cf. Eqs. \eqref{eq:def of Fout}, \eqref{eq:F' final-3}).

Our theoretical benchmarks bring new insights beyond \cite{yang2014quasi,yu_second_2023,liu2026covertMIMO}, which are illustrated via numerical plots in Figs. \ref{fig:dif_lambda_0}--\ref{fig:fading_reliability}. 
Specifically, Fig.~\ref{fig:dif_lambda_0} illustrates the square root law and the inverse relation between the maximal achievable covert rate and warden's signal-to-noise ratio (SNR).
Fig.~\ref{fig:dif_m} illustrates the benefit of MIMO. In particular, when the number of antennas increases from 1 to 16 and 64, the maximal achievable covert rate increases by approximately 15.4 and 49.7 times, respectively. Finally, Fig.~\ref{fig:fading_reliability} highlights that the second-order term vanishes over quasi-static fading channels. As a result, the finite blocklength performance converges rapidly to the first-order rate. Specifically, for the setting in Fig.~\ref{fig:fading_reliability}, when the blocklength equals 5000, the non-asymptotic covert rate over quasi-static fading channel achieves $71.7\%$ of the asymptotic first-order rate while that for the AWGN case achieves $45.3\%$ of the corresponding first-order rate.

\subsection{Organization for the Rest of the Paper}
The rest of the paper is organized as follows. In Section \ref{2section_model_mainresults}, we formulate the problem by specifying the system model, performance metrics and theoretical benchmarks for covert communication over quasi-static MIMO fading channels. Subsequently, in Section \ref{section_main_results}, we present our main results with discussions. Sections \ref{section:ach csit} and \ref{section:con csit} prove the achievability and the converse parts, respectively. Finally, in Section \ref{section:conclusion}, we conclude our paper and discuss future research directions. For smooth presentation, the proofs of supporting lemmas are deferred to the appendices.

\subsection{Notation}
The set of real numbers, positive real numbers, complex numbers and integers are denoted by $\bbR$, $\bbR_+$, $\bbC$, $\bbN$, respectively. For any two integers $(a,b)\in\bbN^2$ such that $a\leq b$, we use $[a:b]$ to denote the set of integers between $a$ and $b$, and $[a]$ to denote $[1:a]$. We use calligraphic font (e.g., $\calX$) to denote all sets. Random variables are in capital (e.g., $X$) and their realizations are in lower cases (e.g. $x$). Random vectors of length $n\in\bbN$ and their particular realizations are denoted by $X^n:=(X_1,\ldots,X_n)$ and $x^n:=(x_1,\ldots,x_n)$, respectively. We use $\|\cdot\|$, $\|\cdot\|_\rmF$ and $\|\cdot\|_2$ to denote the $\ell_2$ norm of a vector, the Frobenius norm of a matrix and the spectral norm of a matrix, respectively. Given a matrix, we use $\tr(\cdot)$ and $\det(\cdot)$ to denote its trace and determinant, respectively, and use $\spn(\cdot)$ to denote the subspace spanned by its column vectors. Let $\Re(\cdot)$ denote the real part of a complex number, and $\imath$ be the imaginary unit. We use $\bbo(\cdot)$ to denote the indicator function, use logarithms with base $e$ and adopt the standard asymptotic notations such as $\Theta(\cdot)$, $O(\cdot)$ and $o(\cdot)$ (cf. \cite{cormen2022introduction}). 

Given any integer $a\in\bbN$, we use $\bI_a$ to denote the identity matrix of size $a\times a$, and use $\bI_{a,b}$ to denote the $a\times b$ matrix containing the first $b$ columns of $\bI_a$ for $a>b$. The probability density function (pdf) of a circularly-symmetric complex Gaussian random vector with a covariance matrix is denoted by $\calCN(\mathbf{0},\cdot)$. 
The usual gamma function is denoted by $\Gamma(x):=\int_{0}^{\infty}t^{x-1}e^{-t}\rmd t$. The complementary Gaussian cumulative distribution function is denoted by $Q(x):=\int_x^\infty \frac{1}{\sqrt{2\pi}}e^\frac{-x^2}{2}\rmd x$ and its inverse function is denoted by $Q^{-1}(\cdot)$. Given any random variable $X\in\bbR$ and real number $a\in(0,1)$, the $a$-quantile function of $X$ is defined as $\calQ_X(a) := \inf\{x: \Pr\{X \le x\} \ge a\}$.
The set of all probability distributions defined on an alphabet $\calX$ is denoted by $\calP(\calX)$ and the set of all conditional probability distributions from $\calX$ to $\calY$ is denoted by $\calP(\calY|\calX)$.  Given a continuous alphabet $\calX$ and any two distributions $P$ and $Q$ with pdf $f_P$, $f_Q$,\footnote{We use $P$ to denote a probability measure on a continuous alphabet $\calX$, and $f_P$ to denote its pdf.} the total variation (TV) distance is defined as $\bbV(P,Q):=\frac{1}{2}\int_\calX |f_P(x)-f_Q(x)|\rmd x$, and the KL divergence is defined as $\bbD(P\|Q):=\int_\calX f_P(x)\log\frac{f_P(x)}{f_Q(x)}\rmd x$.

\section{Problem Formulation} \label{2section_model_mainresults}
\subsection{System Model}
\label{sec:2.1 channel model}
\begin{figure}[tb]
\centering
\includegraphics[width=0.78\columnwidth]{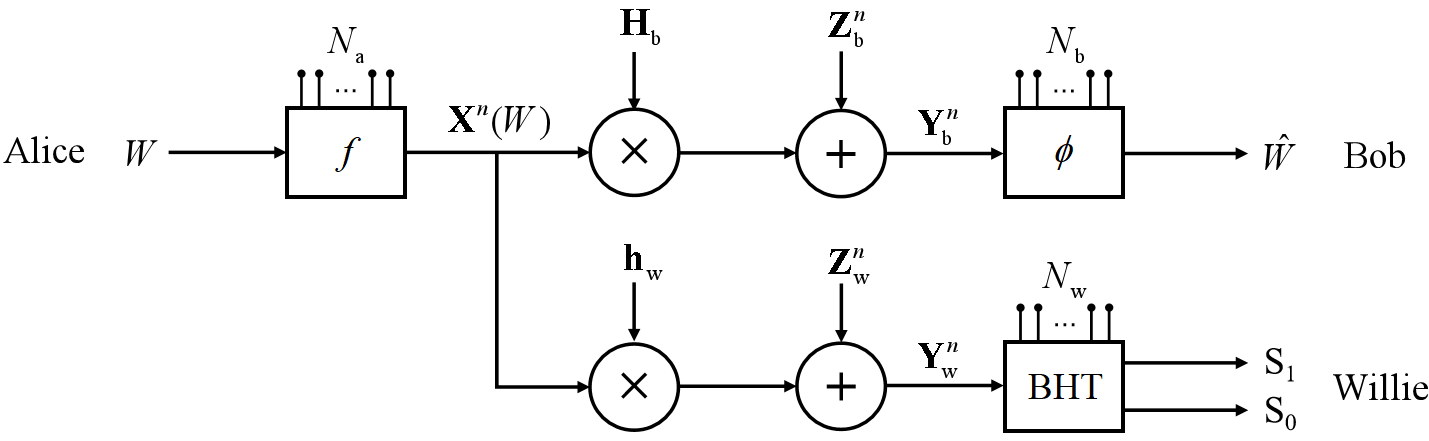}
\caption{System model for covert communication over quasi-static MIMO fading channels.}
\label{fig: covert system model}
\end{figure}

Fix a positive real number $\lambda_0\in\bbR_+$, five integers $(n,M,N_\rma,N_\rmb,N_\rmw)\in\bbN^5$ and five alphabets $(\calX,\calY_\rmb,\calY_\rmw,\calZ_\rmb,\calZ_\rmw)\subseteq \bbC^5$.
As shown in Fig.~\ref{fig: covert system model}, we study covert communication with three users: the transmitter Alice with $N_\rma$ antennas, the receiver Bob with $N_\rmb$ antennas and the warden Willie with $N_\rmw$ antennas. The legitimate channel matrix $\bH_\rmb\in\bbC^{N_\rma\times N_\rmb}$ is random but remains constant over the $n$ channel uses under the quasi-static fading model. Furthermore, consistent with \cite[Section~IV]{wang_covert_2021} and \cite[Section~V]{schaefer2015secrecy}, Willie's channel matrix $\bh_\rmw\in\bbC^{N_\rma\times N_\rmw}$ is modeled as an unknown deterministic matrix with bounded channel gain $\sqrt{\lambda_0}$, i.e., $\bh_\rmw$ falls in the following uncertainty set
\begin{align}
\label{eq:assume of Hw}
\calH_\rmw:=\big\{ \bh_\rmw:\|\bh_\rmw\|_2\leq \sqrt{\lambda_0}\big\},
\end{align}
which is equivalent to a constraint on Willie's received SNR. Consistent with \cite[Section II-A]{liu2026covertMIMO} and \cite[Section II-B]{wang_covert_2021}, we assume that both channel matrices have full rank and $\rank(\bH_\rmb)=\rank(\bh_\rmw)=N_\rma$. This implies that the number of Alice's antennas $N_\rma$ is no larger than those of Bob $N_\rmb$ and Willie $N_\rmw$.



There are two states $\rmS_0$ and $\rmS_1$. Under the state $\rmS_0$, Alice is not transmitting messages. The legitimate receiver Bob receives noise sequence $\bZ_\rmb^n=(\bZ_{\rmb,1},\ldots,\bZ_{\rmb,n})$, which is generated i.i.d. from the complex Gaussian distribution $\calCN(\mathbf{0},\bI_{N_\rmb})$. The warden Willie receives another noise sequence $\bZ_\rmw^n=(\bZ_{\rmw,1},\ldots,\bZ_{\rmw,n})$, generated i.i.d. from $\calCN(\mathbf{0},\bI_{N_\rmw})$. Under the state $\rmS_1$, Alice transmits a uniformly distributed message $W\in[M]$ reliably to Bob and covertly over the warden Willie, with $n$ channel uses. Using an encoder $f$, the legitimate transmitter Alice generates a channel input $\bX^n(W) \in \calX^{n\times N_\rma}$, which is passed over two quasi-static fading channels: for the legitimate receiver Bob, the additive noise $\bZ^n_{\rmb} \in \calZ_\rmb^{n\times N_\rmb}$ corrupts the channel input and yields the noisy channel output $\bY^n_{\rmb} \in \calY_\rmb^{n\times N_\rmb}$; for the warden Willie, the additive noise $\bZ^n_{\rmw} \in \calZ_\rmw^{n\times N_\rmw}$ corrupts the channel input and yields the noisy channel output $\bY^n_{\rmw} \in \calY_\rmw^{n\times N_\rmw}$, i.e.,
\begin{align}
\bY^n_{\rmb}&=\bX^n(W)\bH_\rmb+\bZ^n_{\rmb},\label{eq:yb channel}\\*
\bY^n_{\rmw}&=\bX^n(W)\bh_\rmw+\bZ^n_{\rmw}.\label{eq:yw channel}
\end{align}
After $n$ channel uses, with channel outputs $\bY^n_\rmb=(\bY_{\rmb,1},\ldots,\bY_{\rmb,n})$, Bob uses a decoder $\phi$ to generate an estimate $\hatW$ of the transmitted message $W$. Simultaneously, the warden Willie uses $\bY^n_\rmw=(\bY_{\rmw,1},\ldots,\bY_{\rmw,n})$ and CSI $\bh_\rmw$ to discriminate between states $\rmS_0$ and $\rmS_1$ via binary hypothesis testing (BHT). Consistent with \cite{bash_limits_2013,wang_fundamental_2016,abdelaziz2017fundamental}, we assume that the legitimate users Alice and Bob share a sufficiently long secret key. Consequently, Bob knows when Alice transmits, but Willie does not know the exact transmission time.

We assume that Willie is a powerful warden and knows the perfect CSI $\bh_\rmw$ to perform the optimal BHT, whereas Alice only knows the upper bound $\lambda_0$ in \eqref{eq:assume of Hw}. For CSI $\bH_\rmb$ between legitimate parties Alice and Bob, we consider the following four scenarios:
\begin{itemize}
\item no-CSI: neither the transmitter Alice nor the receiver Bob is aware of the realization of $\bH_\rmb$;
\item CSIT: the transmitter Alice knows $\bH_\rmb$;
\item CSIR: the receiver Bob knows $\bH_\rmb$;
\item CSIRT: both the transmitter Alice and the receiver Bob know $\bH_\rmb$.
\end{itemize}
For brevity, we denote these cases by $\{\rmno,\rmt,\rmr,\rmrt\}$, respectively. When CSIR is available, the legitimate receiver Bob uses maximal likelihood decoding \cite[Appendix A]{polyanskiy_channel_2010},  while Bob uses angle threshold decoding when CSIR is not available (cf. Section \ref{subsec:ach csit coding}).

Furthermore, define $\blambda_0:=\sqrt{\lambda_0}\bI_{N_\rma}$. Let $\bh_{\rmw,0}\in\bbC^{N_\rma\times N_\rmw}$ be a deterministic channel matrix such that
\begin{align}
\label{eq:def of hw0}
\bh_{\rmw,0}\bh_{\rmw,0}^\rmH=\blambda_0^2.
\end{align}
Note that $\bh_{\rmw,0}$ corresponds to the worst case channel for the warden. Our results are derived for $\bh_{\rmw,0}$, which naturally holds any warden's channel in the set $\calH_\rmw$.
The adoption of a worst channel is practical. On the one hand, in covert communication, Willie is generally passive and Alice could not known Willie's location, let alone the channel matrix. On the other hand, a reasonable upper bound on the channel gain always exists with extremely high probability for typical quasi-static fading channels.

\subsection{Performance Metric}
\label{sec:II-B}
Fix three integers $(n,M,K)\in\bbN^3$, and two positive real numbers $(\varepsilon,\delta)\in\bbR_+^2$. Fix any $\star\in\{\rmno,\rmt,\rmr,\rmrt\}$.
\begin{definition}
\label{def:code}
An $(n,M,K)_\star$-code for a quasi-static MIMO fading channel consists of an encoder $f_\star$ and a decoder $\phi_\star$, specified as follows:
\begin{itemize}
\item If $\star=\rmt$: $f_\rmt:[M]\times[K]\times \bbC^{N_\rma\times N_\rmb}\to\calX^{n\times N_\rma}$, and $\phi_\rmt:\calY_\rmb^{n\times N_\rmb}\times[K]\to [M]$;
\item If $\star=\rmr$: $f_\rmr:[M]\times[K]\to\calX^{n\times N_\rma}$, and $\phi_\rmr:\calY_\rmb^{n\times N_\rmb}\times[K]\times \bbC^{N_\rma\times N_\rmb}\to [M]$;
\item If $\star=\rmrt$: $f_\rmrt:[M]\times[K]\times \bbC^{N_\rma\times N_\rmb}\to\calX^{n\times N_\rma}$, and $\phi_\rmrt:\calY_\rmb^{n\times N_\rmb}\times[K]\times \bbC^{N_\rma\times N_\rmb}\to [M]$;
\item If $\star=\rmno$: $f_\rmno:[M]\times[K]\to\calX^{n\times N_\rma}$, and $\phi_\rmno:\calY_\rmb^{n\times N_\rmb}\times[K]\to [M]$.
\end{itemize}
\end{definition}
For ease of notation, we suppress the key index throughout the paper. To evaluate the performance of an $(n,M)_\star$-code, the maximal error probability and the covertness level are adopted as the reliability and security measures, respectively. Specifically, the maximal error probability is defined as follows:
\begin{align}
\rmP_\rme^{(n)}:=\max_{w\in[M]}\Pr\big\{\hatW\neq W|W=w\big\}.
\end{align}
The covertness analysis is related with error probabilities of BHT done by Willie. In BHT, there are two types of error events: Type--I error (or false alarm) occurs when Willie incorrectly rejects $\rmS_0$ and Type--II error (or missed detection) occurs when Willie incorrectly accepts $\rmS_0$. Let Type--I and Type--II error probabilities be denoted by $\alpha$ and $\beta$, respectively. To ensure covertness, one needs to maximize $\alpha+\beta$. Ideally, one would like $\alpha+\beta=1$, which corresponds to the case where Willie randomly guesses whether Alice is transmitting a message to Bob or not.

Recall that under the state $\rmS_0$ when Alice is not transmitting a message, the channel output at Willie $\bY_\rmw^n$ is the additive noise $\bZ_\rmw^n$ that is generated i.i.d. from $\calCN(\mathbf{0},\bI_{N_\rmw})$. Under the state $\rmS_1$, when Alice is transmitting the message $W$, the channel output at Willie satisfies $\bY_\rmw^n=\bX^n(W)\bh_\rmw+\bZ_\rmw^n$. For ease of notation, let the distribution of Willie's output $\bY_\rmw^n$ under states $(\rmS_1,\rmS_0)$ be denoted by $(Q_1^{(n)},Q_0^{(n)})$, respectively. It follows from \cite[Theorem 13.1.1]{lehmann1986testing} that for optimal tests of BHT, the sum of Type--I and Type--II error probabilities is related with TV distance via $\alpha+\beta=1-\bbV(Q_1^{(n)}, Q_0^{(n)})$. As a result, TV distance has been adopted as a metric for covertness \cite{che2013reliable, wang_covert_2021} due to its relevance to the performance of Willie~\cite[Section II-B]{tahmasbi_first-_2019}. Specifically, when TV distance $\bbV(Q_1^{(n)}, Q_0^{(n)})$ approaches zero, Willie cannot find any test to discriminate between states $\rmS_0$ and $\rmS_1$. However, TV distance is generally difficult to compute, especially for product distributions. Instead, KL divergence is easy to compute and related with TV distance via Pinsker's inequality~\cite[Lemma 11.6.1]{cover1999elements}: $\bbV(Q_1^{(n)}, Q_0^{(n)})\leq \sqrt{\frac{1}{2}\bbD(Q_1^{(n)} \| Q_0^{(n)})}$.
Therefore, consistent with \cite{wang_fundamental_2016, bloch_covert_2016, yu_second_2023}, we adopt KL divergence as the covertness level:
\begin{align}
\label{eq:Pc def kl}
\rmP_\rmc^{(n)}:=\bbD(Q_1^{(n)}\|Q_0^{(n)}).
\end{align}

The theoretical benchmark of covert communication characterizes the tradeoff among the maximal achievable rate, the error probability, and the covertness level as a function of the blocklength. Let $R:=\frac{\log M}{n}$, and fix a CSI scenario $\star\in\{\rmno,\rmt,\rmr,\rmrt\}$. Given a blocklength $n\in\bbN$, an error probability constraint $\varepsilon\in(0,1)$ and a covertness constraint $\delta\in\bbR_+$, the fundamental limit of covert communication under equal power allocation is defined as follows:
\begin{align}
\label{eq2:both constraints}
R^*_\star(n,\varepsilon,\delta):=\sup\big\{R:\exists~ \mathrm{an}~(n,M)_\star\mathrm{-code}~\mathrm{s.t.}~\rmP_\rme^{(n)}\leq\varepsilon~\text{and}~\rmP_\rmc^{(n)}\leq\delta,~\forall~\bh_\rmw\in\calH_\rmw\big\}.
\end{align}
From Definition~\ref{def:code}, it follows that
\begin{align}
R^*_\rmno(n,\varepsilon,\delta)&\le R^*_\rmt(n,\varepsilon,\delta)\le R^*_\rmrt(n,\varepsilon,\delta)\\*
R^*_\rmno(n,\varepsilon,\delta)&\le R^*_\rmr(n,\varepsilon,\delta)\le R^*_\rmrt(n,\varepsilon,\delta).
\end{align}
For almost any DMC and AWGN channel, it follows from \cite{wang_fundamental_2016,bash_limits_2013,bloch_covert_2016} that
\begin{align}
\label{eq:theoretical benchmark awgn}
R^*_\star(n,\varepsilon,\delta)=a\cdot n^{-\frac{1}{2}}+b\cdot n^{-\frac{3}{4}}+O\Big(\frac{\log n}{n}\Big),
\end{align}
where $(a,b)\in\bbR_+^2$ are the coefficients of the first- and second-order terms. In this paper, we generalize the above results to quasi-static MIMO fading channels, characterize the first-order coefficient $a$ explicitly and prove that the second-order coefficient $b=0$.

\section{Main Results and Discussions}
\label{section_main_results}

Recall from Section~\ref{sec:2.1 channel model} that $N_\rma$ is the number of Alice's antennas and $\rank(\bH_\rmb)=N_\rma$ under the full-rank assumption. To present our main results, we use the singular value decomposition (SVD) of the channel matrix $\bH_\rmb$, given by
\begin{align}
\bH_\rmb=\bL_\rmb\bLambda_\rmb\bV_\rmb^\rmH,
\label{eq:svd Hb}
\end{align}
where $\bL_\rmb\in\bbC^{N_\rma\times N_\rma}$ and $\bV_\rmb\in\bbC^{N_\rmb\times N_\rma}$ are unitary. The diagonal matrix containing the positive singular values is defined as $\bLambda_\rmb:=\diag(\sqrt{\Lambda_{\rmb,1}},\ldots,\sqrt{\Lambda_{\rmb,N_\rma}})$.
Recall that we use $\tr(\cdot)$ to denote the trace of a matrix.
Recall that $\blambda_0=\sqrt{\lambda_0}\bI_{N_\rma}$, where $\lambda_0$ is given in \eqref{eq:assume of Hw}. Note that $\bLambda_\rmb$ is a random matrix due to the randomness of the channel matrix $\bH_\rmb$. Given any real number $\varepsilon\in(0,1)$, define the following random variable and constant:
\begin{align}
K(\bLambda_\rmb)
&:=\sqrt{2}\tr\big(\bLambda_\rmb^2\big)\big(\tr\big(\blambda_0^4\big)\big)^{-\frac{1}{2}}=\frac{\sqrt{2}}{\lambda_0 \sqrt{N_\rma}}\sum_{j\in[N_\rma]}\Lambda_{\rmb,j},\label{eq:c1}\\*
\kappa_\varepsilon
&:= \inf\Big\{k\in\bbR: \Pr\big\{K(\bLambda_\rmb)\leq k\big\} \geq \varepsilon\Big\}.\label{eq:def k epsilon}
\end{align}
\begin{theorem} 
\label{theorem 1}
Under equal power allocation, given any $\varepsilon\in(0,1)$, $\delta\in\bbR_+$ and $n\in\bbN$, it follows that for any $\star\in\{\rmno,\rmt,\rmr,\rmrt\}$,
\begin{align}
\label{eq:theorem 1}
R^*_\star(n,\varepsilon,\delta)
= \kappa_\varepsilon\sqrt{\frac{\delta}{n}}+O\Big(\frac{\log n}{n}\Big).
\end{align}
\end{theorem}

The achievability and converse proofs of Theorem \ref{theorem 1} are provided in Sections \ref{section:ach csit} and \ref{section:con csit}, respectively.

We make the following remarks. Firstly, Theorem \ref{theorem 1} characterizes the maximal achievable rate for covert communication over P2P quasi-static MIMO fading channels under equal power allocation, subject to a maximal error probability  $\varepsilon$ and a covertness constraint $\delta$ under the KL divergence covertness metric.
Our results extend the non-covert analysis of Yang \emph{et al.} \cite{yang2014quasi} by imposing the covertness constraint. In particular, with the vanishing power in covert communication, both the analysis of asymptotic behaviors and the conditions of using the Cram\'er--Esseen Theorem \cite[Theorem 15]{yang2014quasi} are significantly different. 
Furthermore, Theorem~\ref{theorem 1} generalizes the theoretical benchmarks in \cite{liu2026covertMIMO} for the non-fading AWGN case to quasi-static fading channels. Compared with \cite{liu2026covertMIMO}, our analyses address challenges due to random fading and cover all four different CSI availability scenarios.


\begin{figure}[tb]
\centering
\includegraphics[width=0.58\columnwidth]{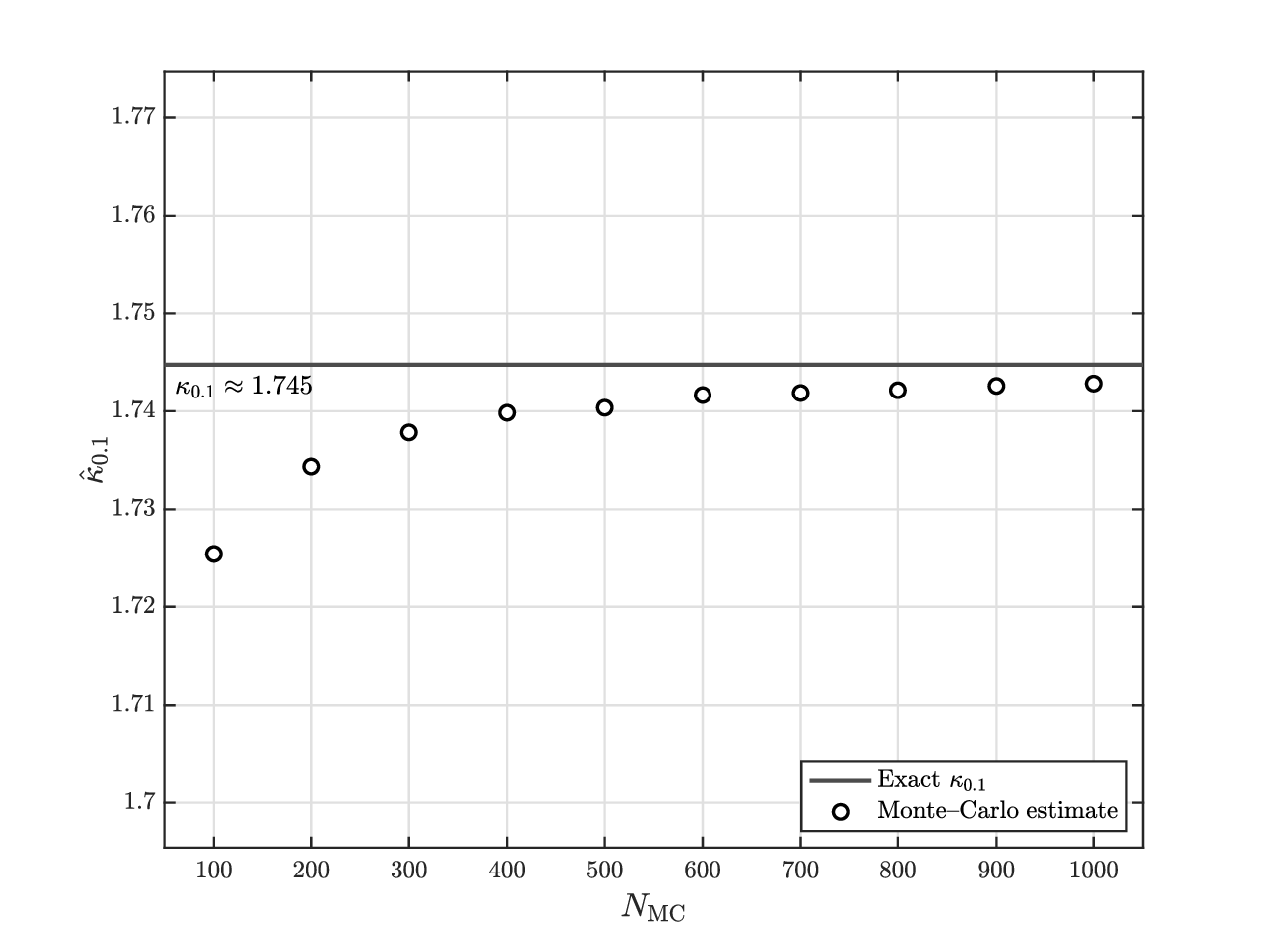}
\caption{Comparison between the exact quantile $\kappa_{0.1}\approx1.745$ and the Monte--Carlo estimates for $N_{\rm MC}=100,200,\ldots,1000$, where each estimate is averaged over $10^5$ independent repetitions. As $N_{\rm MC}$ increases, the Monte--Carlo estimate approaches the exact value, demonstrating improved estimation accuracy with a larger number of samples.}
\label{fig:kappacal}
\end{figure}

Secondly, Theorem \ref{theorem 1} characterizes the first-order asymptotics of $R^*_\star(n,\varepsilon,\delta)$ as $R_1(n,\varepsilon,\delta):=\kappa_\varepsilon\sqrt{\frac{\delta}{n}}$, 
which is consistent with the well-known square root law \cite[Page 1922]{bash_limits_2013}, \cite[Theorem 1]{wang_covert_2021}. With the covertness constraint, the first-order covert rate $\lim_{n\to\infty}\sqrt{n}R^*_\star(n,\varepsilon,\delta)$ equals $\kappa_\varepsilon\sqrt{\delta}$, which reveals the dependence of the theoretical benchmark on the covertness level $\delta$. 
In particular, $\kappa_\varepsilon$ is precisely the lower $\varepsilon$-quantile of $K(\bLambda_\rmb)$. It can be evaluated by inverting the corresponding cumulative distribution function (cdf) when the distribution of $\bH_\rmb$ is known, and estimated via Monte--Carlo trials for general fading models. For a sufficiently large number of simulation trials, the complexity of estimating $\kappa_\varepsilon$ from $N_\mathrm{MC}$ channel samples is $O(N_\mathrm{MC}\log N_\mathrm{MC})$.
The function $K(\bLambda_\rmb)$ depends on the traces of the channel matrices involving the legitimate users $\bH_\rmb$ and the warden $\bh_\rmw$, with the former's effective gain determined by the random matrix $\bLambda_\rmb$, and the latter constrained by the uncertainty set assumed in \eqref{eq:assume of Hw} with a maximal gain of the deterministic matrix $\blambda_0$.
We present two examples to illustrate how $\kappa_\varepsilon$ can be evaluated using the inverse cdf and Monte--Carlo simulation, respectively.
\begin{example}
\label{example:Hb kappa}
Consider a $2\times2$ i.i.d. Rayleigh fading channel whose
entries follow $\calCN(0,1)$. Set $\lambda_0=1$ and $\varepsilon=0.1$. Fix any $(a,b)\in\bbR_+^2$. The pdf for Gamma distribution with parameter $(a,b)$, denoted as $\calG(a,b)$, satisfies that for any $x\ge 0$,
\begin{align}
f_X(x):=\frac{x^{a-1}e^{-x/b}}{\Gamma(a)b^a},
\end{align}
where $\Gamma(x):=\int_{0}^{\infty}t^{x-1}e^{-t}\rmd t$ is the usual gamma function. It follows that $\tr(\bLambda_\rmb^2)=\|\bH_\rmb\|_\rmF^2\sim\calG(4,1)$.
Therefore, numerical evaluation yields
\begin{align}
\kappa_{0.1}=\frac{\sqrt{2}}{1\times \sqrt{2}}F^{-1}_{\calG(4,1)}(0.1)=F^{-1}_{\calG(4,1)}(0.1)\approx 1.745.
\end{align}
\end{example}
\begin{example}
\label{example:kappa mc}
Suppose that the distribution of $\bH_\rmb$ in Example~\ref{example:Hb kappa} is unknown, and $\kappa_\varepsilon$ can be estimated via Monte--Carlo simulation. Specifically, we first generate $N_{\rm MC}$ independent samples $\{K_l\}_{l\in[N_{\rm MC}]}$ of $K(\bLambda_\rmb)$ and sort them in ascending order as
\begin{align}
K_{(1)}\leq K_{(2)}\leq\cdots\leq K_{(N_{\rm MC})}.
\end{align}
The $\varepsilon$-quantile is then estimated by
\begin{align}
\hat{\kappa}_\varepsilon
:=K_{(\lceil N_{\rm MC}\cdot\varepsilon\rceil)}.
\end{align}
For $\varepsilon=0.1$, Fig.~\ref{fig:kappacal} compares the exact quantile $\kappa_{0.1}\approx1.745$ in Example~\ref{example:Hb kappa} with the Monte--Carlo estimates for $N_{\rm MC}=100,200,\ldots,1000$, where each estimate is averaged over $10^5$ independent repetitions. As $N_{\rm MC}$ increases, the Monte--Carlo estimate approaches the exact value, demonstrating improved estimation accuracy with a larger number of samples.
\end{example}
Moreover, consider the following channel matrix to facilitate our discussions below.
\begin{example}
\label{example:rician fading}
Let $\bH_\rmb$ be a $2\times 2$ MIMO Rician fading channel matrix generated according to
\begin{align}
\bH_\rmb = \sqrt{\frac{K}{K+1}}\bH_{\rm LOS} + \sqrt{\frac{1}{K+1}}\bH_{\rm NLOS},
\end{align}
with the Rician factor $K=10$, the deterministic line-of-sight (LOS) component $\bH_{\rm LOS}=\big[\begin{smallmatrix}1 & 1 \\* 1 & 1\end{smallmatrix}\big]$ and the scattered non-line-of-sight (NLOS) component $\bH_{\rm NLOS}$ whose entries are i.i.d. $\calCN(0,1)$.
\end{example}
To illustrate the impact of $\blambda_0$, in Fig. \ref{fig:dif_lambda_0}, we plot the first-order asymptotics $R_1(n,\varepsilon,\delta)$ as a function of the blocklength $n$ for a $2\times 2$ MIMO Rician fading channel with different $\blambda_0$, where $\varepsilon=0.01$, $\delta=0.1$ and $\bH_\rmb$ is as in Example \ref{example:rician fading}.
As observed, $R_1(n,\varepsilon,\delta)$ decreases as $\sqrt{\lambda_0}$ increases, reflecting the inverse relation between the covert transmission rate and Willie's admissible SNR, which implies that a stronger warden necessitates a stricter covertness constraint and inevitably leads to degraded covert performance. Furthermore, each value of $\lambda_0$ specifies a distinct channel uncertainty set in \eqref{eq:assume of Hw} and hence a corresponding well-defined fundamental limit.

\begin{figure}[tb]
\centering
\includegraphics[width=0.58\columnwidth]{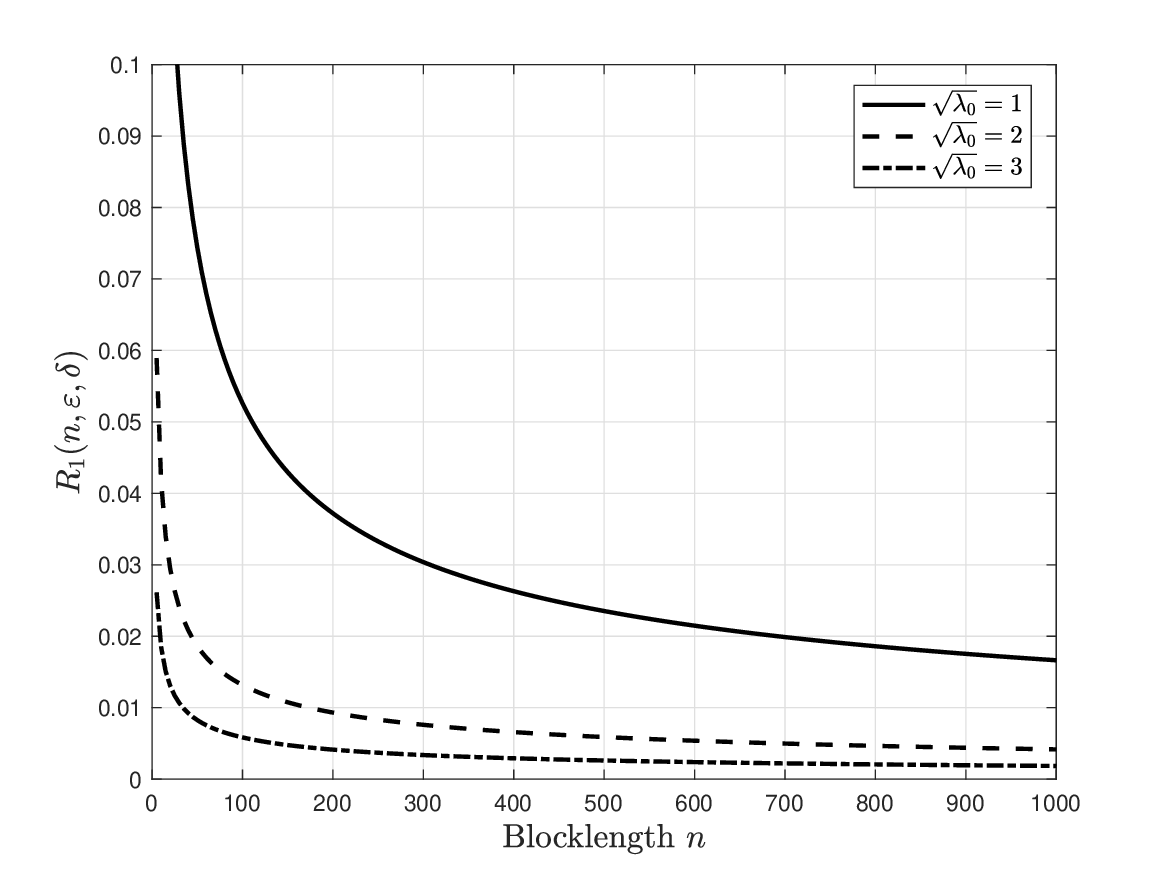}
\caption{Plot of first-order asymptotics $R_1(n,\varepsilon,\delta)$ as a function of blocklength $n$ for a $2\times 2$ MIMO Rician fading channel with different $\blambda_0$, where $\varepsilon=0.01$, $\delta=0.1$ and $\bH_\rmb$ is as in Example \ref{example:rician fading}, based on $10^6$ independent Monte--Carlo trials. As observed, $R_1(n,\varepsilon,\delta)$ decreases as $\sqrt{\lambda_0}$ increases, reflecting the inverse relationship between the covert transmission rate and Willie's admissible SNR. This behavior is also consistent with the square root law.}
\label{fig:dif_lambda_0}
\end{figure}

\begin{figure}[tb]
\centering
\includegraphics[width=0.58\columnwidth]{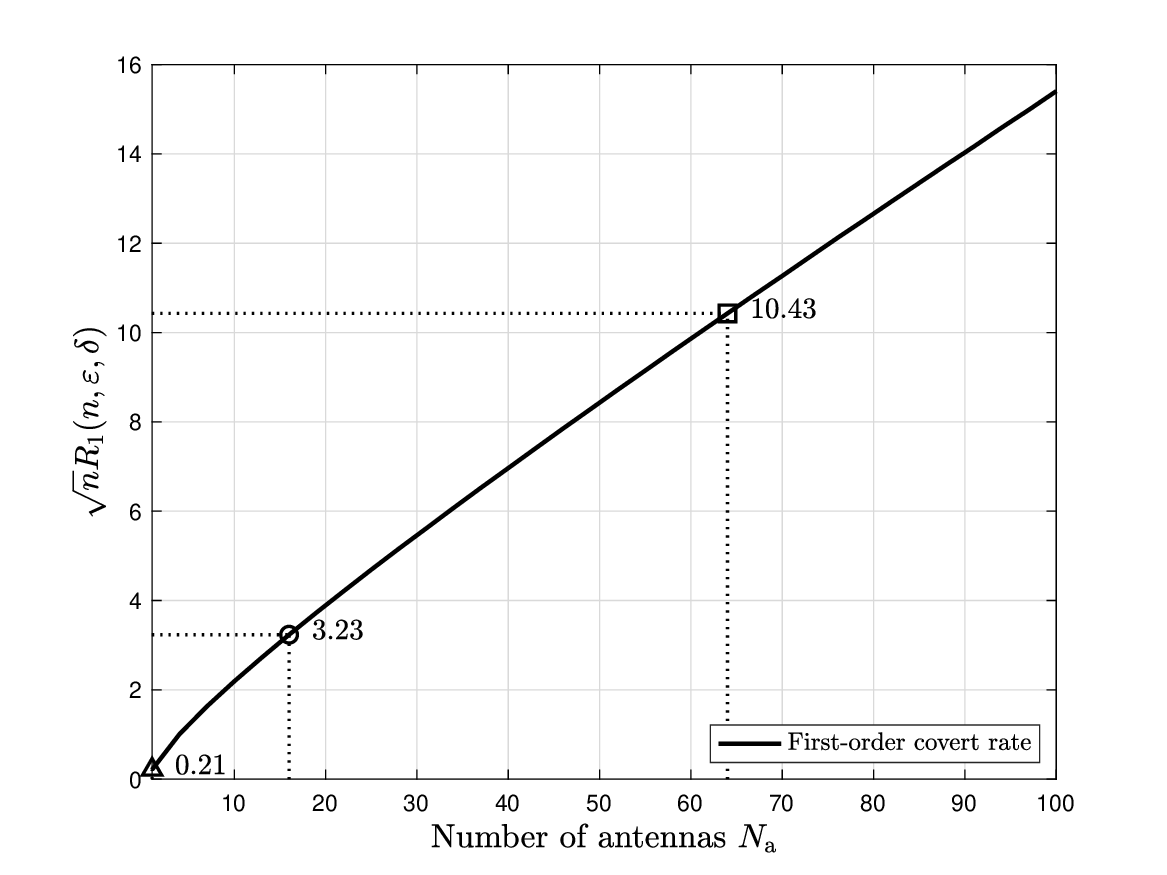}
\caption{Plot of first-order covert rate $\sqrt{n}R_1(n,\varepsilon,\delta)$ as a function of number of antennas $N_\rma$ for MIMO Rician fading channels, where $n=1000$, $\varepsilon=0.01$, $\delta=0.1$, $N_\rmb=N_\rmw=N_\rma$, $\blambda_0=\bI_{N_\rma}$ and $\bH_\rmb$ is as in Example \ref{example:rician fading}, based on $10^6$ independent Monte--Carlo trials. Increasing the number of transmit antennas significantly improves the covert transmission rate. Specifically, when $N_\rma=16$, compared with the single antenna case of $N_\rma=1$, $\sqrt{n}R_1(n,\varepsilon,\delta)$ is improved by a factor of $\frac{3.23}{0.21}\approx 15.4$; when $N_\rma=64$, the improvement factor increases to $\frac{10.43}{0.21}\approx 49.7$.}
\label{fig:dif_m}
\end{figure}

Thirdly, Theorem \ref{theorem 1} reveals the spatial gain of MIMO in improving the covert transmission rate over quasi-static fading channels. To illustrate the impact of the number of antennas, in Fig. \ref{fig:dif_m}, we plot the first-order covert rate $\sqrt{n}R_1(n,\varepsilon,\delta)$ as a function of $N_\rma$ for MIMO Rician fading channels, where $n=1000$, $\varepsilon=0.01$, $\delta=0.1$, $N_\rmb=N_\rmw=N_\rma$, $\blambda_0=\bI_{N_\rma}$ and $\bH_\rmb$ is as in Example \ref{example:rician fading}. When $N_\rma=16$, compared with the single antenna case of $N_\rma=1$, $\sqrt{n}R_1(n,\varepsilon,\delta)$ is improved by a factor of $\frac{3.23}{0.21}\approx 15.4$; when $N_\rma=64$, the improvement factor increases to $\frac{10.43}{0.21}\approx 49.7$. This result highlights the importance of adopting MIMO to improve the covert transmission rate.

\begin{figure}[tb]
\centering
\includegraphics[width=0.58\columnwidth]{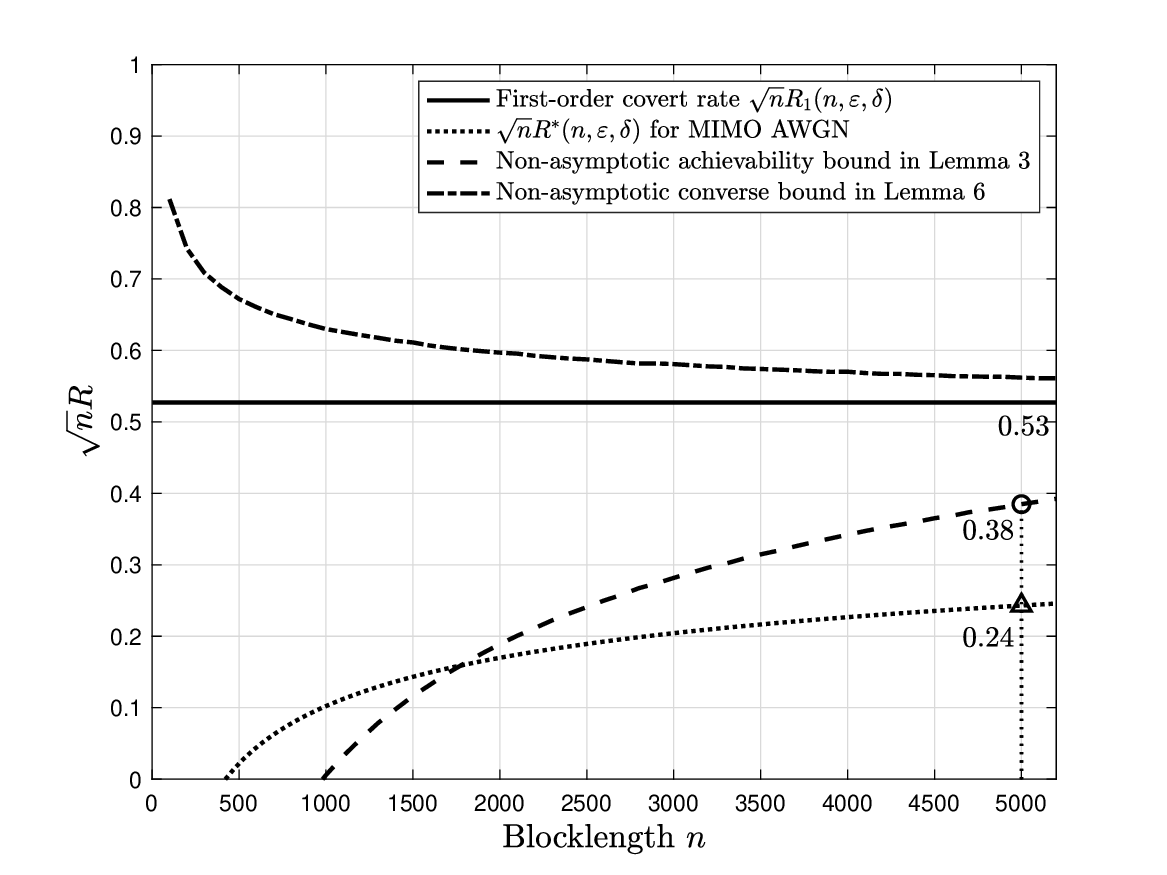}
\caption{Plot of non-asymptotic achievability and converse bounds for $\sqrt{n}R^*_\star(n,\varepsilon,\delta)$ as a function of blocklength $n$ for a $2\times 2$ MIMO Rician fading channel, where $\varepsilon=0.01$, $\delta=0.1$, $\blambda_0=\bI_2$ and $\bH_\rmb$ is as in Example \ref{example:rician fading}, based on $10^7$ independent Monte--Carlo trials. This setting yields the first-order covert rate $\sqrt{n}R_1(n,\varepsilon,\delta) = 0.53$. For comparison, the approximation for $\sqrt{n}R^*(n,\varepsilon,\delta)$ corresponding to a MIMO AWGN channel is also shown, which is reported in \cite[Theorem 1]{liu2026covertMIMO} with $a=0.53$ in \eqref{eq:theoretical benchmark awgn}. As the blocklength grows to $5000$, the achievability bound on $\sqrt{n}R^*_\star(n,\varepsilon,\delta)$ for the quasi-static MIMO fading channel can reach $\frac{0.38}{0.53}\approx 71.7\%$ of $\sqrt{n}R_1(n,\varepsilon,\delta)$, whereas that for the MIMO AWGN channel only achieves $\frac{0.24}{0.53}\approx 45.3\%$. This indicates that $R^*_\star(n,\varepsilon,\delta)$ converges rapidly to $R_1(n,\varepsilon,\delta)$ under quasi-static fading.}
\label{fig:fading_reliability}
\end{figure}

Fourthly, Theorem \ref{theorem 1} shows that the second-order term is zero, which implies that the maximal achievable rate $R^*_\star(n,\varepsilon,\delta)$ converges rapidly to the first-order asymptotics $R_1(n,\varepsilon,\delta)$ as the blocklength $n$ increases.
In non-fading channels, the finite blocklength rate loss is fundamentally driven by the random fluctuations of noise, which is reflected as the $O(n^{-\frac{3}{4}})$ rate penalty in \eqref{eq:theoretical benchmark awgn}.
However, the rate penalty disappears in the quasi-static fading case, which confirms that the reliability behavior of quasi-static fading channels remains dominated by outage even under the covertness constraint.
To illustrate the impact of quasi-static fading, in Fig. \ref{fig:fading_reliability}, we compare the non-asymptotic achievability bound in Lemma \ref{lemma:csit ach non} and the non-asymptotic converse bound in Lemma \ref{lemma:R* con non} for $\sqrt{n}R^*_\star(n,\varepsilon,\delta)$ as a function of blocklength $n$ for $2\times 2$ MIMO Rician fading channels. 
Setting $\varepsilon = 0.01$, $\delta = 0.1$, $\blambda_0 = \bI_2$, and $\bH_\rmb$ as in Example \ref{example:rician fading} yields the first-order covert rate $\sqrt{n}R_1(n,\varepsilon,\delta) = 0.53$. For comparison, Fig. \ref{fig:fading_reliability} also shows the approximation for $\sqrt{n}R^*(n,\varepsilon,\delta)$ corresponding to a MIMO AWGN channel, as reported in \cite[Theorem 1]{liu2026covertMIMO}, with $a=0.53$ in \eqref{eq:theoretical benchmark awgn}. 
For this scenario, as the blocklength grows to $5000$, the achievability bound on $\sqrt{n}R^*_\star(n,\varepsilon,\delta)$ for the quasi-static MIMO fading channel can reach $\frac{0.38}{0.53}\approx 71.7\%$ of $\sqrt{n}R_1(n,\varepsilon,\delta)$, whereas that for the MIMO AWGN channel only achieves $\frac{0.24}{0.53}\approx 45.3\%$. This indicates that $R^*_\star(n,\varepsilon,\delta)$ converges rapidly to $R_1(n,\varepsilon,\delta)$ under quasi-static fading.
Note that the covertness constraint forces the transmit power to decay with blocklength, leading to a fundamentally slower convergence rate compared with the non-covert case.

Furthermore, to illustrate the benefit of CSIT, the corollary below shows that Theorem~\ref{theorem 1} can be improved by optimizing power allocation based on CSIT. Define the following random variable:
\begin{align}
K'(\bLambda_\rmb)&:=\frac{\sqrt{2}}{\lambda_0}
\sqrt{\sum_{j\in[N_\rma]}\Lambda_{\rmb,j}^2}.
\label{eq:K' def}
\end{align}
\begin{corollary}
\label{corollary:CSIT gain}
When CSIT is available, the first-order coefficient
$\kappa_\varepsilon$ under equal power allocation in Theorem~\ref{theorem 1} can be improved to
\begin{align}
\kappa_\varepsilon':=\inf\Big\{k\in\bbR:\Pr\big\{K'(\bLambda_\rmb)\leq k\big\}\geq\varepsilon\Big\}.
\end{align}
\end{corollary}
The proof of Corollary~\ref{corollary:CSIT gain} is provided in Appendix~\ref{appendix:cor1}. In particular, we analyze the achievability bound under CSIT and the converse bound under CSIRT by introducing the power allocation vector and optimizing the transmit power over the effective sub-channels. Similar to the proof of Theorem~\ref{theorem 1}, we consider a random $\bH_\rmb$ and the worst-case channel $\bh_{\rmw,0}$ defined in \eqref{eq:def of hw0}. The randomness of $\bH_\rmb$ can likewise be incorporated into the characterization of the subspace angle distribution under CSIT (cf.~\eqref{eq:def of Tj opt}--\eqref{eq:sin2 pr 1-a opt}).

Compared with equal power allocation in Theorem~\ref{theorem 1}, optimal covert power allocation with CSIT in Corollary~\ref{corollary:CSIT gain} is obtained by maximizing the reliability lower bound subject to the KL divergence covertness constraint, resulting in a power allocation that favors directions with stronger legitimate receiver's channel gains.
Specifically, $\kappa_\varepsilon$ in Theorem~\ref{theorem 1} depends on $\frac{1}{\sqrt{N_\rma}}\sum_{j\in[N_\rma]}\Lambda_{\rmb,j}$, whereas $\kappa_\varepsilon'$ in Corollary~\ref{corollary:CSIT gain} depends on $\sqrt{\sum_{j\in[N_\rma]}\Lambda_{\rmb,j}^2}$. 
By the Cauchy--Schwarz inequality, $K'(\bLambda_\rmb)\ge K(\bLambda_\rmb)$ with equality for a given realization $\bLambda_\rmb=\blambda_\rmb$ if and only if $\lambda_{\rmb,1}=\ldots=\lambda_{\rmb,N_\rma}$, and consequently, $\kappa_\varepsilon'\ge\kappa_\varepsilon$ (cf. \eqref{eq:csit_opt_cs_1}). Furthermore, we would like to emphasize that our power allocation is significantly different from water-filling in \cite{yang2014quasi} without the covertness constraint. Note that water-filling power allocation is designed solely under a total power constraint and concentrates power non-uniformly across transmit directions. With the same total transmit power, water-filling power concentration generally increases the KL divergence between Willie's output distributions under states $\rmS_1$ and $\rmS_0$, potentially violating the covertness constraint~\cite[Section III]{wang_covert_2021}.
The example below illustrates the difference between the three allocations in terms of power allocation results, covertness guarantee under KL divergence and first-order asymptotic covert rate.

\begin{figure}[tb]
\centering
\includegraphics[width=1\columnwidth]{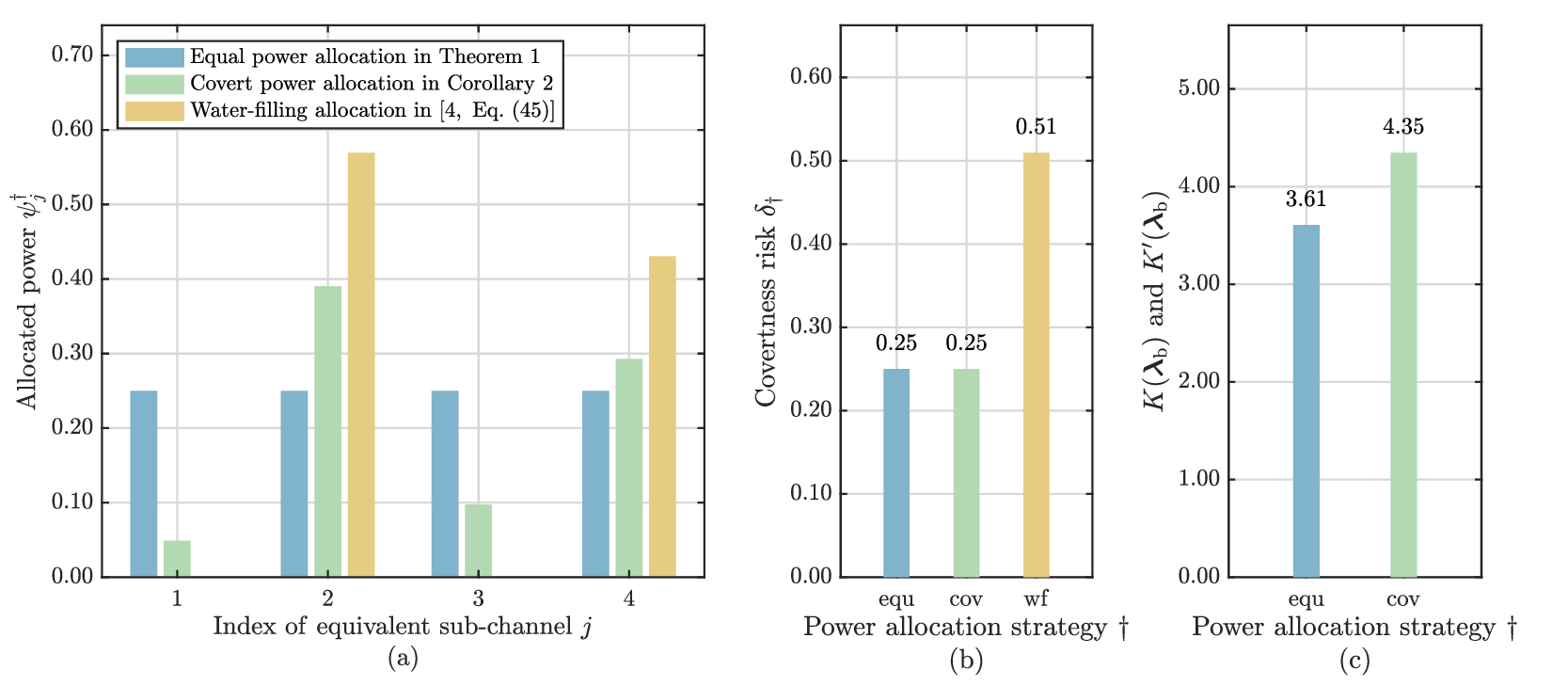}
\caption{Comparison of equal power allocation in Theorem~\ref{theorem 1}, covert power allocation in Corollary~\ref{corollary:CSIT gain} and water-filling power allocation in \cite[Eq.~(45)]{yang2014quasi} for $\bLambda_{\rmb}=(0.3,2.4,0.6,1.8)$ and $\Psi=1$ in Example~\ref{example:wf}. Specifically, subfigure (a) plots the allocated powers across the four equivalent sub-channels where $(\psi^\mathrm{equ}_1,\psi^\mathrm{equ}_2,\psi^\mathrm{equ}_3,\psi^\mathrm{equ}_4)=(0.25,0.25,0.25,0.25)$, $(\psi^\mathrm{cov}_1,\psi^\mathrm{cov}_2,\psi^\mathrm{cov}_3,\psi^\mathrm{cov}_4)=(0.05,0.39,0.10,0.29)$, and $(\psi^\mathrm{wf}_1,\psi^\mathrm{wf}_2,\psi^\mathrm{wf}_3,\psi^\mathrm{wf}_4)=(0,0.57,0,0.43)$; subfigure (b) plots the corresponding covertness risk $\delta_\mathrm{equ}=\delta_\mathrm{cov}=0.25$ and $\delta_\mathrm{wf}\approx 0.51$ under the KL divergence covertness metric in \eqref{eq:Pc def kl}, showing that water-filling increases the covertness risk by approximately 104\% relative to the other two allocations under the same total power constraint; subfigure (c) plots first-order asymptotic covert rates $K(\blambda_\rmb)\approx 3.61$ and $K'(\blambda_\rmb)\approx 4.35$, and demonstrates the benefit of CSIT.}
\label{fig:waterfilling}
\end{figure}

\begin{example}
\label{example:wf}
Consider a deterministic MIMO channel with four equivalent sub-channels, where $\blambda_\rmb=(0.3,2.4,0.6,1.8)$ with a total power constraint $\Psi=1$.
Equal power allocation in Theorem~\ref{theorem 1} gives $(\psi^\mathrm{equ}_1,\psi^\mathrm{equ}_2,\psi^\mathrm{equ}_3,\psi^\mathrm{equ}_4)=(0.25,0.25,0.25,0.25)$, covert power allocation in Corollary~\ref{corollary:CSIT gain} gives $(\psi^\mathrm{cov}_1,\psi^\mathrm{cov}_2,\psi^\mathrm{cov}_3,\psi^\mathrm{cov}_4)=(0.05,0.39,0.10,0.29)$ (cf. \eqref{eq:csit_opt_psi_star}), and water-filling allocation in \cite[Eq.~(45)]{yang2014quasi} gives $(\psi^\mathrm{wf}_1,\psi^\mathrm{wf}_2,\psi^\mathrm{wf}_3,\psi^\mathrm{wf}_4)=(0,0.57,0,0.43)$. For any $\dagger\in\{\mathrm{equ}, \mathrm{cov}, \mathrm{wf}\}$, define $\delta_\dagger:=\sum_{j\in[4]}(\psi_j^\dagger)^2$ as the covertness risk under the KL divergence covertness metric in \eqref{eq:Pc def kl} (cf. \eqref{eq:optimal constraint} and \eqref{eq:KL divergence G 7}). We obtain $\delta_\mathrm{equ}=\delta_\mathrm{cov}=0.25$ and $\delta_\mathrm{wf}\approx 0.51$. Thus, under the same total power constraint, water-filling increases the covertness risk by approximately 104\% relative to the other two allocations. Furthermore, \eqref{eq:c1} and \eqref{eq:K' def} yield $K(\blambda_\rmb)\approx 3.61$ and $K'(\blambda_\rmb)\approx 4.35$, respectively, demonstrating the benefit of CSIT.

For comparison, Fig.~\ref{fig:waterfilling} plots the three power allocation strategies described above. Specifically, subfigure (a) shows the allocated powers across the four equivalent sub-channels; subfigure (b) shows the corresponding covertness risk $\delta_\dagger$ for $\dagger\in\{\mathrm{equ},\mathrm{cov},\mathrm{wf}\}$ in terms of KL divergence metric in \eqref{eq:optimal constraint} and \eqref{eq:KL divergence G 7}; subfigure (c) compared the first-order asymptotic covert rates $K(\blambda_\rmb)$ and $K'(\blambda_\rmb)$ in \eqref{eq:c1} and \eqref{eq:K' def}, respectively.
\end{example}

\section{Achievability Proof for Quasi-Static MIMO Fading Channels}
\label{section:ach csit}

This section presents the achievability proof of Theorem \ref{theorem 1}. Specifically, Section \ref{subsec:ach csit coding} presents a random coding scheme using truncated complex Gaussian distributions, and introduces angle threshold decoding for Bob and BHT detection scheme for Willie; Section \ref{subsec:ach csit covertness} analyzes the covertness constraint and yields an upper bound for the transmit power; Section \ref{subsec:ach csit reliability} analyzes the reliability part and yields a lower bound on the transmission rate; finally, Section \ref{subsec:ach csit final} wraps the above analyses and yields the desired result.

\subsection{Coding Scheme}
\label{subsec:ach csit coding}
\subsubsection{Codebook Generation and Encoding}
\label{subsubsec:ach csit coding codebook_encoding}
We use a truncated complex Gaussian codebook, where each codeword  lies within a thin spherical shell. To define the codeword distribution, let $\rho\in(0,1)$ be a constant and $\Psi(n)=O(\frac{1}{\sqrt{n}})$ be a decreasing function of blocklength $n$, both which are specified later. Fix $r\in\bbR_+$ and let $\calB_0^n(r)$ be the $n$-dimensional sphere with center $\bf0$ and radius $r$, i.e., 
\begin{align}
\calB_0^n(r):=\{x^n\in\bbC^n:\|x^n\|\leq r\}.
\end{align}
Recall that $\calX\subseteq\bbC$.

\begin{definition}
\label{def:distributions}
Define the following pdf $\Pi_{X^n}^\rmG$ of complex Gaussian: for any $x^n\in\calX^n$, 
\begin{align}\label{eq:G}
\Pi_{X^n}^\rmG(x^n):=\frac{\exp(-(\rho\Psi(n))^{-1}\|x^n\|^2)}{(\pi\rho\Psi(n))^{n}}.
\end{align}	
For any $x^n\in\calX^n$, the pdf $\Pi_{X^n}^\rmtG$ for truncated complex Gaussian satisfies that
\begin{align}\label{eq:tG}
\Pi_{X^n}^\rmtG(x^n):=
\begin{cases}
\displaystyle\frac{1}{\Delta_n }\Pi_{X^n}^\rmG(x^n)
&\mathrm{if}~\sqrt{\rho^2n\Psi(n)}\leq\|x^n\|\le\sqrt{n\Psi(n)} \\*
~0
& \mathrm{otherwise},
\end{cases} 
\end{align}
where the normalized coefficient $\Delta_n$ is given by
\begin{align}\label{eq:def Delta}
\Delta_n =\Pr_{\Pi^\rmG_{X^n}}\Big\{X^n\in\big(\calB_0^n(\sqrt{n\Psi(n)})\setminus \calB_0^n(\sqrt{\rho^2n\Psi(n)})\big)\Big\}.
\end{align} 
\end{definition}

Given any $M\in\bbN$, a random codebook $\{\bX^n(w)\}_{w\in[M]}$ is generated such that each codeword $\bX^n(w)\in\calX^{n\times N_\rma}$ is independently drawn from $\Pi_{\bX^n}^{\rmtG}:=\prod_{j\in[N_\rma]} \Pi_{X^n}^{\rmtG}(X^n_j)$. 
As we shall show later, the truncated complex Gaussian codebook enables us to confuse Willie and achieves reliable covert communication. 

\subsubsection{Angle Threshold Decoding}
\label{subsubsec:ach csit coding angle_decoding}

We apply angle threshold decoding, which is a natural approach when CSI $\bH_\rmb$ is not available at the receiver \cite[Section IV-A]{yang2014quasi}. 
As the blocklength $n\to\infty$, the noise is approximately orthogonal to the transmitted codeword, so the received signal lies close to the subspace of the transmitted codeword. Consequently, the angle between the subspace of the transmitted codeword and that of the received signal is significantly smaller than the corresponding angle for any other codeword.
The angle between two subspaces can be characterized by a set of principal angles  \cite[Definition 5]{yang2014quasi}, \cite[Page 289]{absil2006largest}.

Fix $(a,b)\in\bbN^2$. 
Let $\dim(\cdot)$ be the dimension of a subspace, $\calA$ and $\calB$ be subspaces in $\bbC^n$ with dimensions $a:=\dim(\calA)$ and $b:=\dim(\calB)$, where $a \leq b$.
The principal angles between $\calA$ and $\calB$ are defined as
$0\leq\theta_1\leq\ldots\leq\theta_a\leq\frac{\pi}{2}$, where $a$ is determined by the number of linearly independent directions in the subspace of smaller dimension.
It follows from \cite[Page 289]{absil2006largest} that for $j\in[a]$,
\begin{align}\label{7-decoder siso long}
\cos\theta_j:=\max_{\substack{
\ba\in\calA,\bb\in\calB:~\|\ba\|=\|\bb\|=1,\\
\ba^\rmH\ba_{j'}=\bb^\rmH\bb_{j'}=0,~j'\in[j-1]
}}
\ba^\rmH\bb.
\end{align}
It follows from \cite[Eq.~(40)]{yang2014quasi} that the sine value of the angle between subspaces $\calA$ and $\calB$ is defined as
\begin{align}\label{7 1:sin def}
\sin(\calA,\calB):=\prod_{j\in[a]}\sin\theta_j.
\end{align}
To illustrate the concept of principal angles, Fig. \ref{fig:principal angles r2} plots the principal angles between two $2$-dimensional subspaces $\calA$ and $\calB$ in $\bbR^3$, where the black vectors form the first principal angle $\theta_1=0$ due to the nontrivial intersection of the two subspaces in $\bbR^3$, and the red vectors form the second principal angle $\theta_2>0$ and are orthogonal to the black vectors.
Figuratively, the principal angles represent the angles between the closest vectors of the two subspaces along each direction. Thus, \eqref{7 1:sin def} can be used to measure the similarity between the subspaces across all directions, and the smaller $\sin(\cdot)$ is, the closer the two subspaces are.
Let the columns of $\tilde{\ba}$ and $\tilde{\bb}$ be orthonormal bases of $\calA$ and $\calB$, respectively. The cosine values of the principal angles $\{\cos\theta_j\}_{j\in[a]}$ defined in \eqref{7-decoder siso long} are equal to the singular values of $\tilde{\ba}^\rmH \tilde{\bb}$ \cite[Page 289]{absil2006largest}.
\begin{figure}[tb]
\centering
\includegraphics[width=0.47\columnwidth]{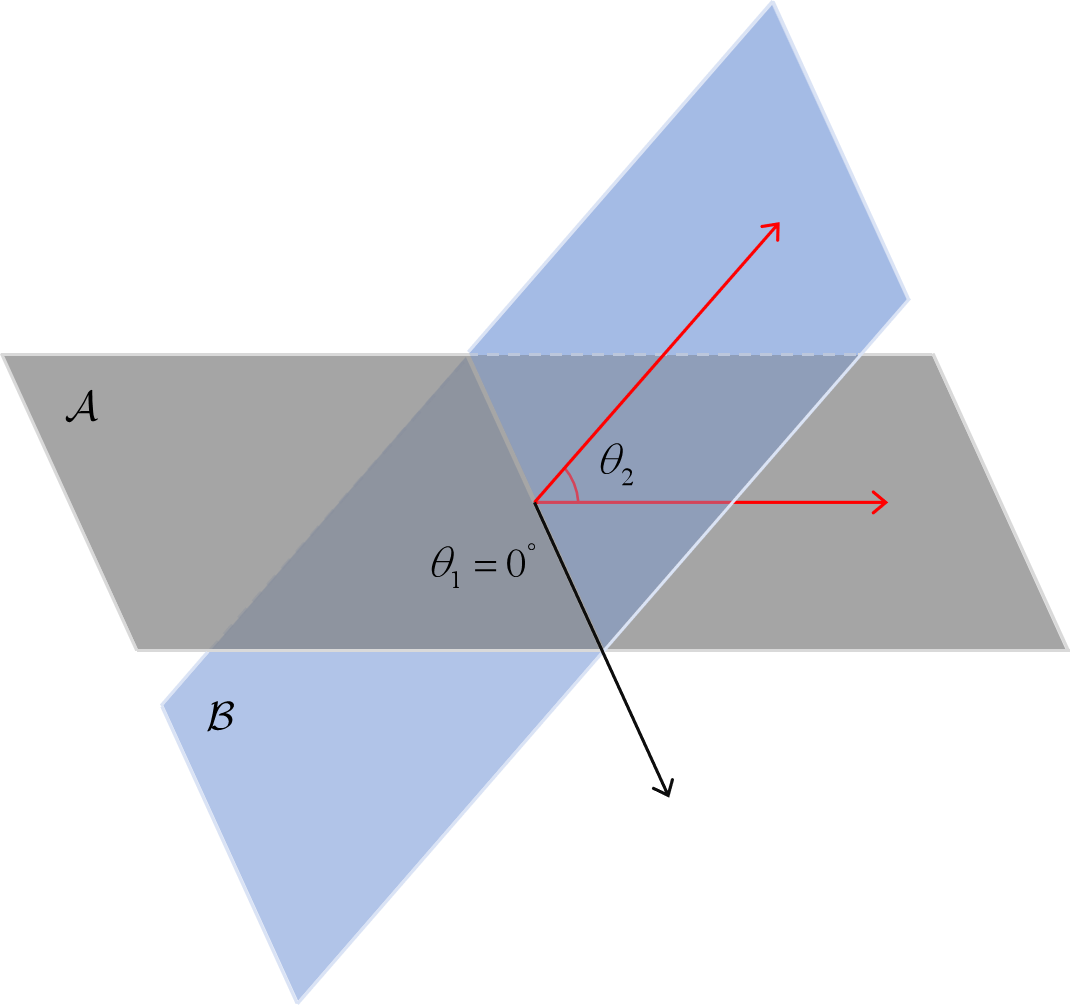}
\caption{Plot of principal angles between two $2$-dimensional subspaces $\calA$ and $\calB$ in $\bbR^3$. The black vectors form the first principal angle $\theta_1=0$ due to the nontrivial intersection of the two subspaces in $\bbR^3$. The red vectors form the second principal angle $\theta_2>0$ and are orthogonal to the black vectors.}
\label{fig:principal angles r2}
\end{figure}

Recall from \eqref{eq:yb channel} that Bob receives the noisy channel output $\bY_\rmb^n=\bX^n(W)\bH_\rmb +\bZ_\rmb^n$. Recall that, for a matrix, $\spn(\cdot)$ denotes the subspace spanned by its column vectors. Given ${\bY}_\rmb^n$, Bob applies the sequential angle threshold decoding to estimate the message as $\hatW$. Specifically, there is an associated threshold $\gamma\in\bbR_+$ for codewords, and the decoder computes the angle between the subspaces, and selects the first codeword that satisfies $\sin^2(\spn(\bX^n(w)),\spn(\bY_\rmb^n))\leq\gamma$.
It yields a tight achievability bound at finite blocklength if the threshold $\gamma$ is chosen appropriately \cite[Section IV]{yang2014quasi}.

Recall that random variables are in capital and their realizations are in lower case. The example below illustrates the principal angles and demonstrates the effectiveness of angle threshold decoding.
\begin{figure}[tb]
\centering
\includegraphics[width=0.58\columnwidth]{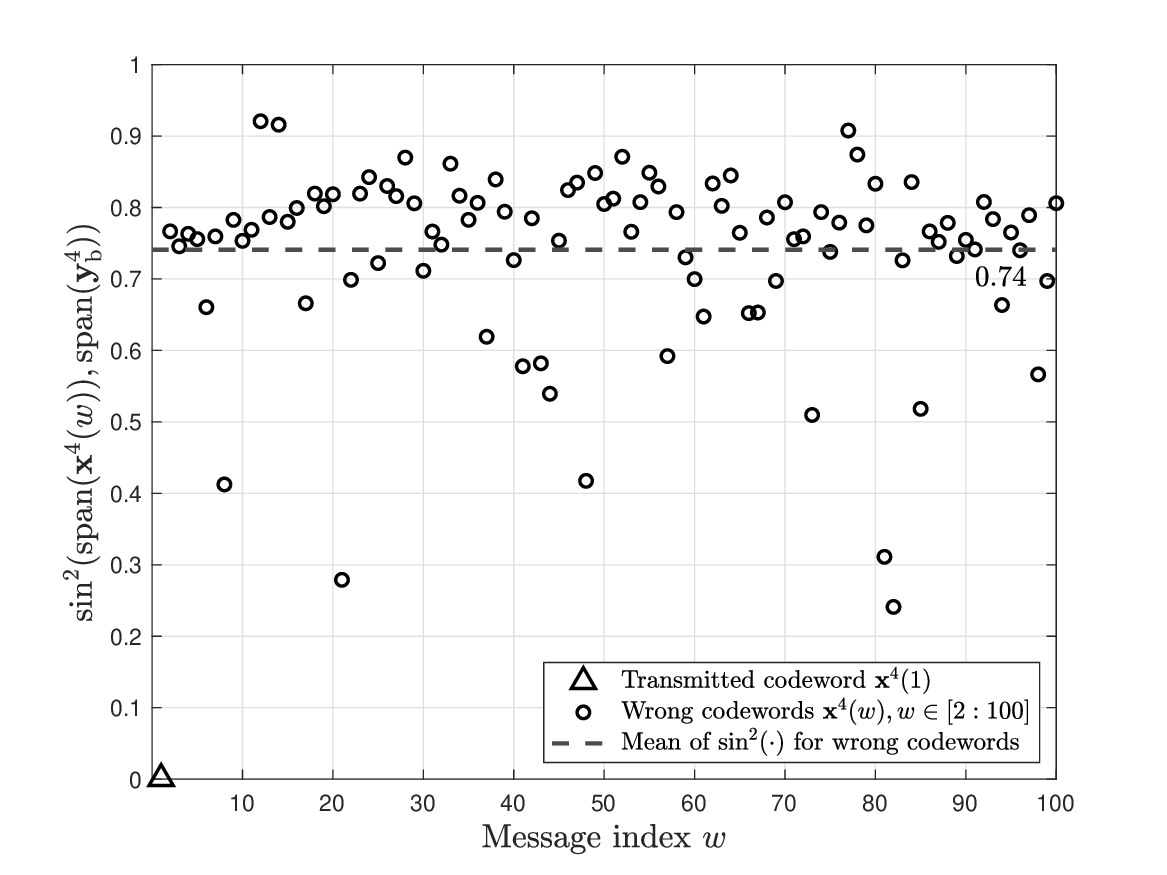}
\caption{Plot of $\sin^2(\spn(\bx^4(w)), \spn(\by_\rmb^4))$ for the transmitted codeword $\bx^4(1)$ and other candidate codewords ${\bx^4(w)}_{w\in[2:100]}$, whose entries are i.i.d. $\calCN(0,2.125)$. The transmitted codeword, shown as a triangle marker, achieves a significantly smaller value $1\times 10^{-3}$ in~\eqref{eq:sin2 theta1 atd} than the wrong codewords, which are shown as circles and whose mean $\sin^2(\cdot)$ is approximately $0.74$. This demonstrates the effectiveness of angle threshold decoding.}
\label{fig: decoding example}
\end{figure}
\begin{example}[Angle Threshold Decoding]
Set $n=4$, $N_\rma=2$. Let the transmitted codeword $\bx^4(1)\in\bbC^{4\times 2}$ and the channel matrix $\bh_\rmb\in\bbC^{2\times 2}$ be
\begin{align}
\bx^4(1) =
\begin{bmatrix}
0 & 2 \\* 1-2\imath & 1 \\* 2 & 1-\imath \\* 0 & -1
\end{bmatrix}, \quad
\bh_\rmb =
\begin{bmatrix}
0.9+0.2\imath & 0.3-0.1\imath \\* -0.2+0.1\imath & 0.8-0.2\imath
\end{bmatrix},\nn
\end{align}
respectively. Given a particular realization $\bz_\rmb\in\bbC^{4\times 2}$ with i.i.d. $\calCN(0,1)$ entries, the channel output $\by_\rmb\in\bbC^{4\times 2}$ is
\begin{align}
\by_\rmb^4=
\begin{bmatrix}
-0.35 - 0.27\imath & 2.39 - 0.19\imath \\* 1.77 - 1.66\imath & 1.13 - 1.30\imath \\* 1.46 + 0.11\imath & 0.91 - 0.82\imath \\* 0.23 - 0.19\imath & -1.25 + 0.53\imath
\end{bmatrix}.\nn
\end{align}
To obtain the principal angles, the basis matrices $\bx^4(1)$ and $\by_\rmb^4$ are orthonormalized as
\begin{align}
\tilde{\bx}^4 =
\begin{bmatrix}
0 & -0.76 \\* -0.33+0.67\imath & -0.25-0.25\imath \\* -0.67 & -0.13+0.38\imath \\* 0 & 0.38
\end{bmatrix}, \quad
\tilde{\by}^4_\rmb =
\begin{bmatrix}
-0.12 - 0.09\imath & -0.83 + 0.03\imath \\* 0.61 - 0.58\imath & -0.15 + 0.07\imath \\* 0.51 + 0.04\imath & -0.05 + 0.22\imath \\* 0.08 - 0.07\imath & 0.43 - 0.21\imath
\end{bmatrix},\nn
\end{align}
respectively.
The projection matrix corresponding to the two orthonormal bases is given by
\begin{align}
(\tilde{\bx}^4)^\rmH\tilde{\by}^4_\rmb =
\begin{bmatrix}
-0.93 - 0.24\imath & 0.13 - 0.07\imath \\* 0.06 + 0.15\imath & 0.90 - 0.17\imath
\end{bmatrix}. \nn
\end{align}
The SVD of $(\tilde{\bx}^4)^\rmH\tilde{\by}^4_\rmb$ yields singular values $0.998$ and $0.898$, which are cosine values of principal angles. Therefore, the subspaces $\spn(\bx^4(1))$ and $\spn(\by^4_\rmb)$ have principal angles $\theta_1\approx3.97^\circ$ and  $\theta_2\approx26.17^\circ$, and
\begin{align}
\label{eq:sin2 theta1 atd}
\sin^2(\spn(\bx^4(1)),\spn(\by^4_\rmb))=\sin \theta_1 \times \sin \theta_2 \approx 1\times 10^{-3}.
\end{align}
For comparison, Fig.~\ref{fig: decoding example} plots
$\sin^2(\spn(\bx^4(w)), \spn(\by_\rmb^4))$
for candidate codewords $\{\bx^4(w)\}_{w\in[2:100]}$, whose entries are i.i.d. as $\calCN(0,2.125)$.
The transmitted codeword $\bx^4(1)$, represented by a triangle marker, achieves a significantly smaller value $1\times 10^{-3}$ in~\eqref{eq:sin2 theta1 atd} than the wrong codewords, which are represented by circles, and the mean of $\sin^2(\cdot)$ for wrong codewords is approximately $0.74$. This demonstrates the effectiveness of angle threshold decoding.
\end{example}

\subsubsection{Covertness Metric}
\label{subsubsec:ach csit coding bht}
Recall from \eqref{eq:yw channel} that Willie receives the noisy channel output $\bY_\rmw^n=\bX^n(W)\bh_\rmw +\bZ_\rmw^n$. Since the warden knows $\bh_\rmw$ and $\rank(\bh_\rmw)=N_\rma$, Willie can post-process the observation $\bY_\rmw^n$ to obtain ${\bY'}_\rmw^n\in\calY_\rmw^{n\times N_\rma}$ without loss of detection performance. With a slight abuse of notation, we reuse ${\bY}_\rmw^n$ to represent ${\bY'}_\rmw^n$.
Given $\bY_\rmw^n$, Willie uses BHT to determine whether Alice is communicating with Bob. By applying the Neyman--Pearson Lemma \cite[Appendix B]{polyanskiy_channel_2010}, the optimal approach for Willie to minimize the detection error probability is the likelihood ratio test, whose performance is characterized by the sum of Type--I and Type--II error probabilities. Recall from Section \ref{sec:II-B} that we adopt KL divergence as the covertness metric, as it is closely related to Willie’s error probabilities. Let the distributions of Willie's output $\bY_\rmw^n$ induced under states $(\rmS_1,\rmS_0)$ be denoted by $(Q_{\bY_\rmw^n}^\rmtG,Q_{\bZ_\rmw^n})$, which correspond to $(Q_1^{(n)},Q_0^{(n)})$ defined in Section \ref{sec:II-B}, respectively. In this case, the covertness constraint in \eqref{eq2:both constraints} is equivalent to
\begin{align}
\label{eq4:covertness constraint}
\bbD(Q^\rmtG_{\bY^n_\rmw}\| Q_{\bZ_\rmw^n})\leq \delta.
\end{align}

\begin{table}[tb]
\centering
\scriptsize
\caption{Commonly Used Notation}
\label{tab:notations}
\renewcommand{\arraystretch}{1.2}
\begin{tabular}{cc}
\hline
Symbols & Meaning\\*
\hline
$N_\rma$ & The number of Alice's antennas \\*
$\rho$ & A truncation parameter with value in $(0,1)$ \\*
$\Psi(n)$ & A power scaling function of order $O(\frac{1}{\sqrt{n}})$ \\*
$\bH_\rmb$ & Channel matrix from Alice to Bob with $\rank(\bH_\rmb)=N_\rma$ \\*
$\bh_\rmw$ & Channel matrix from Alice to Willie with $\rank(\bh_\rmw)=N_\rma$ \\*
$\bLambda_\rmb$ & Singular values of $\bH_\rmb$ with $\bLambda_\rmb=\diag(\{\sqrt{\Lambda_{\rmb,j}}\}_{j\in[N_\rma]})$ \\*
$\blambda_0$ & The upper bound on channel gains of $\bh_\rmw$ with $\blambda_0=\sqrt{\lambda_0}\bI_{N_\rma}$ \\*
$\Pi_{\bX^n}^\rmG$ & Complex Gaussian distribution with \eqref{eq:G} \\*
$\Pi_{\bX^n}^\rmtG$ & Truncated complex Gaussian distribution with \eqref{eq:tG} \\*
$Q_{\bY_\rmw^n}^\rmtG$ & Channel output distribution induced by $\Pi_{\bX^n}^{\rmtG}$ at Willie \\*
$Q_{\bZ_\rmw^n}$ & Noise distribution at Willie \\*
\hline
\end{tabular}
\end{table}

For ease of subsequent analyses, Table \ref{tab:notations} summarizes the symbols used throughout the paper.

\subsection{Covertness Analysis}
\label{subsec:ach csit covertness}

This section shows our coding scheme satisfies the covertness constraint when the power of the truncated Gaussian codebook is set properly.
Recall from Table \ref{tab:notations} that $\blambda_0$ is the upper bound on channel gains of $\bh_\rmw$ with $\blambda_0=\sqrt{\lambda_0}\bI_{N_\rma}$, and $\rho$ is a truncation parameter.
Define the sequence $\{\nu(n)\}_{n\in\bbN}$ such that for each $n\in\bbN$, $\nu(n)>1$ and $\lim_{n\to\infty}\nu(n)=1$. Fix any $\rho\in(0,1)$.
\begin{lemma}\label{lemma:covert}
For any $n\in\bbN$, when the power is set as
\begin{align}
\label{eq4:power level}
\Psi(n)=\sqrt{\frac{2\delta}{n\rho^2\nu^2(n)\tr\big(\blambda_0^4\big)}},
\end{align}
using the truncated complex Gaussian codebook where each codeword is generated from $\Pi_{\bX^n}^\rmtG$ in \eqref{eq:tG}, the covertness constraint \eqref{eq4:covertness constraint} is satisfied.
\end{lemma}

The proof of Lemma \ref{lemma:covert} is provided in Appendix \ref{appendix:covert ach}. The proof uses the notion of quasi-$\eta$-neighborhood \cite[Definition 4]{yu_second_2023}, which concerns the local differential geometry of the noise distribution under the KL divergence covertness constraint. Specifically, we show that the output distribution $Q_{\bY_\rmw^n}^\rmtG$, which corresponds to the input distribution with power $\Psi(n)=O(\frac{1}{\sqrt{n}})$, is in the quasi-$\eta$-neighborhood of the noise distribution $Q_{\bZ_\rmw^n}$.

Lemma \ref{lemma:covert} shows that the power level $\Psi(n) = O(\frac{1}{\sqrt{n}})$ required to satisfy the covert constraint follows the classical square root law \cite[Page 1922]{bash_limits_2013}, \cite[Theorem 1]{wang_covert_2021}. Using such a transmit power confuses the warden Willie, who cannot reliably detect whether Alice is transmitting meaningful information.

Recall that $f_\star$ with $\star\in\{\rmno,\rmt,\rmr,\rmrt\}$ represents the encoder associated with a specific CSI $\bH_\rmb$ setting between the legitimate users.
Under the worst-case assumption in Section~\ref{sec:2.1 channel model}, given a blocklength $n\in\bbN$, an error probability constraint $\varepsilon\in(0,1)$ and a maximal power constraint $\Psi(n)=O(\frac{1}{\sqrt{n}})$, the maximal achievable rate for quasi-static MIMO fading channels is defined as
\begin{align}
\label{eq:def M max}
R_{\star,\max}^*(n,\varepsilon,\Psi(n)):=\sup\big\{R:\exists~ \mathrm{an}~(n,M)_\star\mathrm{-code}~\mathrm{s.t.}~\rmP_\rme^{(n)}\leq\varepsilon~\text{and}~\max_{w\in[M]}\|f_\star(w)\|_\rmF^2\leq n\Psi(n),~\forall~\bh_\rmw\in\calH_\rmw\big\}.
\end{align}
As shown in Lemma \ref{lemma:covert}, the truncated complex Gaussian distribution with $\Psi(n)$ in \eqref{eq4:power level} ensures a covert codebook. It follows from the definition of the theoretical benchmark for covert communication in \eqref{eq2:both constraints} that
\begin{align}
\label{eq:achM}
R^*_\star(n,\varepsilon,\delta)\geq R_{\star,\max}^*(n,\varepsilon,\Psi(n)).
\end{align}
Thus, the theoretical benchmark of covert communication is lower bounded by the theoretical benchmark of communication without the covertness constraint but with a maximal power constraint. This connection helps us establish the reliability analysis in the following subsection.

\subsection{Reliability Analysis}
\label{subsec:ach csit reliability}

This section presents the reliability analysis under a vanishing power constraint $\Psi(n)=O(\frac{1}{\sqrt{n}})$. We begin with a non-asymptotic bound and subsequently derive the asymptotic result. In particular, we consider the no-CSI case where neither the transmitter Alice nor the receiver Bob knows $\bH_\rmb$.

\subsubsection{Non-Asymptotic Bound}

Fix a real-valued continuous alphabet $\calX\in\bbR$, two positive real numbers $(a,b)\in\bbR_+^2$. Recall that $\Gamma(\cdot)$ is the usual gamma function. The pdf for Beta distribution with parameter $(a,b)$, denoted as $\calB(a,b)$, satisfies that for any $x\in\calX$,
\begin{align}\label{eq:def of 1beta}
f_X(x):=
\begin{cases}
\displaystyle \frac{\Gamma(a+b)}{\Gamma(a)\Gamma(b)}x^{a-1}(1-x)^{b-1}
&\mathrm{if}~x\in(0,1) \\*
~0
& \mathrm{otherwise}.
\end{cases} 
\end{align}
Recall that $(N_\rma,N_\rmb)$ are the numbers of antennas at Alice and Bob, respectively, and $\{\sqrt{\Lambda_{\rmb,j}}\}_{j\in[N_\rma]}$ are singular values of $\bH_\rmb$. Let $Z_{\rmb,ij}$ be the $(i,j)$-th entry of noise matrix $\bZ_\rmb^n$. For each $j\in[N_\rma]$, define the following random variable
\begin{align}
\label{eq:def of Tj}
T_j:=\frac{\sum_{i\in[2,n]}\big|Z_{\rmb,ij}\big|^2}{\big|\sqrt{n\rho\Psi(n)\Lambda_{\rmb,j}}+Z_{\rmb,1j}\big|^2+\sum_{i\in[2,n]}\big|Z_{\rmb,ij}\big|^2}.
\end{align}

\begin{lemma}
\label{lemma:csit ach non}
Fix any $\varepsilon\in(0,1)$ and $\tau\in(0,\varepsilon)$. There exists an $(n,M)_\rmno$-code with the maximal error probability constraint $\varepsilon$ such that
\begin{align}
\label{eq:logM/n non ach}
\frac{\log M}{n}\geq\frac{1}{n}\log\frac{\tau}{\Pr\Big\{\prod_{j\in[N_\rmb]}B_j\leq\gamma_n\Big\}},
\end{align}
where $\{B_j\}_{j\in[N_\rmb]}$ are i.i.d. with $B_j\sim\calB(n-N_\rma-j+1,N_\rma)$, and $\gamma_n$ is chosen so that
\begin{align}
\label{eq:sin2 pr 1-a}
\Pr\bigg\{\prod_{j\in[N_\rma]}T_j \leq {\gamma _n}\bigg\}\geq 1-\varepsilon+\tau.
\end{align}
\end{lemma}
The proof of Lemma \ref{lemma:csit ach non} is provided in Appendix \ref{appendix:csit ach non}.
To analyze the performance of the angle threshold decoding, we apply \cite[Corollary 1]{yu_finite_2021} on a physically degraded version of the channel \eqref{eq:yb channel}. The randomness of $\bH_\rmb$ is incorporated directly into the decoding error probability. By using angle threshold decoding, the relevant subspace angle statistics are characterized through the Beta distribution and the resulting error probability is then related to the covert outage probability.

Although our non-asymptotic achievability bound resembles \cite[Theorems 1, 4]{yang2014quasi}, it differs significantly in both numerical behavior and proof techniques.
Without the covertness constraint, \cite[Theorems 1, 4]{yang2014quasi} employed the $\kappa\beta$ bound \cite[Theorem 25]{polyanskiy_channel_2010}, where the codebook is deterministic and the transmit power is non-vanishing.
Specifically, for the CSIT case, \cite[Theorem 1]{yang2014quasi} applied precoding together with the water-filling power allocation strategy to obtain a tight non-asymptotic achievable bound.
When CSI is not available at the transmitter, \cite[Theorem 4]{yang2014quasi} assumed an input covariance matrix $\rvQ^*$ in \cite[Eq. (64)]{yang2014quasi} achieving the infimum of the outage probability to support the maximal achievable rate, and the existence of $\rvQ^*$ is proved in \cite[Proposition 5]{yang2014quasi}.

However, under the covertness constraint, we employ a truncated complex Gaussian codebook with \emph{random} codebooks and a \emph{vanishing} transmit power.
Since random coding is used, the $\kappa\beta$ bound is no longer applicable. Instead, we derive a new non-asymptotic bound inspired by \cite[Theorem 1]{yu_finite_2021}.
As discussed in Section \ref{section_main_results}, the water-filling power allocation strategy cannot be adopted due to the covertness constraint, which is replaced by an equal power allocation strategy in our proof.
Subsequently, in Appendix \ref{subapx:csit ach non beta}, we prove in detail the association between the Beta distribution and the squared sine value of the angle between the codeword-spanned and output-spanned subspaces, which leads to \eqref{eq:logM/n non ach}. Building on this, in Appendix \ref{subapx:csit ach non sin2}, we derive the probability bound for correct angle threshold decoding in the no-CSI case to obtain \eqref{eq:sin2 pr 1-a}.

\subsubsection{Asymptotic Analysis}

Given $\Psi(n)=O(\frac{1}{\sqrt{n}})$, define the covert outage rate with $\Psi(n)$ as
\begin{align}
\label{eq:ach C final def}
C_{\varepsilon}(\Psi(n))
:=\sup\Big\{ R:\Pr\Big\{ \log\det\big(\bI_{N_\rmb}+\rho\Psi(n)\bH_\rmb^\rmH\bH_\rmb\big)<R \Big\}\leq\varepsilon \Big\}.
\end{align}

\begin{lemma}
\label{lemma:csit ach reliability}
For quasi-static MIMO fading channels with power function $\Psi(n)=O(\frac{1}{\sqrt{n}})$ and $\varepsilon\in(0,1)$, it follows that
\begin{align}
\label{eq:ach reliability result}
R_{\rmno,\max}^*(n,\varepsilon,\Psi(n))\geq C_{\varepsilon}(\Psi(n))+O\Big(\frac{\log n}{n}\Big).
\end{align}
\end{lemma}
The proof of Lemma \ref{lemma:csit ach reliability} is provided in Appendix \ref{appendix:csit ach asymptotic}. We start by analyzing the denominator on the right hand side of \eqref{eq:logM/n non ach} using the asymptotic characteristics of the Beta distribution. Subsequently, we carefully analyze the effect of $\bH_\rmb$ in Lemma~\ref{lemma:18ach} and use it to verify the reliability of the chosen decoding threshold, whose proof is deferred to Appendix \ref{appendix:lemma:18ach}.

Lemma \ref{lemma:csit ach reliability} establishes from the achievability that the second-order term in quasi-static fading channels remains zero under the covertness constraint, indicating that the outage event dominates the finite blocklength performance.
Compared with the MIMO AWGN case \cite[Lemma 3]{liu2026covertMIMO}, the second-order term in \eqref{eq:theoretical benchmark awgn} is missing in Lemma \ref{lemma:csit ach reliability}.
In the non-asymptotic regime, we exploit the outage characteristics of quasi-static fading channels, and transform the decoding error probability into a Beta distribution via angle threshold decoding. Lemma \ref{lemma:csit ach reliability} analyzes the asymptotic behavior of the Beta distribution, which replaces the role of the Berry--Esseen Theorem \cite[Theorem 44]{polyanskiy_channel_2010}.
This approach directly eliminates the influence of the second-order term arising from the randomness of the noise.
Due to the vanishing transmit power imposed by the covertness constraint, we relax the constraint on $\bH_\rmb$ in \cite[Eqs. (81)--(83)]{yang2014quasi} to require only that it is independent of $n$, which is naturally satisfied in quasi-static fading channels.

To determine an appropriate threshold $\gamma_n$ and establish \eqref{eq:ach reliability result}, we relate the left hand side of \eqref{eq:sin2 pr 1-a} to the covert outage probability.
We introduce two new random variables $K_1$ and $\tilG_1$ in \eqref{eq:def K ach} and \eqref{eq:tilG1} to connect $T_j$ with the covert outage probability, and analyze their asymptotic behavior using Lemma \ref{lemma:B sqrtn A ff'}, where the upper bounds of the pdf $f_{\tilG_1}$ of $\tilG_1$, and its derivative $f'_{\tilG_1}$ are essential. With the vanishing power $\Psi(n)=O(\frac{1}{\sqrt{n}})$, $f_{\tilG_1}=O(n^\frac{1}{4})$ and $f'_{\tilG_1}=O(n^\frac{1}{2})$ (cf. Appendices \ref{subapx:f tilG 1} and \ref{subapx:f tilG 1 prime}), which is significantly different from $f_{\tilG_1}=f'_{\tilG_1}=O(1)$ in \cite[Page 4255]{yang2014quasi}.
This change renders \cite[Lemma 18]{yang2014quasi} invalid, leading to an unexpected $O(\frac{1}{\sqrt{n}})$ term, which is comparable to the square root law in covert communication and makes it impossible to even derive the first-order covert rate.
To handle this critical issue, we exploit the characteristics of covert communication to show that the derivative of the covert outage probability $F'_\rmout$ at the covert outage rate $C_\varepsilon(\Psi(n))$ is $O(\sqrt{n})$ under vanishing power $\Psi(n)=O(\frac{1}{\sqrt{n}})$ (cf. Appendix \ref{subapx:F' calculation}), which contrasts with $F'_\rmout(C_\varepsilon)>0$ assumed in \cite[Page 4253]{yang2014quasi}. This precisely compensates for the inflated remainder term.

\subsection{Final Steps}
\label{subsec:ach csit final}

Using the covertness analysis in Lemma \ref{lemma:covert} and the reliability analysis in Lemma \ref{lemma:csit ach reliability}, we lower bound $R^*_\rmno(n,\varepsilon,\delta)$ in \eqref{eq2:both constraints} as
\begin{align}
R^*_\rmno(n, \varepsilon, \delta)
&\geq R_{\rmno,\max}^*(n,\varepsilon,\Psi(n))\label{eq:ach final 1C-1}\\*
&\geq C_{\varepsilon}(\Psi(n))+O\Big(\frac{\log n}{n}\Big),\label{eq:ach final 1C}
\end{align}
where \eqref{eq:ach final 1C-1} follows from \eqref{eq:achM}, and \eqref{eq:ach final 1C} follows from \eqref{eq:ach reliability result}.

Recall from Section \ref{sec:2.1 channel model} that $\bLambda_\rmb=\diag(\{\sqrt{\Lambda_{\rmb,j}}\}_{j\in[N_\rma]})$ contains singular values of $\bH_\rmb$, implying that $\bLambda_\rmb^2=\diag(\{\Lambda_{\rmb,j}\}_{j\in[N_\rma]})$ contains eigenvalues of $\bH_\rmb^\rmH\bH_\rmb$ with zeros padded for remaining $(N_\rmb-N_\rma)$ eigenvalues. Recall from Section \ref{section_main_results} that $K(\bLambda_\rmb)=\sqrt{2}\tr(\bLambda_\rmb^2)(\tr(\blambda_0^4))^{-\frac{1}{2}}$ and $\kappa_\varepsilon=\inf\{k:\Pr\{K(\bLambda_\rmb)\leq k\}\geq\varepsilon\}$. 
Recall that $\calQ_{(\cdot)}(\cdot)$ denotes the quantile function. 
Choose $\rho=1-\frac{1}{n}$ and $\nu^2(n)=1+\frac{1}{n}$ in \eqref{eq4:power level}. With sufficiently large $n$, it follows from \eqref{eq:ach C final def} that
\begin{align}
C_\varepsilon(\Psi(n))
&=\sup\Big\{ R:\Pr\Big\{ \log\det\big(\bI_{N_\rmb}+\rho\Psi(n)\bH_\rmb^\rmH\bH_\rmb\big)<R \Big\}\leq\varepsilon \Big\}\\*
&=\sup\Bigg\{ R:\Pr\Bigg\{ \sum_{j\in[N_\rmb]}\log \big( \rho\Psi(n)\Lambda_{\rmb,j}+1 \big)<R \Bigg\}\leq\varepsilon \Bigg\}\label{eq:final ach 1}\\*
&=\sup\Bigg\{ R:\Pr\Bigg\{ \sum_{j\in[N_\rmb]}\Big( \rho\Psi(n)\Lambda_{\rmb,j}-\frac{1}{2}\big(\rho\Psi(n)\Lambda_{\rmb,j}\big)^2+o(n^{-1}) \Big)<R \Bigg\}\leq\varepsilon \Bigg\}\label{eq:final ach 2}\\*
&=\sup\Bigg\{ R:\Pr\Bigg\{ \sum_{j\in[N_\rmb]}\bigg(\Lambda_{\rmb,j}\sqrt{\frac{2\delta}{n(1+\frac{1}{n})\tr\big(\blambda_0^4\big)}}-\frac{\Lambda_{\rmb,j}^2\delta}{n(1+\frac{1}{n})\tr\big(\blambda_0^4\big)}+o(n^{-1})\bigg)<R \Bigg\}\leq\varepsilon \Bigg\}\label{eq:final ach 3}\\*
&=\sup\Bigg\{ R:\Pr\Bigg\{ \sum_{j\in[N_\rmb]}\bigg(\Lambda_{\rmb,j}\sqrt{\frac{2\delta}{n\tr\big(\blambda_0^4\big)}}+O(n^{-1})\bigg)<R \Bigg\}\leq\varepsilon \Bigg\}\label{eq:final ach 4}\\*
&=\sup\Bigg\{ R:\Pr\Bigg\{ \tr\big(\bLambda_\rmb^2\big)\sqrt{\frac{2\delta}{n\tr\big(\blambda_0^4\big)}}+O(n^{-1})<R \Bigg\}\leq\varepsilon \Bigg\}\label{eq:final ach 5}\\*
&=\sup\bigg\{ R:\Pr\bigg\{ K(\bLambda_\rmb)\sqrt{\frac{\delta}{n}}+O(n^{-1})<R \bigg\}\leq\varepsilon \bigg\}\label{eq:quantile ach 2}\\*
&=\kappa_\varepsilon\sqrt{\frac{\delta}{n}}+O(n^{-1}),\label{eq:quantile ach 3}
\end{align}
where \eqref{eq:final ach 1} follows from the fact that 
$\det(\bA)=\prod_j \Lambda_j(\bA)$ with $\Lambda_i(\bA)$ denoting the eigenvalues of $\bA$, \eqref{eq:final ach 2} follows from the Taylor Series of expansion of $\log(1+x)$ around $x=0$ which implies $\log(1+x)=x-\frac{1}{2}x^2+o(x^2)$ and the fact that $\Psi(n)=O(n^{-\frac{1}{2}})$, \eqref{eq:final ach 3} follows from \eqref{eq4:power level} and the fact that $\rho=1-\frac{1}{n}$ and $\nu^2(n)=1+\frac{1}{n}$, \eqref{eq:final ach 4} follows from the Taylor Series of expansions of $(1+x)^{-\frac{1}{2}}$ and $(1+x)^{-1}$ around $x=0$ which imply $(1+x)^{-\frac{1}{2}}=1-\frac{1}{2}x+o(x)$ and $(1+x)^{-1}=1-x+o(x)$, respectively, and collecting $o(n^{-1})$, $O(n^{-\frac{3}{2}})$, $O(n^{-2})$ into $O(n^{-1})$, \eqref{eq:final ach 5} follows from the definition of $\bLambda_\rmb^2=\diag(\{\Lambda_{\rmb,j}\}_{j\in[N_\rma]})$ with zeros padded for remaining $(N_\rmb-N_\rma)$ eigenvalues, \eqref{eq:quantile ach 2} follows from the definition of $K(\bLambda_\rmb)$ in \eqref{eq:c1}, and \eqref{eq:quantile ach 3} follows from the fact that $\calQ_Y=a\cdot \calQ_X+b$ for random variables $X$ and $Y=aX+b$, $(a,b)\in\bbR^2$, and note that $\kappa_\varepsilon$ in \eqref{eq:def k epsilon} is indeed the $\varepsilon$-quantile of $K(\bLambda_\rmb)$.

Substituting \eqref{eq:quantile ach 3} into \eqref{eq:ach final 1C} leads to
\begin{align}
R^*_\rmno(n, \varepsilon, \delta)
\geq \kappa_\varepsilon\sqrt{\frac{\delta}{n}}+O\Big(\frac{\log n}{n}\Big).
\end{align}
The achievability proof of Theorem \ref{theorem 1} is now completed.

\section{Converse Proof for Quasi-Static MIMO Fading Channels}
\label{section:con csit}
This section presents the converse proof of Theorem \ref{theorem 1}. Specifically, Section \ref{subsec:con csit covertness} analyzes the covertness constraint and yields an upper bound on the power level required to ensure covertness; Section \ref{subsec:con csit reliability} analyzes the reliability part and yields an upper bound on the transmission rate with vanishing transmit power; finally, Section \ref{subsec:con csit final} combines the above analyses and yields the desired converse result for covert communication over quasi-static MIMO fading channels.

\subsection{Covertness Analysis}
\label{subsec:con csit covertness}
This section analyzes the maximal transmit power to satisfy the covertness constraint for any coding scheme. Recall from Table \ref{tab:notations} that $\blambda_0=\sqrt{\lambda_0}\bI_{N_\rma}$ is the upper bound for channel gains of the channel matrix $\bh_\rmw$ from Alice to Willie. Fix any $\star\in\{\rmno,\rmt,\rmr,\rmrt\}$. Recall that $f_\star$ represents the encoder associated with a specific CSI $\bH_\rmb$ setting between the legitimate users. For any $(n,M)_\star$-code, the average power of the code is defined as 
\begin{align}
\label{eq:def of avg power}
\Psi_\rma(n):=\frac{1}{nM}\sum_{w\in[M]}\|f_\star(w)\|_\rmF^2.
\end{align}
Fix any $\omega>1$ and define a power level
\begin{align}
\label{eq:con power level}
\Psi_\rmm(n):=\sqrt{\frac{2\delta \omega }{n \tr( \blambda_0^4)}}.
\end{align}

\begin{lemma}
\label{lemma:converse power}
To satisfy the covertness constraint in \eqref{eq2:both constraints}, any $(n,M)_\star$-code should have average power $\Psi_\rma(n)$ such that
\begin{align}
\label{eq:con power bound}
\Psi_\rma(n) \leq \Psi_\rmm(n).
\end{align}
\end{lemma}
The proof of Lemma \ref{lemma:converse power} is provided in Appendix \ref{appendix:covert con}. Analogously to \cite{bloch_covert_2016, tahmasbi_first-_2019} for DMCs, the core idea is to convert the covertness constraint, which is imposed on the channel output distributions, into a constraint on the average power of each codeword. Consider the worst case channel defined in \eqref{eq:def of hw0}, we use the Gaussian maximum-entropy property under second moment constraints to derive an upper bound on the average power induced by the covertness constraint in \eqref{eq2:both constraints}.

In Lemma \ref{lemma:converse power}, the power level $\Psi_\rmm(n) = O(\frac{1}{\sqrt{n}})$ follows the square root law.
Under the worst-case assumption in Section~\ref{sec:2.1 channel model}, given a blocklength $n\in\bbN$ and an error probability constraint $\varepsilon\in(0,1)$, define the maximal achievable rate subject to an average power constraint $\Psi_\rma(n)$ as follows:
\begin{align}
R_{\star,\avg}^*(n,\varepsilon,\Psi_\rma(n)):=\sup\big\{R:\exists~ \mathrm{an}~(n,M)_\star\mathrm{-code}~\mathrm{s.t.}~\rmP_\rme^{(n)}\leq\varepsilon~\text{and}~\frac{1}{M}\sum_{w\in[M]}\|f_\star(w)\|_\rmF^2\leq n\Psi_\rma(n),~\forall~\bh_\rmw\in\calH_\rmw\big\}.
\end{align}
It follows from the definition of the theoretical benchmark in \eqref{eq2:both constraints} and Lemma \ref{lemma:converse power} that
\begin{align}
\label{eq:conM}
R^*_\star(n,\varepsilon,\delta)\leq R_{\star,\avg}^*(n,\varepsilon,\Psi_\rmm(n)).
\end{align}
Analogously to the achievability part, this connection enables us to bound the covert transmission rate by the maximal transmission rate subject to a vanishing average power constraint, as done in the next subsection.

\subsection{Reliability Analysis}
\label{subsec:con csit reliability}
This section presents the reliability analysis under a vanishing power constraint $\Psi(n)=O(\frac{1}{\sqrt{n}})$. We begin with a non-asymptotic bound and subsequently derive the asymptotic result. In the converse analysis, we assume that $\bH_\rmb$ is available at both the transmitter Alice and the receiver Bob.

\subsubsection{Non-Asymptotic Bound}
Recall that $N_\rma$ is the number of Alice's antennas, $\bLambda_\rmb=\diag(\{\sqrt{\Lambda_{\rmb,j}}\}_{j\in[N_\rma]})$ is the singular value matrix of $\bH_\rmb$, and $Z_{\rmb,ij}$ is the $(i,j)$-th entry of noise matrix $\bZ_\rmb^n$.
Define two random variables
\begin{align}
L_n(\bLambda_\rmb)&:=\sum_{i\in[n]}\sum_{j\in[N_\rma]}\bigg(\log\big(1+\Lambda_{\rmb,j}\Psi(n)\big)+1-\Big|\sqrt{\Lambda_{\rmb,j} \Psi(n)}Z_{\rmb,ij}-\sqrt{1+\Lambda_{\rmb,j} \Psi(n)}\Big|^2\bigg),\label{eq:def of Lnrt}\\*
S_n(\bLambda_\rmb)&:=\sum_{i\in[n]}\sum_{j\in[N_\rma]}\bigg(\log\big(1+\Lambda_{\rmb,j}\Psi(n)\big)+1-\frac{\big|\sqrt{\Lambda_{\rmb,j} \Psi(n)}Z_{\rmb,ij}-1\big|^2}{1+\Lambda_{\rmb,j} \Psi(n)}\bigg).\label{eq:def of Snrt}
\end{align}
For the CSIRT case, under the worst-case assumption in Section~\ref{sec:2.1 channel model}, given a blocklength $n\in\bbN$, an error probability constraint $\varepsilon\in(0,1)$ and an equal power constraint $\Psi(n)=O(\frac{1}{\sqrt{n}})$, the maximal achievable rate for quasi-static MIMO fading channels is defined as
\begin{align}
R_{\rmrt,\equ}^{*}(n,\varepsilon,\Psi(n)):=\sup\big\{R:\exists~ \mathrm{an}~(n,M)_\rmrt\mathrm{-code}~\mathrm{s.t.}~\rmP_\rme^{(n)}\leq\varepsilon~\text{and}~\forall~{w\in[M]},~\|f_\rmrt(w)\|_\rmF^2= n\Psi(n),~\forall~\bh_\rmw\in\calH_\rmw\big\}.
\end{align}

\begin{lemma}
\label{lemma:R* con non}
Given any $\varepsilon\in(0,1)$, the maximal transmission rate subject to a maximal error probability $\varepsilon\in(0,1)$ and an equal power constraint $\Psi(n)$ satisfies that
\begin{align}
\label{eq:R*rt non}
R_{\rmrt,\equ}^*(n,\varepsilon,\Psi(n))\leq\frac{1}{n}\log\frac{1}{\Pr\big\{L_n(\bLambda_\rmb)\geq n\gamma_n\big\}},
\end{align}
where $\gamma_n$ is chosen such that
\begin{align}
\label{eq:gamma S=epsilon}
\Pr\big\{S_n(\bLambda_\rmb)\leq n\gamma_n\big\}=\varepsilon.
\end{align}
\end{lemma}
The proof of Lemma \ref{lemma:R* con non} is provided in Appendix \ref{appendix:csirt con non}.
To obtain the non-asymptotic bound, we apply the meta-converse theorem \cite[Theorem 30]{polyanskiy_channel_2010} to the sub-channels of the channel \eqref{eq:yb channel} obtained via SVD of $\bH_\rmb$. Under CSIRT, we factor the actual and auxiliary joint distributions using their common marginal $P_{\bLambda_\rmb}$, which cancels in the likelihood ratio and reduces it to a conditional likelihood ratio given $\bLambda_\rmb$.

Lemma \ref{lemma:R* con non} differs significantly from the corresponding results in \cite[Theorems 2, 6]{yang2014quasi} in forms and proofs.
Unlike the non-covert scenario, the vanishing transmit power imposed by the covertness constraint motivates us to redesign the log-likelihood ratio, whose distributions under the auxiliary output and the actual channel output are presented in \eqref{eq:def of Lnrt} and \eqref{eq:def of Snrt}, respectively.

\subsubsection{Asymptotic Analysis}

Given $\Psi(n)=O(\frac{1}{\sqrt{n}})$, define the covert outage rate with $\Psi(n)$ as
\begin{align}
\label{eq:con C final def}
C_{\varepsilon}(\Psi(n))
:=\sup\Big\{ R:\Pr\Big\{ \log\det\big(\bI_{N_\rmb}+\Psi(n)\bH_\rmb^\rmH\bH_\rmb\big)<R \Big\}\leq\varepsilon \Big\}.
\end{align}
Compared with the definition in \eqref{eq:ach C final def}, \eqref{eq:con C final def} does not involve the truncation parameter $\rho$ and is associated with a distinct value of $\Psi(n)$, which will be specified later.
Nevertheless, this difference does not affect the final characterization, as the achievability and converse bounds match.
Analogously to \cite[Lemma 39]{polyanskiy_channel_2010}, the asymptotic result under the maximal power constraint can be obtained from the non-asymptotic bound under the equal power constraint in Lemma \ref{lemma:R* con non}.

\begin{lemma}
\label{lemma:csirt con reliability}
For quasi-static MIMO fading channels with power function $\Psi(n)=O(\frac{1}{\sqrt{n}})$ and $\varepsilon\in(0,1)$, it follows that
\begin{align}
\label{eq:con reliability R}
R_{\rmrt,\max}^*(n,\varepsilon,\Psi(n))\leq C_{\varepsilon}(\Psi(n))+O\Big(\frac{\log n}{n}\Big).
\end{align}
\end{lemma}
The proof of Lemma \ref{lemma:csirt con reliability} is provided in Appendix \ref{appendix:csirt con asy}.
Using the Neyman--Pearson Lemma \cite[Appendix B]{polyanskiy_channel_2010}, \eqref{eq:R*rt non} involving $L_n(\bLambda_\rmb)$ is transformed into \eqref{eq:R* original S} involving $S_n(\bLambda_\rmb)$, thereby enabling analysis under the actual channel output. The remainder of the proof is provided in Appendices \ref{appendix:Proof of CLT S} and \ref{appendix:Proof of eliminate Q in S}.
Specifically, \emph{1)} we condition on $\bLambda_\rmb=\blambda_\rmb$ and use the Cram\'er--Esseen theorem \cite[Theorem 15]{yang2014quasi} to obtain a normal approximation of $\Pr\{S_n(\bLambda_\rmb)\leq n\gamma_n|\bLambda_\rmb\}$, and \emph{2)} we average the asymptotic result over $\bLambda_\rmb$ to eliminate the influence of the second-order term.

Lemma \ref{lemma:csirt con reliability} shows from the converse that quasi-static fading channels retain a zero second-order term under the covertness constraint, in contrast to the non-zero second-order term in the benchmark \eqref{eq:theoretical benchmark awgn} for non-fading channels.
Compared with MIMO AWGN channels \cite[Lemma 5]{liu2026covertMIMO}, the change of the second-order term relies on two key steps. Firstly, following \cite[Appendix IV-A]{yang2014quasi}, we invoke the Cram\'er--Esseen Theorem \cite[Theorem 15]{yang2014quasi} to obtain an expansion accurate up to an $O(\frac{1}{n})$ term, rather than relying on the Berry--Esseen Theorem \cite[Theorem 44]{polyanskiy_channel_2010}, whose expansion is only accurate up to $O(\frac{1}{\sqrt{n}})$.
Secondly, as in \cite[Appendix IV-C]{yang2014quasi}, the randomness of the fading channel allows averaging over all channel realizations, and reapplying Lemma \ref{lemma:B sqrtn A ff'} along with the outage property of quasi-static fading channels, we conclude that the second-order term vanishes.

Compared with the non-covert case, under the vanishing power $\Psi(n)=O(\frac{1}{\sqrt{n}})$, both the quantities involved in verifying the conditions of the Cram\'er--Esseen Theorem and the resulting second-order asymptotics take different forms.
Furthermore, we construct the normalized random variable $u(\bLambda_\rmb)$ of $S_n(\bLambda_\rmb)$ in \eqref{eq:S sum variance} to connect the processed right hand side of \eqref{eq:R*rt non} with the covert outage probability, and characterize its asymptotic behavior using Lemma~\ref{lemma:B sqrtn A ff'}.
Under the scaling of $\Psi(n)$, the pdf of $u(\bLambda_\rmb)$ satisfies $f_U=O(n^\frac{1}{4})$ and its derivative satisfies $f'_U=O(n^\frac{1}{2})$ (cf. Appendices \ref{subapx:f U} and \ref{subapx:f' U}), in contrast to $f_U=f'_U=O(1)$ in \cite[Page 4249]{yang2014quasi}. This change causes the remainder term $O(\frac{1}{n})$ in \cite[Eq.~(163)]{yang2014quasi} to inflate to $O(\frac{1}{\sqrt{n}})$. Analogously to the achievability part, these effects lead to the unexpected $O(\frac{1}{\sqrt{n}})$ term, and are compensated by the outage nature of covert communication.

\subsection{Final Steps}
\label{subsec:con csit final}

Using the covertness analysis in Lemma \ref{lemma:converse power} and the reliability analysis in Lemma \ref{lemma:csirt con reliability}, we upper bound $R^*_\rmrt(n,\varepsilon,\delta)$ in \eqref{eq2:both constraints} as
\begin{align}
R^*_\rmrt(n, \varepsilon, \delta)
&\leq R_{\rmrt,\avg}^*(n,\varepsilon,\Psi_\rmm(n))\label{eq:con R* final 1}\\*
&\leq R_{\rmrt,\max}^*\Big(n,\varepsilon,(1+\frac{1}{n})\Psi_\rmm(n)\Big)+ \frac{\log(n+1)}{n} \label{eq:con R* final 2}\\*
&\leq C_{\varepsilon}\Big((1+\frac{1}{n})\Psi_\rmm(n)\Big)+O\Big(\frac{\log n}{n}\Big),\label{eq:con R* final 3}
\end{align}
where \eqref{eq:con R* final 1} follows from \eqref{eq:conM}, \eqref{eq:con R* final 2} follows from \cite[Lemma 39]{polyanskiy_channel_2010} with $P=\Psi_\rmm(n)$ and $P'=(1+\frac{1}{n})\Psi_\rmm(n)$ in \cite[Eq. (196)]{polyanskiy_channel_2010}, and \eqref{eq:con R* final 3} follows from \eqref{eq:con reliability R}.

Recall from Section \ref{sec:2.1 channel model} that $\bLambda_\rmb=\diag(\{\sqrt{\Lambda_{\rmb,j}}\}_{j\in[N_\rma]})$ contains singular values of $\bH_\rmb$, implying that $\bLambda_\rmb^2=\diag(\{\Lambda_{\rmb,j}\}_{j\in[N_\rma]})$ contains eigenvalues of $\bH_\rmb^\rmH\bH_\rmb$ with zeros padded for remaining $(N_\rmb-N_\rma)$ eigenvalues. Recall from Section \ref{section_main_results} that $K(\bLambda_\rmb)=\sqrt{2}\tr(\bLambda_\rmb^2)(\tr(\blambda_0^4))^{-\frac{1}{2}}$ and $\kappa_\varepsilon=\inf\{k:\Pr\{K(\bLambda_\rmb)\leq k\}\geq\varepsilon\}$. 
Choose $\omega=1+\frac{1}{n}$ in \eqref{eq:con power level}. With sufficiently large $n$, it follows from \eqref{eq:con C final def} that
\begin{align}
C_{\varepsilon}\Big((1+\frac{1}{n})\Psi_\rmm(n)\Big)
&=\sup\bigg\{ R:\Pr\bigg\{ \log\det\Big(\bI_{N_\rmb}+(1+\frac{1}{n})\Psi_\rmm(n)\bH_\rmb^\rmH\bH_\rmb\Big)<R \bigg\}\leq\varepsilon \bigg\}\\*
&=\sup\Bigg\{ R:\Pr\Bigg\{ \sum_{j\in[N_\rmb]}\log \Big( (1+\frac{1}{n})\Psi_\rmm(n)\Lambda_{\rmb,j}+1 \Big)<R \Bigg\}\leq\varepsilon \Bigg\}\label{eq:final con 1}\\*
&=\sup\Bigg\{ R:\Pr\Bigg\{ \sum_{j\in[N_\rmb]}\Big( (1+\frac{1}{n})\Psi_\rmm(n)\Lambda_{\rmb,j}-\frac{1}{2} \big((1+\frac{1}{n})\Psi_\rmm(n)\Lambda_{\rmb,j}\big)^2+o(n^{-1}) \Big)<R \Bigg\}\leq\varepsilon \Bigg\}\label{eq:final con 2}\\*
&=\sup\Bigg\{ R:\Pr\Bigg\{ \sum_{j\in[N_\rmb]}\bigg(\Lambda_{\rmb,j}\sqrt{\frac{2\delta (1+\frac{1}{n})^3 }{n \tr( \blambda_0^4)}}-\frac{\Lambda_{\rmb,j}^2\delta(1+\frac{1}{n})^3}{n\tr( \blambda_0^4)}+o(n^{-1})\bigg)<R \Bigg\}\leq\varepsilon \Bigg\}\label{eq:final con 3}\\*
&=\sup\Bigg\{ R:\Pr\Bigg\{ \sum_{j\in[N_\rmb]}\bigg(\Lambda_{\rmb,j}\sqrt{\frac{2\delta}{n\tr\big(\blambda_0^4\big)}}+O(n^{-1})\bigg)<R \Bigg\}\leq\varepsilon \Bigg\}\label{eq:final con 4}\\*
&=\kappa_\varepsilon\sqrt{\frac{\delta}{n}}+O(n^{-1}),\label{eq:quantile con 3}
\end{align}
where \eqref{eq:final con 1} follows from the fact that 
$\det(\bA)=\prod_j \Lambda_j(\bA)$ with $\Lambda_i(\bA)$ denoting the eigenvalues of $\bA$, \eqref{eq:final con 2} follows from the Taylor Series of expansion of $\log(1+x)$ around $x=0$ which implies $\log(1+x)=x-\frac{1}{2}x^2+o(x^2)$ and the fact that $\Psi_\rmm(n)=O(n^{-\frac{1}{2}})$, \eqref{eq:final con 3} follows from \eqref{eq:con power level} and the fact that $\omega=1+\frac{1}{n}$, \eqref{eq:final con 4} follows from the Taylor Series of expansions of $(1+x)^{\frac{3}{2}}$ and $(1+x)^{3}$ around $x=0$, which implies $(1+x)^{\frac{3}{2}}=1+\frac{3}{2}x+o(x)$ and $(1+x)^{3}=1+3x+o(x)$, respectively. Note that the right hand side of \eqref{eq:final con 4} equals the right hand side of \eqref{eq:final ach 4}, \eqref{eq:quantile con 3} obviously follows from \eqref{eq:quantile ach 3}.

Substituting \eqref{eq:quantile con 3} into \eqref{eq:con R* final 3} leads to
\begin{align}
R^*_\rmrt(n, \varepsilon, \delta)
\leq \kappa_\varepsilon\sqrt{\frac{\delta}{n}}+O\Big(\frac{\log n}{n}\Big).
\end{align}
The converse proof of Theorem \ref{theorem 1} is now completed.

\section{Conclusion}
\label{section:conclusion}

We studied the second-order asymptotics of optimal codes for covert communication over quasi-static MIMO fading channels under the KL divergence covertness metric.
Specifically, we showed that the first-order covert rate follows the square root law and scales in the order of $\Theta(n^{-\frac{1}{2}})$ with a coefficient determined by the traces of the channel matrices of legitimate users and the warden, and the second-order rate vanishes. In particular, we showed that CSI availability at the legitimate transmitter can improve the first-order covert rate, and we demonstrated the critical role of the number of antennas to achieve high throughput covert communication.
Our theoretical benchmarks generalized the previous results for the non-fading AWGN case in~\cite{liu2026covertMIMO} and the non-covert case in~\cite{yang2014quasi}, and provided new design insights.
For the covertness analysis, we extended the quasi-$\eta$-neighborhood framework \cite{yu_second_2023,liu2026covertMIMO} to quasi-static fading channels. For the reliability analysis, due to the vanishing power imposed by the covertness constraint, we refined second-order asymptotic techniques \cite{yang2014quasi}, by controlling higher-order terms and exploiting the order behavior of the covert outage probability.

We next discuss future research directions.
Firstly, we focused on P2P MIMO channels with only one legitimate transmitter or receiver. However, practical communication systems usually involved multiple users. Thus, it is worthwhile to extend our results to the multiple access channel~\cite{arumugam2019covert} and the broadcast channel~\cite{tan2018time} .
Secondly, we assumed a secret key of infinite length is available between legitimate users. However, in practice, key sizes are usually limited. Thus, it is valuable to bound the key size necessary to ensure covert communication and study the tradeoff between the key rate and the covert rate~\cite{tahmasbi_first-_2019,wang_covert_2021}. Thirdly, we focused on theoretical benchmarks and used random coding with exhaustive search based decoding, which has prohibitively high computational complexity. However, practical communication systems are usually resource limited. Thus, it is of great interest to design low-complexity covert coding schemes, potentially using machine learning techniques~\cite{jiang2019turbo,choi2019neural,ozyilkan2025learning}.
Finally, we assumed that the warden Willie is strong and has perfect knowledge of the CSI for the channel between the legitimate transmitter and the warden. However, in practice, obtaining exact CSI is challenging. Thus, it is beneficial to study the impact of CSI imperfection of the warden on the covert transmission rate~\cite{soltani2018covert,he2017covert,sobers2017covert,hayashi2023covert}.

\appendix

\subsection{Proof of Corollary~\ref{corollary:CSIT gain}}
\label{appendix:cor1}


\subsubsection{Achievability Proof of Corollary~\ref{corollary:CSIT gain}}
\label{subsubapx:ach opt power}

Following Definition~\ref{def:distributions}, we now allow different powers to be allocated across the transmit antennas. Define the vector power allocation as $\bPsi(n):=(\Psi_1(n),\ldots,\Psi_{N_\rma}(n))$, where $\Psi_j(n)=O(\frac{1}{\sqrt{n}})$ for any $j\in[N_\rma]$. Given any $M\in\bbN$, a random codebook $\{\bX^n(w)\}_{w\in[M]}$ is generated such that each codeword $\bX^n(w)\in\calX^{m\times n}$ is independently drawn from $\Pi_{\bX^n}^{\mathrm{tG}}=\prod_{j\in[N_\rma]} \Pi_{X_j^n}^{\mathrm{tG}}(X^n_j)$, where $\Pi_{X_j^n}^{\mathrm{tG}}$ is the truncated Gaussian distribution defined in \eqref{eq:tG} with $\Psi(n)$ therein replaced by $\Psi_j(n)$ for each $j\in[N_\rma]$. Recall from Table \ref{tab:notations} that $\blambda_0=\sqrt{\lambda_0}\bI_{N_\rma}$ is the upper bound for channel gains of the channel matrix $\bh_\rmw$ from Alice to Willie. Recall that $\rho\in(0,1)$, $\nu(n)>1$ and $\lim\limits_{n\to\infty}\nu(n)=1$. For any $n\in\bbN$, when the power is set as
\begin{align}
\label{eq4:power level opt}
\frac{n\rho^2\nu^2(n)\lambda_0^2}{2}\sum_{j\in[N_\rma]}\Psi_j^2(n)\le\delta,
\end{align}
using the truncated Gaussian codebook where each codeword is generated from $\Pi_{\bX^n}^\mathrm{tG}$, the covertness constraint \eqref{eq4:covertness constraint} is satisfied. The proof of \eqref{eq4:power level opt} is identical to that of Lemma~\ref{lemma:covert}, except that equal power allocation is replaced by the different sub-channel power $\Psi_j(n)$ for each $j\in[N_\rma]$.
We next analyze reliability by establishing a non-asymptotic bound and then deriving its asymptotic characterization. Recall that $\calB(\cdot,\cdot)$ is the Beta distribution, $(N_\rma,N_\rmb)$ are the numbers of antennas at Alice and Bob, respectively, $\{\sqrt{\Lambda_{\rmb,j}}\}_{j\in[N_\rma]}$ are singular values of $\bH_\rmb$, and $Z_{\rmb,ij}$ be the $(i,j)$-th entry of noise matrix $\bZ_\rmb^n$. For each $j\in[N_\rma]$, define the following random variable
\begin{align}
\label{eq:def of Tj opt}
T_j':=\frac{\sum_{i\in[2,n]}\big|Z_{\rmb,ij}\big|^2}{\big|\sqrt{n\rho\Psi_j(n)\Lambda_{\rmb,j}}+Z_{\rmb,1j}\big|^2+\sum_{i\in[2,n]}\big|Z_{\rmb,ij}\big|^2}.
\end{align}
Fix any $\varepsilon\in(0,1)$ and $\tau\in(0,\varepsilon)$. There exists an $(n,M)_\rmt$-code with the maximal error probability constraint $\varepsilon$ such that
\begin{align}
\label{eq:logM/n non ach opt}
\frac{\log M}{n}\geq\frac{1}{n}\log\frac{\tau}{\Pr\Big\{\prod_{j\in[N_\rmb]}B_j\leq\gamma_n\Big\}},
\end{align}
where $\{B_j\}_{j\in[N_\rmb]}$ are i.i.d. with $B_j\sim\calB(n-N_\rma-j+1,N_\rma)$, and $\gamma_n$ is chosen so that
\begin{align}
\label{eq:sin2 pr 1-a opt}
\Pr\bigg\{\prod_{j\in[N_\rma]}T_j' \leq {\gamma _n}\bigg\}\geq 1-\varepsilon+\tau.
\end{align}
The proof of \eqref{eq:logM/n non ach opt} is identical to that of Lemma~\ref{lemma:csit ach non}, except that equal power allocation is replaced by the optimizable vector power allocation $\{\Psi_j(n)\}_{j\in[N_\rma]}$, and CSIT enables precoding using $\bL_\rmb^\rmH$ in \eqref{eq:svd Hb}.
Define the feasible set of vector power allocations satisfying \eqref{eq4:power level opt} as
\begin{align}
\calV_n:=\bigg\{\bPsi(n)\in\bbR_+^{N_\rma}:\frac{n\rho^2\nu^2(n)\lambda_0^2}{2}\sum_{j\in[N_\rma]}\Psi_j^2(n)\leq\delta\bigg\}.
\label{eq:feasible power set}
\end{align}
For any $j\in[N_\rma]$, the covert outage rate with $\{\Psi_j(n)\}_{j\in[N_\rma]}=O(\frac{1}{\sqrt{n}})$ for CSIT under optimal power allocation is given by
\begin{align}
C_{\varepsilon}^\rmt(n):=\sup\Bigg\{R:\Pr\Bigg\{
\max_{\bPsi(n)\in\calV_n}\sum_{j\in[N_\rma]}\log\big(1+\rho\Psi_j(n)\Lambda_{\rmb,j}\big)<R\Bigg\}\leq\varepsilon\Bigg\}.
\label{eq:ach C final def opt}
\end{align}
The proof of \eqref{eq:ach C final def opt} is identical to that of Lemma~\ref{lemma:csit ach reliability}, except that, after removing the equal power restriction, we optimize power allocation across the transmit directions to maximize the covert outage rate.
It follows from \cite[Lemma~49]{poly_Channel_2010} that optimizing only the leading term of the covert outage rate in \eqref{eq:ach C final def opt} is sufficient. Recall that $\{\Psi_j(n)\}_{j\in[N_\rma]}=O(\frac{1}{\sqrt{n}})$.
Given $\bLambda_\rmb=\blambda_\rmb$, combining the covertness constraint in \eqref{eq4:power level opt} and the leading term of the reliability bound in \eqref{eq:ach C final def opt} yields the following optimization problem:
\begin{align}
&\max_{\Psi_j(n)\ge0,\ j\in[N_\rma]}\sum_{j\in[N_\rma]}\lambda_{\rmb,j}\Psi_j(n),\label{eq:optimal solve}\\
&\text{s.t. } \frac{n\rho^2\nu^2(n)\lambda_0^2}{2}\sum_{j\in[N_\rma]}\Psi_j^2(n)\le \delta .\label{eq:optimal constraint}
\end{align}
For any feasible $\bPsi(n)$, it follows from the Cauchy--Schwarz inequality that
\begin{align}
\bigg(\sum_{j\in[N_\rma]}\lambda_{\rmb,j}\Psi_j(n)\bigg)^2
&\leq\bigg(\sum_{j\in[N_\rma]}\Psi_j^2(n)\bigg)\bigg(\sum_{j\in[N_\rma]}\lambda_{\rmb,j}^2\bigg)\label{eq:csit_opt_cs_1}\\
&\leq\frac{2\delta}
{n\rho^2\nu^2(n)\lambda_0^2}
\sum_{j\in[N_\rma]}\lambda_{\rmb,j}^2,
\label{eq:csit_opt_cs_2}
\end{align}
where \eqref{eq:csit_opt_cs_2} follows from \eqref{eq:optimal constraint}.
It remains to show that the upper bound in
\eqref{eq:csit_opt_cs_2} is achievable. It follows from the Cauchy--Schwarz inequality that the equality in \eqref{eq:csit_opt_cs_1} holds if and only if there exists a constant $c>0$ such that for any $j\in[N_\rma]$,
\begin{align}
\Psi_j(n)=c\lambda_{\rmb,j}.
\label{eq:csit_opt_equality}
\end{align}
To satisfy the constraint in \eqref{eq:optimal constraint}, we choose $c$ such that
\begin{align}
\frac{n\rho^2\nu^2(n)\lambda_0^2}{2}c^2\sum_{j\in[N_\rma]}\lambda_{\rmb,j}^2=\delta.
\label{eq:csit_opt_c}
\end{align}
Given $\bLambda_\rmb=\blambda_\rmb$, let $\bPsi^*(n,\blambda_\rmb):=(\Psi_1^*(n,\blambda_\rmb),\ldots,\Psi_{N_\rma}^*(n,\blambda_\rmb))$ be optimal power allocation.
Substituting \eqref{eq:csit_opt_c} into
\eqref{eq:csit_opt_equality} leads that for any $j\in[N_\rma]$,
\begin{align}
\label{eq:csit_opt_psi_star}
\Psi_j^*(n,\blambda_\rmb)=\frac{\lambda_{\rmb,j}}
{\sqrt{\sum_{j'\in[N_\rma]}\lambda_{\rmb,j}^2}}\cdot
\sqrt{\frac{2\delta}{n\rho^2\nu^2(n)\lambda_0^2}}
=\frac{\lambda_{\rmb,j}}{\sqrt{\tr(\blambda_\rmb^4)}}\cdot\sqrt{\frac{2\delta}{n\rho^2\nu^2(n)\lambda_0^2}}.
\end{align}
Recall from Section \ref{section_main_results} that $\bLambda_\rmb=\diag(\{\sqrt{\Lambda_{\rmb,j}}\}_{j\in[N_\rma]})$ contains singular values of $\bH_\rmb$, $K'(\bLambda_\rmb)=\frac{\sqrt{2}}{\lambda_0}\sqrt{\sum_{j\in[N_\rma]}\Lambda_{\rmb,j}^2}$ and $\kappa_\varepsilon'=\inf\{k:\Pr\{K'(\bLambda_\rmb)\leq k\}\geq\varepsilon\}$. 
With sufficiently large $n$, similarly to \eqref{eq:quantile ach 3}, choosing $\rho=1-\frac{1}{n}$ and $\nu^2(n)=1+\frac{1}{n}$, and substituting \eqref{eq:csit_opt_psi_star} into \eqref{eq:ach C final def opt} lead to
\begin{align}
C_\varepsilon^\rmt(n)=\kappa_\varepsilon'\sqrt{\frac{\delta}{n}}+O(n^{-1}).
\end{align}
The achievability proof of Corollary~\ref{corollary:CSIT gain} is completed.

\subsubsection{Converse Proof of Corollary~\ref{corollary:CSIT gain}}
\label{subsubapx:con opt power}

Based on the SVD in Section~\ref{section_main_results}, let $x_{ij}(w)\in\calX$ be the symbol transmitted on the $j$-th SVD sub-channel at $i$-th channel use when message $w\in[M]$ is transmitted. Define a vector of average power functions $\bPsi_\rma(n):=(\Psi_{\rma,1}(n),\ldots,\Psi_{\rma,N_\rma}(n))$ where for any $j\in[N_\rma]$,
\begin{align}
\label{eq:def of psi aj}
\Psi_{\rma,j}(n):=\frac{1}{nM}\sum_{i\in[n]}\sum_{w\in[M]}\big|x_{ij}(w)\big|^2.
\end{align}
Recall from Table \ref{tab:notations} that $\blambda_0=\sqrt{\lambda_0}\bI_{N_\rma}$.
Fix any $\omega>1$. Without equal power allocation, any $(n,M)_\star$-code that satisfies the covertness constraint in \eqref{eq2:both constraints} has average power $\Psi_{\rma,j}(n)$ such that
\begin{align}
\label{eq:con power bound opt}
\frac{n\lambda_0^2}{2\omega}\sum_{j\in[N_\rma]}\Psi_{\rma,j}^2(n)\le\delta.
\end{align}
The proof of \eqref{eq:con power bound opt} is identical to that of Lemma~\ref{lemma:converse power}, except that equal power allocation is replaced by the different sub-channel average power $\Psi_{\rma,j}(n)$ for each $j\in[N_\rma]$. Define the feasible set of vector power allocations satisfying \eqref{eq:con power bound opt} as
\begin{align}
\calV_{\rma,n}:=\bigg\{\bPsi_\rma(n)\in\bbR_+^{N_\rma}:\frac{n\lambda_0^2}{2\omega}\sum_{j\in[N_\rma]}\Psi_{\rma,j}^2(n)\le\delta\bigg\}.
\label{eq:feasible power set con}
\end{align}
Recall that the asymptotic reliability bound in
Lemma~\ref{lemma:csirt con reliability} is derived for the
CSIRT case under equal power allocation. Extending the same
analysis to arbitrary feasible vector power allocations, we define the optimized covert outage rate for CSIRT as
\begin{align}
C_\varepsilon^\rmrt(n):=\sup\Bigg\{R:\Pr\Bigg\{
\max_{\bPsi_\rma(n)\in\calV_{\rma,n}}\sum_{j\in[N_\rma]}\log\big(1+\Lambda_{\rmb,j}\Psi_{\rma,j}(n)\big)<R\Bigg\}\leq\varepsilon\Bigg\}.
\label{eq:con C final def opt}
\end{align}
Using the covertness analysis in \eqref{eq:con power bound opt} and the reliability analysis in \eqref{eq:con C final def opt}, we optimize the power vector $\bPsi_\rma(n)$. Let $\bPsi_\rma^*(n,\blambda_\rmb):=(\Psi_{\rma,1}^*(n,\blambda_\rmb),\ldots,\Psi_{\rma,N_\rma}^*(n,\blambda_\rmb))$ be the power vector that maximizes the converse bound. 
Given $\bLambda_\rmb=\blambda_\rmb$, among all vector average power vectors $\bPsi_\rma(n)$ satisfying the KL divergence constant in \eqref{eq:con power bound opt}, the power vector $\bPsi_\rma^*(n)$ that maximizes the covert outage rate in \eqref{eq:con C final def opt} is that for any $j\in[N_\rma]$,
\begin{align}
\Psi_{\rma,j}^*(n,\blambda_\rmb)=\frac{\lambda_{\rmb,j}}{\sqrt{\tr(\blambda_\rmb^4)}}\cdot\sqrt{\frac{2\delta\omega}{n\lambda_0^2}}.
\label{eq:converse_optimized_power}
\end{align}
The proof of \eqref{eq:converse_optimized_power} is identical to that of \eqref{eq:csit_opt_psi_star}, except that $\Psi_j(n)$ in \eqref{eq:optimal solve} is replaced by $\Psi_{\rma,j}(n)$, and the constraint \eqref{eq:optimal constraint} is replaced by \eqref{eq:con power bound opt}.
Recall from Section \ref{section_main_results} that $\bLambda_\rmb=\diag(\{\sqrt{\Lambda_{\rmb,j}}\}_{j\in[N_\rma]})$ contains singular values of $\bH_\rmb$, $K'(\bLambda_\rmb)=\frac{\sqrt{2}}{\lambda_0}\sqrt{\sum_{j\in[N_\rma]}\Lambda_{\rmb,j}^2}$ and $\kappa_\varepsilon'=\inf\{k:\Pr\{K'(\bLambda_\rmb)\leq k\}\geq\varepsilon\}$. 
With sufficiently large $n$, similarly to \eqref{eq:quantile con 3}, choosing $\omega=1+\frac{1}{n}$ and substituting \eqref{eq:converse_optimized_power} into \eqref{eq:con C final def opt} lead to
\begin{align}
C_\varepsilon^\rmrt(n)=\kappa_\varepsilon'\sqrt{\frac{\delta}{n}}+O(n^{-1}).
\end{align}
The converse proof of Corollary~\ref{corollary:CSIT gain} is completed.

\subsection{Proof of Lemma \ref{lemma:covert} (Covertness Analysis for Achievability)}
\label{appendix:covert ach}

For clarity, we restate the definition of the quasi-$\eta$-neighborhood here.
Fix two integers $(n,m)\in\bbN^2$, a positive real number $\eta\in\bbR_+$ and a continuous alphabet $\calY\subseteq\bbC$. Let $\by^n\in\calY^{n\times m}$. For distributions $P$ and $Q$ on the alphabet $\calY^{n\times m,}$, let $f_P(\by^n)$ and $f_Q(\by^n)$ be their densities, respectively. Assume that $f_P$ is absolutely continuous with respect to $f_Q$. The $\chi^2$-divergence can be defined as 
\begin{align}
\label{eq:def chi2}
\chi_2(P\|Q):=\int_{\by^n\in\calY^{n\times m}}\frac{(f_P(\by^n)-f_Q(\by^n))^2}{f_Q(\by^n)}\rmd \by^n.
\end{align}
Given any distribution $Q$ on the alphabet $\calY^{n\times m}$, the quasi-$\eta$-neighborhood \cite[Definition 4]{yu_second_2023} of $Q$ is a set of distributions $\{P\}$ in a $\chi^2$-divergence ball of $\eta^2$, i.e.,
\begin{align}
\label{def_epsilon_neighborhood}
\calN_\eta(Q):= \big\{P\in\calP(\calY^{n\times m}):\chi_2(P\|Q)\leq\eta^2 \big\}.
\end{align}
Furthermore, define an auxiliary function,
\begin{align}\label{eq:phi def}
\calD_{\eta}(\by^n):=\frac{f_P(\by^n)-f_Q(\by^n)}{\eta\cdot f_Q(\by^n)}.
\end{align}
The following lemma bounds KL divergence and TV distance for distributions in the quasi-$\eta$-neighborhood.
\begin{lemma}\label{lemma_epsilon_dv}
For any distribution $P\in\calN_\eta(Q)$ with $\sup_{\by^n\in\calY^{n\times m}}|\calD_\eta(\by^n)|\leq 1$, we have
\begin{align}
\bbV(P, Q)&\leq\frac{1}{2}\eta,\label{eq:vtd epsilon}\\*
\bbD(P\|Q)&\leq\eta^2.\label{eq:kld epsilon}
\end{align}
\end{lemma}
The proof of Lemma \ref{lemma_epsilon_dv} can be found in \cite[Appendix E]{liu2026covertMIMO}.

In the following, Appendix \ref{subapx:neighborhood 1} proves that the truncated complex Gaussian distribution defined in \eqref{eq:tG} lies within the quasi-$\eta$-neighborhood of the noise distribution and relates to the covertness constraint, and Appendix \ref{subapx:neighborhood 3} bounds the power constraint of the truncated complex codebook to satisfy the covertness constraint in \eqref{eq4:covertness constraint}.

\subsubsection{Quasi-$\eta$-Neighborhood of Channel Noise for Covertness}
\label{subapx:neighborhood 1}
Fix a deterministic channel matrix $\bh_\rmw\in\calH_\rmw$.
Recall from Table~\ref{tab:notations} that $Q_{\bZ_\rmw^n}$ is the noise distribution at Willie.
With a slight abuse of notation, we use $\Pi_{\bX^n\bh_\rmw}^\mathrm{tG}$ to denote the distribution of $\bX^n\bh_\rmw$, where $\bX^n\sim\Pi^\mathrm{tG}_{\bX^n}$, and similarly for $\Pi_{\bX^n\bh_\rmw}^\rmG$.
From the definition of the channel, we have $Q_{\bY_\rmw^n}^\mathrm{tG}=\Pi_{\bX^n\bh_\rmw}^\mathrm{tG}*Q_{\bZ_\rmw^n}$.
For any $x\in\calX$, the $1$-dimensional Dirac measure is denoted as $\delta_0^1(x)$, which equals zero when $x\ne 0$ and satisfies $\int_{x\in\calX}\delta_0^1(x)\rmd x=1$. Analogously, we define the $n$-dimensional Dirac measure $\delta_0^n(x^n):= \prod_{i\in[n]}\delta_0^1(x_i)$. 
Let $\delta_0^{n\times N_\rmw}$ be the Dirac measure for $\bX^n\bh_\rmw\in\calX^{n\times N_\rma}\times\bbC^{N_\rma\times N_\rmw}$. From the property of the Dirac measure, we have $Q_{\bZ_\rmw^n}=\delta_0^{n\times N_\rmw}*Q_{\bZ_\rmw^n}$.
Therefore, $Q_{\bY_\rmw^n}^\mathrm{tG}$ and $Q_{\bZ_\rmw^n}$ can be viewed as the results of convolving $Q_{\bZ_\rmw^n}$ with $\Pi_{\bX^n\bh_\rmw}^\mathrm{tG}$ and $\delta_0^{n\times N_\rmw}$, respectively.
The realization with respect to $Q_{\bY_\rmw^n}^\mathrm{tG}$ of the quasi-$\eta$-neighborhood distribution of $Q_{\bZ_\rmw^n}$ is closely related to the convergence between $\Pi_{\bX^n\bh_\rmw}^\mathrm{tG}$ and $\delta_0^{n\times N_\rmw}$.
Let $\|\cdot\|^*$ be a weak convergence metric of probability measures such that a sequence of probability distributions $\{P_n\}$ on $\calX$ converges weakly to a probability distribution $P_0$ with $\|P_n-P_0\|^*\to 0$ if for any bounded continuous function $g:\calX\to\bbC$,
\begin{align}
\label{eq:def of weak}
\lim_{n\to\infty}\Big|\int_{x\in\calX}g(x)\rmd P_n(x)-\int_{x\in\calX}g(x)\rmd P_0(x)\Big|=0.
\end{align}
We first show that the distribution $\Pi_{\bX^n\bh_\rmw}^\mathrm{tG}$ with $\Psi(n)=O(\frac{1}{\sqrt{n}})$ satisfies
\begin{align}\label{lemma:appendix_ach_weak}
\big\|\Pi^\mathrm{tG}_{\bX^n\bh_\rmw}-\delta_0^{n\times N_\rmw}\big\| ^* \to 0 \text{ as }n\to\infty.
\end{align}
As $\|\cdot\|^*$ is a metric, we have
\begin{align}
\big \|\Pi_{\bX^n\bh_\rmw}^\mathrm{tG}-\delta_0^{n\times N_\rmw} \big \| ^*\leq \big \| \Pi_{\bX^n\bh_\rmw}^\mathrm{tG}-\Pi_{\bX^n\bh_\rmw}^\rmG \big \| ^* + \big \| \Pi_{\bX^n\bh_\rmw}^\rmG-\delta_0^{n\times N_\rmw} \big \| ^*.
\end{align}

Note that the distributions $\Pi_{\bX^n\bh_\rmw}^\mathrm{tG}$ and $\Pi_{\bX^n\bh_\rmw}^\rmG$ with $\Psi(n)=O(1/\sqrt{n})$ satisfy
\begin{align}
\big \| \Pi_{\bX^n\bh_\rmw}^\mathrm{tG}-\Pi_{\bX^n\bh_\rmw}^\rmG\big \|^*&\to 0,\label{eq:awgn_ach_weak_GtG}\\*
\big \| \Pi_{\bX^n\bh_\rmw}^\rmG-\delta_0^{n\times N_\rmw} \big \| ^*&\to 0,\label{eq:awgn_ach_weak_0G}
\end{align}
as $n\to\infty$, respectively.
The proofs of \eqref{eq:awgn_ach_weak_GtG} and \eqref{eq:awgn_ach_weak_0G} are consistent with those in the MIMO AWGN case \cite[Eqs. (178), (179)]{liu2026covertMIMO}, respectively. It follows from \cite[Eqs. (180)--(185)]{liu2026covertMIMO} that for any $\eta\in\bbR_+$, there exists an integer $N(\eta)\in\bbN$ such that for any $n>N(\eta)$, the distribution $Q^\mathrm{tG}_{\bY_\rmw^n}$, which is the induced output distribution of the quasi-static MIMO fading channel from Alice to Willie for the input distribution $\Pi_{\bX^n}^\mathrm{tG}$ with $\Psi(n)=O(\frac{1}{\sqrt{n}})$, is in the quasi-$\eta$-neighborhood of $Q_{\bZ_\rmw^n}$ and satisfies $\sup_{\by_\rmw^n\in\calY_\rmw^{n\times N_\rmw}}|\calD_\eta(\by_\rmw^n)|\leq 1$. Furthermore, it follows from Lemma \ref{lemma_epsilon_dv} that when $\Psi(n) = O(\frac{1}{\sqrt{n}})$ and $n$ is sufficiently large,
\begin{align}
\bbV(Q^\mathrm{tG}_{\bY_\rmw^n}, Q_{\bZ_\rmw^n})&\leq \frac{1}{2}\eta \label{equ_output_condition_v},\\*
\bbD(Q^\mathrm{tG}_{\bY_\rmw^n} \| Q_{\bZ_\rmw^n})&\leq \eta^2.\label{equ_output_condition_d}
\end{align}
Recall that $\Pi_{\bX^n}^\rmG$ in \eqref{eq:G} is the complex Gaussian distribution, without the truncation operation. Let the distribution of Willie's output $\bY_\rmw^n$ induced by $\Pi_{\bX^n}^\rmG$ with $\Psi(n)=O(\frac{1}{\sqrt{n}})$ under the state $\rmS_1$ be denoted $Q_{\bY_\rmw^n}^\rmG$. Analogously, we can show that $Q_{\bY_\rmw^n}^\rmG$ also lies in the quasi-$\eta$-neighborhood of $Q_{\bZ_\rmw^n}$.

Although \eqref{equ_output_condition_d} shows that truncated complex Gaussian codebook satisfies the covertness constraint, the KL divergence $\bbD(Q^\mathrm{tG}_{\bY_\rmw^n} \| Q_{\bZ_\rmw^n})$ is challenging to calculate. To solve this problem, in the following lemma, we provide a sufficient condition to satisfy the covertness constraint when the truncated complex Gaussian codebook is replaced by the \iid complex Gaussian codebook. Recall that for each $n\in\bbN$, $\nu(n)>1$ and $\lim_{n\to\infty}\nu(n)=1$.

\begin{lemma}\label{lemma:awgn_ach_covert_GtG}
For any $\delta\in\bbR_+$, there exists a sufficiently large integer $N(\delta)\in\bbN$ such that for any $n>N(\delta)$, if
\begin{align}
\label{eq:KL divergence transform}
\bbD(Q^\mathrm{G}_{\bY^n_\rmw}\| Q_{\bZ_\rmw^n}) \leq \frac{\delta}{\nu^2(n)},
\end{align}
the covertness constraint \eqref{eq4:covertness constraint} is satisfied.
\end{lemma}
The proof of Lemma \ref{lemma:awgn_ach_covert_GtG} is identical to that of \cite[Lemma 8]{liu2026covertMIMO}, except that the condition is replaced here by \eqref{equ_output_condition_v} and \eqref{equ_output_condition_d}.

\subsubsection{An Upper Bound on Transmit Power for Covertness}
\label{subapx:neighborhood 3}
We shall upper bound the transmit power $\Psi(n)$ to satisfy the covertness constraint in \eqref{eq4:covertness constraint} using the truncated complex Gaussian codebook via Lemma \ref{lemma:awgn_ach_covert_GtG}.
It follows from Definition \ref{def:distributions} that $Q_{\bY_{\rmw}}^\mathrm{G}=\calCN(\mathbf{0},\rho\Psi(n)\bh_\rmw^\rmH\bh_\rmw+\bI_{N_\rmw})$. Thus, we have
\begin{align}
\bbD(Q^\rmG_{\bY_\rmw^n} \| Q_{\bZ_\rmw^n})
&=\bbE_{Q^\rmG_{\bY_\rmw^n}} \bigg[ \log\frac{Q^\rmG_{\bY_\rmw^n}(\bY_\rmw^n)}{Q_{\bZ_\rmw^n}(\bY_\rmw^n)} \bigg]\label{eq:KL divergence G 1}\\*
&=\sum_{i\in[n]}\bbE_{Q^\rmG_{\bY_{\rmw}}}\bigg[\log\frac{Q^\rmG_{\bY_{\rmw}}(\bY_{\rmw,i})}{Q_{\bZ_\rmw}(\bY_{\rmw,i})}\bigg]\label{eq:KL divergence G 2}\\*
&=\sum_{i\in[n]}\bbE_{Q^\rmG_{\bY_{\rmw}}}\bigg[\log\frac{(\pi^{N_\rmw}\det(\rho\Psi(n)\bh_\rmw^\rmH\bh_\rmw+\bI_{N_\rmw}))^{-1}\exp(-\bY_{\rmw,i}^\rmH\bY_{\rmw,i}(\rho\Psi(n)\bh_\rmw^\rmH\bh_\rmw+\bI_{N_\rmw})^{-1})}{(\pi^{N_\rmw})^{-1}\exp(-\bY_{\rmw,i}^\rmH\bY_{\rmw,i})}\bigg]\label{eq:KL divergence G 3}\\*
&=\sum_{i\in[n]}\bbE_{Q^\rmG_{\bY_{\rmw}}}\big[-\bY_{\rmw,i}^\rmH\bY_{\rmw,i}(\rho\Psi(n)\bh_\rmw^\rmH\bh_\rmw+\bI_{N_\rmw})^{-1}+\bY_{\rmw,i}^\rmH\bY_{\rmw,i}\big]\nn\\*
&\quad-n\log\det\big(\rho\Psi(n)\bh_\rmw^\rmH\bh_\rmw+\bI_{N_\rmw}\big)\label{eq:KL divergence G 4}\\*
&=\sum_{i\in[n]}\Big(-\tr\big((\rho\Psi(n)\bh_\rmw^\rmH\bh_\rmw+\bI_{N_\rmw})(\rho\Psi(n)\bh_\rmw^\rmH\bh_\rmw+\bI_{N_\rmw})^{-1}\big)+\tr\big(\rho\Psi(n)\bh_\rmw^\rmH\bh_\rmw+\bI_{N_\rmw}\big)\Big)\nn\\*
&\quad-n\log\det\big(\rho\Psi(n)\bh_\rmw^\rmH\bh_\rmw+\bI_{N_\rmw}\big)\label{eq:KL divergence G 5}\\*
&=n\Big(\tr\big(\rho\Psi(n)\bh_\rmw^\rmH\bh_\rmw\big)-\log\det\big(\rho\Psi(n)\bh_\rmw^\rmH\bh_\rmw+\bI_{N_\rmw}\big)\Big)\label{eq:KL divergence G 6}\\
&<n\cdot\frac{(\rho\Psi(n))^2}{2}\tr\big((\bh_\rmw^\rmH\bh_\rmw)^2\big)\label{eq:KL divergence G 7}
\end{align}
where \eqref{eq:KL divergence G 1} follows from the definition of KL divergence, \eqref{eq:KL divergence G 2} follows since the outputs of each channel use are independent, \eqref{eq:KL divergence G 3} follows from the fact that $Q_{\bY_{\rmw}}^\mathrm{G}=\calCN(\mathbf{0},\rho\Psi(n)\bh_\rmw^\rmH\bh_\rmw+\bI_{N_\rmw})$ and $Q_{\bZ_\rmw}=\calCN(\mathbf{0},\bI_{N_\rmw})$, \eqref{eq:KL divergence G 4} follows from algebraic simplification, \eqref{eq:KL divergence G 5} follows from the fact that $\bbE_{Q^\rmG_{\bY_{\rmw}}}[\bY_{\rmw,i}^\rmH\bY_{\rmw,i}]=\tr(\rho\Psi(n)\bh_\rmw^\rmH\bh_\rmw+\bI_{N_\rmw})$, \eqref{eq:KL divergence G 6} follows since $\tr(\ba\ba^{-1})=\tr(\bI_{N_\rmw})={N_\rmw}$ for $\ba\in\bbC^{N_\rmw\times N_\rmw}$ and $\tr(\ba+\bb)=\tr(\ba)+\tr(\bb)$, and \eqref{eq:KL divergence G 7} follows from the Taylor Series of expansion of $\log(1+x)$ around $x=0$ which implies $\log(1+x)>x-\frac{x^2}{2}$.

Recall that $\rho\in(0,1)$, $\nu(n)>1$ and $\lim\limits_{n\to\infty}\nu(n)=1$. Combining \eqref{eq:KL divergence G 7} and Lemma \ref{lemma:awgn_ach_covert_GtG} leads that to satisfy the covertness constraint \eqref{eq4:covertness constraint} over $\calH_\rmw$, the maximal transmit power is given by
\begin{align}
\label{eq:ach power}
\Psi(n)&=\inf_{\bh_\rmw\in\calH_\rmw}\sqrt{\frac{2\delta}{n\rho^2\nu^2(n)\tr\big((\bh_\rmw^\rmH\bh_\rmw)^2\big)}}\\
&=\sqrt{\frac{2\delta}{n\rho^2\nu^2(n)\tr\big(\blambda_0^4\big)}},\label{eq:inf hw}
\end{align}
where \eqref{eq:inf hw} follows from the fact that $\bh_{\rmw,0}$ is the worst covert communication case satisfying $\bh_{\rmw,0}\bh_{\rmw,0}^\rmH=\blambda_0^2$.
Note that $\Psi(n)=O(\frac{1}{\sqrt{n}})$, it follows from the analysis in Appendix \ref{subapx:neighborhood 1} that the truncated complex Gaussian distribution with transmit power \eqref{eq:ach power} satisfies the covertness constraint \eqref{eq4:covertness constraint}.
The proof of Lemma~\ref{lemma:covert} is now completed.


\subsection{Proof of Lemma \ref{lemma:csit ach non} (Non-Asymptotic Achievability Bound for no-CSI)}
\label{appendix:csit ach non}
To facilitate the analysis, the codeword is now $\bX^n=(X_1^n,\ldots,X_{N_\rma}^n)\in\calX^{n\times N_\rma}$. Recall from Table \ref{tab:notations} that $\rho\in(0,1)$ is a parameter to constrain the power of codewords, which may depend on $n$. Define two sets
\begin{align}
\calF^{n\times N_\rma}&:=\big\{\bX^n:\forall~j\in[N_\rma],~\|X^n_j\|^2\leq n\Psi(n)\big\}\subseteq\calX^{n\times N_\rma},\label{eq:def fset}\\*
\bar{\calF}^{n\times N_\rma}&:=\big\{\bX^n:\forall~j\in[N_\rma],~\rho^2\cdot n\Psi(n)\leq\|X^n_j\|^2\leq n\Psi(n)\big\}\subseteq\calF^{n\times N_\rma}.\label{eq:def fbar}
\end{align}
For $j\in[N_\rma]$, $X^n_j$ is generated i.i.d. from $\calCN(\mathbf{0},\rho\Psi(n)\bI_n)$, and let $\bX^n\in\bar{\calF}^{n\times N_\rma}$. 
We consider a physically degraded version of the channel \eqref{eq:yb channel}, whose output is given by
\begin{align}
\label{eq:degraded channel}
\Omega_{\bY_\rmb^n}:=\spn(\bX^n\bH_\rmb+\bZ_\rmb^n).
\end{align}
Thus, the rate achievable on \eqref{eq:degraded channel} is a lower bound on the rate achievable on \eqref{eq:yb channel}. 
From the angle threshold decoding, given $\bX^n=\bx^n$, there is an associated threshold $\gamma_n\in\bbR_+$, and the decoder computes $\sin^2(\spn(\bx^n),\Omega_{\by_\rmb^n})$ for the received channel output $\Omega_{\by_\rmb^n}\in\bbC^{n\times N_\rmb}$ and selects the first codeword that satisfies $\sin^2(\spn(\bx^n),\Omega_{\by_\rmb^n})\leq\gamma_n$.

To obtain the no-CSI non-asymptotic achievability bound, we apply \cite[Corollary 1]{yu_finite_2021}, which obtained the achievability bound with the maximal error probability for random coding and input constraints, to the channel \eqref{eq:degraded channel}. The following performance metric is needed. Given two distributions $P$ and $Q$ on a common measurable space $\calY$, a randomized test between $P$ and $Q$ is defined as a random transformation $P_{S|Y}:\calY\mapsto\{0,1\}$ where $0$ indicates that the test chooses $Q$. Define the minimum error probability under hypothesis $Q$ if the error probability under hypothesis $P$ is not larger than $(1-\alpha)$ as
\begin{align}
\label{eq:def beta}
\beta_\alpha(P,Q):=\min\int P_{S|Y}(1|y)Q(\rmd y),
\end{align}
where the minimum is over all probability distributions $P_{S|Y}$ satisfying
\begin{align}
\label{eq:def beta min}
\int P_{S|Y}(1|y)P(\rmd y)\geq\alpha.
\end{align}
Let $P_{\Omega_{\bY_\rmb^n}|\bX^n}$ be the output distribution induced by the codeword $\bX^n$ defined on the subspace $\Omega_{\bY_\rmb^n}$.
By \cite[Corollary 1]{yu_finite_2021}, we have that for any $\varepsilon\in(0,1)$, $\tau\in(0,\varepsilon)$ and any auxiliary distribution $\tilP_{\Omega_{\bY_\rmb^n}}$ defined on the subspace $\Omega_{\bY_\rmb^n}$, the $(n,M)_\rmno$-code based on the above coding scheme whose codewords are chosen from $\bar{\calF}^{n\times N_\rma}$ defined in \eqref{eq:def fbar} with the maximal error probability $\varepsilon$ satisfying
\begin{align}
\label{eq:cor.1M}
M 
&\geq \sup_{\tau\in(0,\varepsilon)}\frac{\tau}{\sup_{\bx^n\in\bar{\calF}^{n\times N_\rma}}\beta_{1-\varepsilon+\tau}(P_{\Omega_{\bY_\rmb^n}|\bX^n=\bx^n},\tilP_{\Omega_{\bY_\rmb^n}})}\\*
&\geq \sup_{\tau\in(0,\varepsilon)}\frac{\tau}{\sup_{\bx^n\in\bar{\calF}^{n\times N_\rma}}\Pr_{\tilP_{\Omega_{\bY_\rmb^n}}}\big\{\sin^2(\spn(\bx^n),\Omega_{\bY_\rmb^n})\leq\gamma_n\big\}},\label{eq:beta and sin2}
\end{align}
where \eqref{eq:beta and sin2} follows from the angle threshold decoding and removing the minimization in \eqref{eq:def beta}, and $\gamma_n$ satisfies
\begin{align}
\label{eq:pr sin2 span}
\Pr_{P_{\Omega_{\bY_\rmb^n}|\bX^n=\bx^n}}\Big\{\sin^2\big(\spn(\bx^n),\Omega_{\bY_\rmb^n}\big)\leq\gamma_n\Big\}\geq 1-\varepsilon+\tau.
\end{align}
We choose the uniform distribution as the auxiliary output distribution $\tilP_{\Omega_{\bY_\rmb^n}}$, which is further denoted as $P_{\Omega_{\bY_\rmb^n}}^\rmu$. Under this choice, the output probabilities for every codewords are identical, i.e., $\Pr_{P_{\Omega_{\bY_\rmb^n}}^\rmu}\{\sin^2(\spn(\bx^n),\Omega_{\bY_\rmb^n})\leq\gamma_n\}$ does not depend on $\bx^n$. 
Since the sphere hardening effect \cite{hamkins_gaussian_2002} makes the codewords concentrate on the sphere as $n\to\infty$ and to simplify calculations, we choose $\bx^n=\bx_1^n:=\sqrt{n\rho\Psi(n)}\bI_{n,N_\rma}$. 
In the following, we first simplify \eqref{eq:beta and sin2} in Appendix \ref{subapx:csit ach non beta}, and then \eqref{eq:pr sin2 span} in Appendix \ref{subapx:csit ach non sin2}.

\subsubsection{Evaluation of \eqref{eq:beta and sin2}}
\label{subapx:csit ach non beta}
Given $\bX^n=\bx_1^n$, the term of $\sin^2(\cdot,\cdot)$ in \eqref{eq:beta and sin2} is now
\begin{align}
\sin^2\big(\spn(\bx^n),\Omega_{\bY_\rmb^n}\big)
&=\sin^2\big(\spn(\bx_1^n),\spn(\bx_1^n\bH_\rmb+\bZ_\rmb^n)\big)\\*
&=\sin^2\big(\spn(\bI_{n,N_\rma}),\spn(\bx_1^n\bH_\rmb+\bZ_\rmb^n)\big),\label{eq:sin2 span I}
\end{align}
where \eqref{eq:sin2 span I} follows since scaling by $\sqrt{n\rho\Psi(n)} \in \mathbb{R}_+$ does not change the span.
The complex multivariate Beta distribution is critical to analysis the distribution of the right hand side of \eqref{eq:sin2 span I}.
Fix a continuous alphabet $\calY\in\bbC$, two positive real numbers $(a,b)\in\bbR_+^2$, and an integer $p\in\bbN$. Recall that $\Gamma(\cdot)$ is the usual gamma function, and let $\tilde{\Gamma}_p(n):=\pi^\frac{p(p-1)}{2}\prod_{p'\in[p]}\Gamma(n-p'+1)$ be the complex multivariate gamma function. The pdf for complex multivariate Beta distribution with parameter $(a,b)$, denoted as $\tilde{\calB}_p(a,b)$, satisfies that for any Hermitian positive definite matrix $\bY\in\calY^{p\times p}$,
\begin{align}
f_{\bY}(\by):=\frac{\tilde{\Gamma}_p(a+b)}{\tilde{\Gamma}_p(a)\tilde{\Gamma}_p(b)}(\det \by)^{a-p}\det(\bI_p-\by)^{b-p},
\end{align}
for $a\geq p$, $b\geq p$, and $f_{\bY}(\by)=0$ elsewhere.
Define a matrix $\tilde{\bY}$ such that
\begin{align}
\label{eq:tilbY=sin2}
\det\tilde{\bY}=\sin^2\big(\spn(\bI_{n,N_\rma}),\spn(\bx_1^n\bH_\rmb+\bZ_\rmb^n)\big),
\end{align}
and we shall show that $\tilde{\bY}\sim\tilde{\calB}_{N_\rmb}(n-N_\rma,N_\rma)$ in the end of this part.
Thus, we obtain
\begin{align}
\label{eq:sin2 to det tilY}
\Pr_{P_{\Omega_{\bY_\rmb^n}}^\rmu}\Big\{\sin^2\big(\spn(\bx^n),\Omega_{\bY_\rmb^n}\big)\leq\gamma_n\Big\}=\Pr\Big\{\det\tilde{\bY}\leq\gamma_n\Big\}.
\end{align}
Recall that $\calB(\cdot,\cdot)$ is the Beta distribution defined in \eqref{eq:def of 1beta}. It follows from \cite[Corollary 1]{roh2006design} that the distribution of $\det\tilde{\bY}$ coincides with the distribution of $\prod_{j\in[N_\rmb]}B_j$, where $\{B_j\}_{j\in[N_\rmb]}$ are i.i.d. with $B_j\sim\calB(n-N_\rma-j+1,N_\rma)$. Combining with the right hand side of \eqref{eq:beta and sin2} and \eqref{eq:sin2 to det tilY} leads to
\begin{align}
\label{eq:M Beta_j}
M\geq \frac{\tau}{\Pr\Big\{\prod_{j\in[N_\rmb]}B_j\leq\gamma_n\Big\}}.
\end{align}
Therefore, the desired lower bound \eqref{eq:logM/n non ach} follows by taking the logarithm of both sides of \eqref{eq:M Beta_j}, and by dividing by the blocklength $n$.

\emph{Proof of $\tilde{\bY}\sim\tilde{\calB}_{N_\rmb}(n-N_\rma,N_\rma)$:}
Let $\otimes$ be the Kronecker product.
Fix a matrix $\bP_0\in\bbC^{p\times p}$, and two independent complex Gaussian random matrices $\bA\in\bbC^{a\times p}$ and $\bB\in\bbC^{b\times p}$, where $\bA\sim\calCN(\mathbf{0},\bI_a\otimes\bP_0)$ and $\bB\sim\calCN(\mathbf{0},\bI_b\otimes\bP_0)$. Let $\bP_1\in\bbC^{p\times p}$ be an upper triangular matrix and $\bP_2\in\bbC^{p\times p}$ be a Hermitian matrix such that $\bP_1^\rmH\bP_1=\bA^\rmH\bA+\bB^\rmH\bB$ and $\bP_1^\rmH\bP_2\bP_1=\bA^\rmH\bA$. It follows from \cite[Theorem 1]{roh2006design} that $\bP_2\sim\tilde{\calB}_p(a,b)$.
Based on the above method, we construct a complex multivariate Beta distribution. Let 
\begin{align}
\bI_{n,N_\rma}=
\begin{bmatrix}
\bI_{N_\rma} \\* \mathbf{0}_{(n-N_\rma)\times N_\rma}
\end{bmatrix}, \quad
\bY_0:=\bx_1^n\bH_\rmb+\bZ_\rmb^n=
\begin{bmatrix}
\bY_1 \\* \bY_2
\end{bmatrix},
\end{align}
where $\bY_1\in\bbC^{N_\rma\times N_\rmb}$ and $\bY_2\in\bbC^{(n-N_\rma)\times N_\rmb}$.
To obtain the principal angles, first orthonormalize $\bY_0$ to obtain
\begin{align}
\hat{\bY}_0=\bY_0(\bY_1^\rmH\bY_1+\bY_2^\rmH\bY_2)^{-\frac{1}{2}}.
\end{align}
Let $\{\theta_j\}_{j\in[N_\rma]}$ be the principal angles between $\spn(\bI_{n,N_\rma})$ and $\spn(\bx_1^n\bH_\rmb+\bZ_\rmb^n)$. From the definition of the principal angles in \eqref{7-decoder siso long}, it follows that
\begin{align}
\cos\theta_j=\max_{\ba\in\spn(\bI_{n,N_\rma}),~\bb\in\spn(\hat{\bY}_0)}\ba^\rmH\bb,
\end{align}
which implies that $\{\cos\theta_j\}_{j\in[N_\rma]}$ are the singular values of
\begin{align}
\label{eq:I hatY}
\bI_{n,N_\rma}^\rmH\hat{\bY}_0=\bY_1(\bY_1^\rmH\bY_1+\bY_2^\rmH\bY_2)^{-\frac{1}{2}}.
\end{align}
Thus, we obtain that $\{\cos^2\theta_j\}_{j\in[N_\rma]}$ are the eigenvalues of the matrix $(\bY_1^\rmH\bY_1+\bY_2^\rmH\bY_2)^{-1}\bY_1^\rmH\bY_1$. Since $(\bY_1^\rmH\bY_1+\bY_2^\rmH\bY_2)$ is a symmetric positive definite matrix, the Cholesky decomposition \cite[Page 112]{allaire2008numerical} guarantees that there exists an upper triangular matrix $\bY_3\in\bbC^{N_\rmb\times N_\rmb}$ with positive diagonal elements such that
\begin{align}
\label{eq:Y3HY3}
\bY_3^\rmH\bY_3=\bY_1^\rmH\bY_1+\bY_2^\rmH\bY_2.
\end{align}
Let $\bY_4\in\bbC^{N_\rmb\times N_\rmb}$ be the Hermitian matrix defined from
\begin{align}
\label{eq:def of Y4}
\bY_4:=\bY_3^{-\rmH}\bY_1^\rmH\bY_1\bY_3^{-1}.
\end{align}
It follows from \cite[Theorem 1]{roh2006design} that $\bY_4\sim\tilde{\calB}_{N_\rmb}(N_\rma,n-N_\rma)$. Combining \eqref{eq:I hatY}, \eqref{eq:Y3HY3} and \eqref{eq:def of Y4} leads that $\{\cos^2\theta_j\}_{j\in[N_\rma]}$ are the eigenvalues of $\bY_4$. Furthermore, the symmetry of the Beta function implies that $\{\sin^2\theta_j\}_{j\in[N_\rma]}$ have the same distribution as the eigenvalues of $\tilde{\bY}\sim\tilde{\calB}_{N_\rmb}(n-N_\rma,N_\rma)$, and the definition of the angle between the subspaces in \eqref{7 1:sin def} leads to \eqref{eq:tilbY=sin2}.

\subsubsection{Evaluation of \eqref{eq:pr sin2 span}}
\label{subapx:csit ach non sin2}

With a slight abuse of notation, we abbreviate $\sin(\spn(\cdot),\spn(\cdot))$ with $\sin(\cdot,\cdot)$. Let $\bzero_{a,b}$ be the all zero matrix of size $a\times b$. Recall from the SVD in \eqref{eq:svd Hb} that the left transformation matrix $\bL_\rmb\in\bbC^{N_\rma\times N_\rma}$ and the right transformation matrix $\bV'_\rmb\in\bbC^{N_\rmb\times N_\rmb}$ are unitary, and $\bSigma_\rmb=\begin{bmatrix}\bLambda_\rmb & \bzero_{N_\rma,(N_\rmb-N_\rma)} \end{bmatrix}$ is the singular value matrix, where $\bLambda_\rmb=\diag(\{\sqrt{\Lambda_{\rmb,j}}\}_{j\in[N_\rma]})$ contains singular values of $\bH_\rmb$.
Define an extended matrix of $\bL_\rmb$:
\begin{align}
\tilde{\bL}_\rmb:=
\begin{bmatrix}
\bL_\rmb^\rmH & \bzero_{(n-N_\rma),N_\rma} \\*
\bzero_{N_\rma,(n-N_\rma)} & \bI_{n-N_\rma}
\end{bmatrix}.
\end{align}
Given $\bX^n=\bx_1^n$, the left hand side of \eqref{eq:pr sin2 span} equals
\begin{align}
\Pr_{P_{\Omega_{\bY_\rmb^n}|\bX^n=\bx^n}}\Big\{\sin^2\big(\spn(\bx^n),\Omega_{\bY_\rmb^n}\big)\leq\gamma_n\Big\}
&=\Pr\Big\{\sin^2\big(\bI_{n,N_\rma},\sqrt{n\rho\Psi(n)}\bI_{n,N_\rma}\bH_\rmb+\bZ_\rmb^n\big)\leq\gamma_n\Big\}\label{eq:nocsit sin2 gamma -1}\\*
&=\Pr\Big\{\sin^2\big(\bI_{n,N_\rma}\bL_\rmb,(\sqrt{n\rho\Psi(n)}\bI_{n,N_\rma}\bH_\rmb+\bZ_\rmb^n)\bV'_\rmb \big) \leq {\gamma _n} \Big\}\label{eq:nocsit sin2 gamma -2}\\* 
&=\Pr\Big\{\sin^2\big(\tilde{\bL}_\rmb\bI_{n,N_\rma}\bL_\rmb,\tilde{\bL}_\rmb(\sqrt{n\rho\Psi(n)}\bI_{n,N_\rma}\bH_\rmb+\bZ_\rmb^n)\bV'_\rmb \big) \leq {\gamma _n} \Big\} \label{eq:nocsit sin2 gamma -3}\\*
&=\Pr\Big\{\sin^2\big(\bI_{n,N_\rma},\sqrt{n\rho\Psi(n)}\bI_{n,N_\rma}\bSigma_\rmb+\bZ_\rmb^n \big) \leq {\gamma _n}\Big\} \label{eq:nocsit sin2 gamma -4},
\end{align}
where \eqref{eq:nocsit sin2 gamma -1} follows from the fact that $\bx^n=\bx_1^n=\sqrt{n\rho\Psi(n)}\bI_{n,N_\rma}$ and $\Omega_{\bY_\rmb^n}=\spn(\bx_1^n\bH_\rmb+\bZ_\rmb^n)$, \eqref{eq:nocsit sin2 gamma -2} follows since $\spn(\bA)=\spn(\bA\bB)$ for each invertable matrix $\bB$ and $\bL_\rmb$, $\bV'_\rmb$ are both invertible, \eqref{eq:nocsit sin2 gamma -3} follows since the principal angles between two subspaces are invariant under simultaneous rotation of the two subspaces, and \eqref{eq:nocsit sin2 gamma -4} follows from \eqref{eq:svd Hb} and the fact that $\bZ_\rmb^n$ is isotropically distributed, implying that $\tilde{\bL}_\rmb\bZ_\rmb^n\bV'_\rmb$ has the same distribution as $\bZ_\rmb^n$.
Recall that $Z_{\rmb,j}^n$ is the $j$-th column of $\bZ_\rmb^n$, and let $\be_j$ stand for the $j$-th column of $\bI_{n,N_\rma}$. It follows from \cite[Lemma 13]{yang2014quasi} that
\begin{align}
\Pr\Big\{\sin^2\big(\bI_{n,N_\rma},\sqrt{n\rho\Psi(n)}\bI_{n,N_\rma}\bSigma_\rmb+\bZ_\rmb^n \big) \leq {\gamma _n}\Big\}
&\geq \Pr\Bigg\{\prod_{j\in [N_\rma]}\sin^2\big(\be_j,\sqrt{n\rho\Psi(n)\Lambda_{\rmb,j}}\be_j+Z_{\rmb,j}^n \big) \leq {\gamma _n}\Bigg\}\label{eq:sin2 toej-1}\\*
&=\Pr\Bigg\{\prod_{j\in [N_\rma]}\sin^2\big(\be_1,\sqrt{n\rho\Psi(n)\Lambda_{\rmb,j}}\be_1+Z_{\rmb,j}^n \big) \leq {\gamma _n}\Bigg\}\label{eq:sin2 toej-2},
\end{align}
where \eqref{eq:sin2 toej-2} follows by symmetry.
Note that $\be_1=(1,0,\ldots,0)$.
Let $\bA_1:=\be_1$, $\bA_2:=\sqrt{n\rho\Psi(n)\Lambda_{\rmb,j}}\be_1+Z_{\rmb,j}^n $, and $\bA_0:=[\bA_1,\bA_2]$. It follows from \cite[Lemma 12]{yang2014quasi} and by calculating that
\begin{align}
\sin^2\big(\be_1,\sqrt{n\rho\Psi(n)\Lambda_{\rmb,j}}\be_1+Z_{\rmb,j}^n \big)
&=\frac{\det\big(\bA_0^\rmH\bA_0\big)}{\det\big(\bA_1^\rmH\bA_1\big)\det\big(\bA_2^\rmH\bA_2\big)}\\*
&=\frac{\sum_{i\in[2,n]}\big|Z_{\rmb,ij}\big|^2}{\big|\sqrt{n\rho\Psi(n)\Lambda_{\rmb,j}}+Z_{\rmb,1j}\big|^2+\sum_{i\in[2,n]}\big|Z_{\rmb,ij}\big|^2},\label{eq:cal of Tj}
\end{align}
where \eqref{eq:cal of Tj} follows from the fact that $\det\big(\bA_1^\rmH\bA_1\big)=1$,
\begin{align}
\det\big(\bA_2^\rmH\bA_2\big)
&=\Big\|\sqrt{n\rho\Psi(n)\Lambda_{\rmb,j}}\be_1+Z_{\rmb,j}^n\Big\|^2\\*
&=\Big|\sqrt{n\rho\Psi(n)\Lambda_{\rmb,j}}+Z_{\rmb,1j}\Big|^2+\sum_{i\in[2,n]}\big|Z_{\rmb,ij}\big|^2,
\end{align}
and
\begin{align}
\det\big(\bA_0^\rmH\bA_0\big)
&=(\bA_1^\rmH\bA_1)(\bA_2^\rmH\bA_2) - \big|\bA_1^\rmH\bA_2\big|^2\\*
&=1\cdot \Big(\Big|\sqrt{n\rho\Psi(n)\Lambda_{\rmb,j}}+Z_{\rmb,1j}\Big|^2+\sum_{i\in[2,n]}\big|Z_{\rmb,ij}\big|^2\Big) - \Big|\sqrt{n\rho\Psi(n)\Lambda_{\rmb,j}}+Z_{\rmb,1j}\Big|\\*
&=\sum_{i\in[2,n]}\big|Z_{\rmb,ij}\big|^2.
\end{align}
Indeed, the right hand side of \eqref{eq:cal of Tj} is $T_j$ in \eqref{eq:def of Tj}. Combining \eqref{eq:pr sin2 span}, \eqref{eq:nocsit sin2 gamma -4}, \eqref{eq:sin2 toej-2} and \eqref{eq:cal of Tj} leads to \eqref{eq:sin2 pr 1-a}.

\subsection{Proof of Lemma \ref{lemma:csit ach reliability} (Asymptotic Achievability Bound for no-CSI)}
\label{appendix:csit ach asymptotic}

We start by analyzing the denominator on the right hand side of \eqref{eq:logM/n non ach}. Let $n_0=n-N_\rma-N_\rmb$. Thus,
\begin{align}
\Pr\bigg\{\prod_{j\in[N_\rmb]}B_j\leq\gamma_n\bigg\}
&= \Pr\bigg\{\prod_{j\in[N_\rmb]}B_j^{-n_0}\geq\gamma_n^{-n_0}\bigg\}\label{eq:Bj simplify-1}\\*
&\leq \frac{\bbE\Big[\prod_{j\in[N_\rmb]}B_j^{-n_0}\Big]}{\gamma_n^{-n_0}}\label{eq:Bj simplify-2}\\*
&= \gamma_n^{n_0}\prod_{j\in[N_\rmb]}\bbE\Big[B_j^{-n_0}\Big],\label{eq:Bj gamma Bj}
\end{align}
where \eqref{eq:Bj simplify-2} follows from the Markov inequality which implies that $\Pr\{X>a\}\leq\frac{\bbE[X]}{a}$ for any nonnegative random variable $X$ and any $a\in\bbR_+$, and \eqref{eq:Bj gamma Bj} follows since $\{B_j\}_{j\in[N_\rmb]}$ are independent.
Recall that $B_j\sim\calB(n-N_\rma-j+1,N_\rma)$, we obtain that for each $j\in[N_\rmb]$,
\begin{align}
\bbE\Big[B_j^{-n_0}\Big]
&= \frac{\Gamma(n-j+1)}{\Gamma(n-N_\rma-j+1)\Gamma(N_\rma)}\int_{0}^{1}s^{N_\rmb-j}(1-s)^{N_\rma-1}\rmd s\label{eq:Bj n^Na-1}\\*
&\leq \frac{\Gamma(n-j+1)}{\Gamma(n-N_\rma-j+1)\Gamma(N_\rma)}\label{eq:Bj n^Na-2}\\*
&\leq n^{N_\rma},\label{eq:Bj n^Na}
\end{align}
where \eqref{eq:Bj n^Na-1} follows from the definition of the Beta distribution in \eqref{eq:def of 1beta}, \eqref{eq:Bj n^Na-2} follows from the fact that $\int_0^1 s^{N_\rmb-j}(1-s)^{N_\rma-1}\rmd s\leq 1$, and \eqref{eq:Bj n^Na} follows from the fact that $\Gamma(a+1)=a\Gamma(a)$ and $\Gamma(b)\geq 1$ for $b\geq 1$.
Substituting \eqref{eq:Bj n^Na} into \eqref{eq:Bj gamma Bj}, we have
\begin{align}
\label{eq:Bj finally leq}
\Pr\bigg\{\prod_{j\in[N_\rmb]}B_j\leq\gamma_n\bigg\}\leq n^{N_\rma N_\rmb}\gamma_n^{n_0}.
\end{align}
Choose $\tau=O(\frac{1}{\sqrt{n}})$ and $\gamma_n=\exp(-C_{\varepsilon}(\Psi(n))+O(\frac{1}{n}))$, and we show in Appendix \ref{subapx:threshold} that this choice of $\gamma_n$ indeed satisfies \eqref{eq:sin2 pr 1-a}.
Combining with \eqref{eq:logM/n non ach} and \eqref{eq:Bj finally leq} leads to
\begin{align}
\frac{\log M}{n}
&\geq \frac{1}{n}\log\frac{\tau}{n^{N_\rma N_\rmb}\gamma_n^{n_0}}\\*
&=C_{\varepsilon}(\Psi(n))-N_\rma N_\rmb\frac{\log n}{n}+O\Big(\frac{\log n}{n}\Big)+O\Big(\frac{1}{n}\Big)+O\Big(\frac{1}{n^2}\Big)\\*
&=C_{\varepsilon}(\Psi(n))+O\Big(\frac{\log n}{n}\Big).
\end{align}
The proof of Lemma \ref{lemma:csit ach reliability} is now completed.

\subsubsection{Threshold Design}
\label{subapx:threshold}
In this part, we shall show that $\gamma_n=\exp(-C_{\varepsilon}(\Psi(n))+O(\frac{1}{n}))$ satisfies \eqref{eq:sin2 pr 1-a}.
Recall that $T_j$ is defined in \eqref{eq:def of Tj}.
The following lemma is critical to verify the reliability of the chosen decoding threshold.
\begin{lemma}
\label{lemma:18ach}
Fix an arbitrary $\xi_0\in(0,1)$. There exists an $\iota>0$ such that
\begin{align}
\label{eq:lemma 18}
\inf_{\xi\in(\xi_0-\iota,\xi_0+\iota)}\Bigg(\Pr\bigg\{\prod_{j\in[N_\rma]}T_j\leq\xi\bigg\}-\Pr\bigg\{\prod_{j\in[N_\rma]}\frac{1}{1+\rho\Psi(n)\Lambda_{\rmb,j}}\leq\xi\bigg\}\Bigg)>O\Big(\frac{1}{\sqrt{n}}\Big).
\end{align}
\end{lemma}
The proof of Lemma \ref{lemma:18ach} is provided in Appendix \ref{appendix:lemma:18ach}.

For any $R\in\bbR_+$, it follows from \eqref{eq:ach C final def} that the covert outage probability is given by
\begin{align}
\label{eq:def of Fout}
F_\rmout(R):=\Pr\Big\{ \log\det \big( \bI_{N_\rmb}+ \rho\Psi(n)\bH_\rmb^\rmH\bH_\rmb \big)<R \Big\}.
\end{align}
The covert outage probability can be interpreted as the probability that the channel gain $\bH_\rmb$ is too weak to support reliable communication at any target rate $R$.

Recall that we choose $\tau=O(\frac{1}{\sqrt{n}})$ and $\gamma_n=\exp(-C_{\varepsilon}(\Psi(n))+O(\frac{1}{n}))$. From Lemma \ref{lemma:18ach}, there exists an $\iota'>0$ such that for any $\gamma_n\in(e^{-C_{\varepsilon}(\Psi(n))-\iota'},e^{-C_{\varepsilon}(\Psi(n))+\iota'})$,
\begin{align}
\Pr\bigg\{\prod_{j\in[N_\rma]}T_j\leq\gamma_n\bigg\}
&\geq \Pr\bigg\{\prod_{j\in[N_\rma]}\frac{1}{1+\rho\Psi(n)\Lambda_{\rmb,j}}\leq\gamma_n\bigg\} + O\Big(\frac{1}{\sqrt{n}}\Big)\label{eq:use lem18-1}\\*
&= 1-\Pr\bigg\{\sum_{j\in[N_\rma]}\log\big(1+\rho\Psi(n)\Lambda_{\rmb,j}\big)\leq-\log\gamma_n\bigg\}+ O\Big(\frac{1}{\sqrt{n}}\Big)\label{eq:use lem18-2}\\*
&= 1-F_\rmout(-\log\gamma_n)+O\Big(\frac{1}{\sqrt{n}}\Big)\label{eq:use lem18-3}\\*
&=1-\varepsilon-F'_\rmout(C_{\varepsilon}(\Psi(n)))(-\log\gamma_n-C_{\varepsilon}(\Psi(n)))+O\Big(\frac{1}{\sqrt{n}}\Big)\label{eq:use lem18-4}\\*
&=1-\varepsilon+O\Big(\frac{1}{\sqrt{n}}\Big),\label{eq:use lem18-5}
\end{align}
where \eqref{eq:use lem18-2} follows by taking the logarithm inside the probability, \eqref{eq:use lem18-3} follows from the definition of the covert outage probability $F_\rmout(\cdot)$ in \eqref{eq:def of Fout}, \eqref{eq:use lem18-4} follows from the Taylor Series of expansion of $F_\rmout(-\log\gamma_n)$ around the covert outage rate $C_{\varepsilon}(\Psi(n))$ which implies
\begin{align}
\label{eq:Ft taylor}
F_\rmout(-\log\gamma_n)=F_\rmout\big(C_{\varepsilon}(\Psi(n))\big)+F'_\rmout\big(C_{\varepsilon}(\Psi(n))\big)\big(-\log\gamma_n-C_{\varepsilon}(\Psi(n))\big)+O\big(\big(-\log\gamma_n-C_{\varepsilon}(\Psi(n))\big)^2\big),
\end{align}
and the fact that $F_\rmout(C_{\varepsilon}(\Psi(n)))=\varepsilon$ and $O((-\log\gamma_n-C_{\varepsilon}(\Psi(n)))^2)=O(n^{-2})$, and \eqref{eq:use lem18-5} follows from the fact that $F'_\rmout(C_{\varepsilon}(\Psi(n)))=O(\sqrt{n})$, whose proof is deferred to Appendix \ref{subapx:F' calculation}.
Thus, by \eqref{eq:use lem18-5}, this choice of $\gamma_n=\exp(-C_{\varepsilon}(\Psi(n))+O(\frac{1}{n}))$ indeed satisfies \eqref{eq:sin2 pr 1-a}. The proof is completed.

\subsubsection{Proof of $F'_\rmout(C_{\varepsilon}(\Psi(n)))=O(\sqrt{n})$}
\label{subapx:F' calculation}
For ease of observing the asymptotic behavior with respect to $n$, we simplify the MIMO channel as $N_\rma$ i.i.d. sub-channels sharing the same random channel gain $\Lambda_\rmb$. Recall that $\rho\in(0,1)$ and $\Psi(n)=O(\frac{1}{\sqrt{n}})$. From the definition of the covert outage rate in \eqref{eq:ach C final def}, it follows that
\begin{align}
C_{\varepsilon}(\Psi(n))
&=\sup\Big\{R:\Pr\big\{N_\rma\cdot\log\big(1+\rho\Psi(n)\Lambda_\rmb\big)<R\big\}\leq \varepsilon\Big\}.\label{eq:iid C def}
\end{align}
Recall that the $a$-quantile function of the random variable $X$ is defined as
\begin{align}
\calQ_{X}(a):= \inf\Big\{x: \Pr\big\{X\leq x\big\} \geq a\Big\}.
\end{align}
Let $\tilde{\Lambda}_\rmb:=N_\rma\log(1+\rho\Psi(n)\Lambda_\rmb)$, and the right hand side of \eqref{eq:iid C def} can be rewritten as
\begin{align}
C_{\varepsilon}(\Psi(n))
&=\calQ_{\tilde{\Lambda}_\rmb}(\varepsilon)=N_\rma\rho\Psi(n)\cdot \calQ_{\Lambda_\rmb}(\varepsilon)+o(\Psi(n))\label{eq:C order final-1}\\*
&=O(n^{-\frac{1}{2}}),\label{eq:C order final-2}
\end{align}
where the second equation in \eqref{eq:C order final-1} follows from the Taylor Series of expansion of $\log(1+x)$ around $x=0$ which implies $\log(1+x)=x+o(x)$ and the fact that $\tilde{\Lambda}_\rmb$ is monotonic and continuous respect to $\Lambda_\rmb$, and \eqref{eq:C order final-2} follows from the fact that $\calQ_{\Lambda_\rmb}(\varepsilon)=O(1)$.
Next, from the definition of the covert outage probability in \eqref{eq:def of Fout}, it follows that
\begin{align}
F_\rmout(R)=\Pr\big\{N_\rma\cdot\log\big(1+\rho\Psi(n)\Lambda_\rmb\big)<R\big\}.
\end{align}
Thus, the derivative of $F_\rmout(R)$ is indeed the pdf of $\tilde{\Lambda}_\rmb$ defined as $f_{\tilde{\Lambda}_\rmb}(R)$. Let $f_{\Lambda_\rmb}$ be the pdf of $\Lambda_\rmb$, it follows that
\begin{align}
\label{eq:F' to pdf}
F'_\rmout(R)
=f_{\tilde{\Lambda}_\rmb}(R)
=\frac{e^{\frac{R}{N_\rma}}}{N_\rma\rho\Psi(n)}\cdot f_{\Lambda_\rmb}\bigg(\frac{e^{\frac{R}{N_\rma}}-1}{\rho\Psi(n)}\bigg),
\end{align}
where the second equation in \eqref{eq:F' to pdf} follows from the fact that $\Lambda_\rmb=\frac{e^{\tilde{\Lambda}_\rmb/N_\rma}-1}{\rho\Psi(n)}$.
Given $R=C_{\varepsilon}(\Psi(n))=O(n^{-\frac{1}{2}})$, we obtain
\begin{align}
F'_\rmout(C_{\varepsilon}(\Psi(n)))
&=\frac{e^{\frac{C_{\varepsilon}(\Psi(n))}{N_\rma}}}{N_\rma\rho\Psi(n)}\cdot f_{\Lambda_\rmb}\bigg(\frac{e^{\frac{C_{\varepsilon}(\Psi(n))}{N_\rma}}-1}{\rho\Psi(n)}\bigg)\label{eq:F' final-1}\\*
&=\frac{1+\frac{C_{\varepsilon}(\Psi(n))}{N_\rma}+o\big(\frac{C_{\varepsilon}(\Psi(n))}{N_\rma}\big)}{N_\rma\rho\Psi(n)}\cdot f_{\Lambda_\rmb}\bigg(\frac{\frac{C_{\varepsilon}(\Psi(n))}{N_\rma}+o\big(\frac{C_{\varepsilon}(\Psi(n))}{N_\rma}\big)}{\rho\Psi(n)}\bigg)\label{eq:F' final-2}\\*
&=O(\sqrt{n}),\label{eq:F' final-3}
\end{align}
where \eqref{eq:F' final-2} follows from the Taylor Series of expansion of $e^x$ around $x=0$ which implies $e^x=1+x+o(x)$, and \eqref{eq:F' final-3} follows from the fact that $\Psi(n)=O(\frac{1}{\sqrt{n}})$ and $f_{\Lambda_\rmb}(\cdot)=O(1)$ for typical values.

\subsection{Proof of Lemma \ref{lemma:18ach}}
\label{appendix:lemma:18ach}

Recall that $\xi_0\in(0,1)$ and $T_j$ is defined in \eqref{eq:def of Tj}. Choose $\iota>0$ such that $\iota\leq\frac{\xi_0}{2}$. Recall that $\xi\in(\xi_0-\iota,\xi_0+\iota)$. Define $\tilT:=\frac{\xi}{\prod_{j\in[2,N_\rma]}T_j}$. The first term in the left hand side of \eqref{eq:lemma 18} can be written as
\begin{align}
\Pr\bigg\{\prod_{j\in[N_\rma]}T_j\leq\xi\bigg\}
&= \Pr\big\{T_1\leq\tilT\big\}\\*
&= \Pr\big\{T_1\leq\tilT,\tilT\geq 1\big\}+\Pr\big\{T_1\leq\tilT,\tilT< 1\big\}\\*
&= \Pr\big\{\tilT\geq 1\big\}+\Pr\big\{T_1\leq\tilT,\tilT< 1\big\},\label{eq:+tilT 1}
\end{align}
where \eqref{eq:+tilT 1} follows from the fact that $T_1\leq 1$ by \eqref{eq:def of Tj} and further $\Pr\big\{T_1\leq\tilT|\tilT\geq 1\big\}=1$. Define $G_j:=\rho\Psi(n)\Lambda_{\rmb,j}$. We mainly focus on the second term of the right hand side of \eqref{eq:+tilT 1}:
\begin{align}
\label{eq:bbE of tilT}
\Pr\big\{T_1\leq\tilT,\tilT< 1\big\}=\bbE_{\tilT,G_2,\ldots,G_{N_\rma}}\Big[\bbo(\tilT<1)\times \Pr\big\{T_1\leq\tilT|\tilT,G_2,\ldots,G_{N_\rma}\big\}\Big].
\end{align}
Since events of measure zero do not affect \eqref{eq:bbE of tilT}, we can
assume without loss of generality that the joint pdf of $\tilT,G_2,\ldots,G_{N_\rma}$ is strictly positive.
For $j=1$, let $Z_{\rmb,(n+1)1}\sim\calCN(0,1)$ be independent of all other random variables appearing in the definition of $\{T_j\}_{j\in[N_\rma]}$ in \eqref{eq:def of Tj}.
To lower-bound \eqref{eq:bbE of tilT}, we first bound the second term inside the expectation. For $\tilT<1$, it follows that
\begin{align}
\Pr\big\{T_1\leq\tilT|\tilT,G_2,\ldots,G_{N_\rma}\big\}
&= \Pr\Bigg\{\frac{\sum_{i\in[2,n]}\big|Z_{\rmb,i1}\big|^2}{\big|\sqrt{nG_1}+Z_{\rmb,11}\big|^2+\sum_{i\in[2,n]}\big|Z_{\rmb,i1}\big|^2}\leq\tilT\Bigg|\tilT,G_2,\ldots,G_{N_\rma}\Bigg\}\label{eq:before B sqrtn A-1}\\*
&= \Pr\Bigg\{\big|\sqrt{nG_1}+Z_{\rmb,11}\big|^2\geq(\tilT^{-1}-1)\sum_{i\in[2,n]}\big|Z_{\rmb,i1}\big|^2\Bigg|\tilT,G_2,\ldots,G_{N_\rma}\Bigg\}\label{eq:before B sqrtn A-2}\\*
&\geq \Pr\Bigg\{\big|\sqrt{nG_1}+\Re\big(Z_{\rmb,11}\big)\big|^2\geq(\tilT^{-1}-1)\sum_{i\in[2,n+1]}\big|Z_{\rmb,i1}\big|^2\Bigg|\tilT,G_2,\ldots,G_{N_\rma}\Bigg\}\label{eq:before B sqrtn A-3}\\*
&\geq \Pr\Bigg\{\sqrt{nG_1}\geq -\Re\big(Z_{\rmb,11}\big)+\sqrt{\tilT^{-1}-1}\sqrt{\sum_{i\in[2,n+1]}\big|Z_{\rmb,i1}\big|^2}\Bigg|\tilT,G_2,\ldots,G_{N_\rma}\Bigg\},\label{eq:before B sqrtn A}
\end{align}
where \eqref{eq:before B sqrtn A-1} follows from \eqref{eq:def of Tj}, \eqref{eq:before B sqrtn A-3} follows since $|\sqrt{nG_1}+\Re(Z_{\rmb,11})|^2\leq|\sqrt{nG_1}+Z_{\rmb,11}|^2$ and $\sum_{i\in[2,n+1]}|Z_{\rmb,i1}|^2\geq\sum_{i\in[2,n]}|Z_{\rmb,i1}|^2$.
We shall need the following lemma, which concerns the speed of convergence of $\Pr\big\{B\geq\frac{A}{\sqrt{n}}\big\}$ to $\Pr\big\{B\geq0\big\}$ as $n\to\infty$ for two independent random variables $A$ and $B$.
\begin{lemma}
\label{lemma:B sqrtn A ff'}
Let $A$ be a real random variable with zero mean and unit variance, and $B$ be a real random variable independent of $A$ with continuously differentiable pdf $f_B$, and its derivative is denoted by $f_B'$. Let $\iota>0$ be independent with $n$. Thus,
\begin{align}
\label{eq:B sqrtn A ff'}
\bigg|\Pr\bigg\{B\geq\frac{A}{\sqrt{n}}\bigg\}-\Pr\big\{B\geq 0\big\}\bigg|\leq\frac{2}{\iota^2n}+\frac{f_B}{\iota n}+\frac{f_B'}{2n}.
\end{align}
\end{lemma}
The proof of Lemma \ref{lemma:B sqrtn A ff'} follows by tailoring \cite[Lemma 17]{yang2014quasi} to preserve the terms of $f_B$ and $f'_B$.

We next rewrite the right hand side of \eqref{eq:before B sqrtn A} into the form of the left hand side of \eqref{eq:B sqrtn A ff'} to use Lemma \ref{lemma:B sqrtn A ff'}.
Recall that $G_j=\rho\Psi(n)\Lambda_{\rmb,j}$ for $j\in[N_\rma]$.
Let $\mu_Z$ and $\sigma_Z^2$ be the mean and the variance of the random variable $\sqrt{\sum_{i\in[2,n+1]}\big|Z_{\rmb,i1}\big|^2}$, respectively. Define $\tilT_\dagger:=\sqrt{\tilT^{-1}-1}$. Below, we construct the forms of $A$ and $B$ in Lemma \ref{lemma:B sqrtn A ff'}. Define two random variables
\begin{align}
K_1&:=\frac{1}{\sqrt{\frac{1}{2}+\tilT_\dagger^2\sigma_Z^2}}\Bigg(-\Re\big(Z_{\rmb,11}\big)+\tilT_\dagger\sqrt{\sum_{i\in[2,n+1]}\big|Z_{\rmb,i1}\big|^2}-\mu_Z\tilT_\dagger\Bigg),\label{eq:def K ach}\\*
\tilG_1&:=\frac{1}{\sqrt{\frac{1}{2}+\tilT_\dagger^2\sigma_Z^2}}\bigg(\sqrt{G_1}-\frac{\mu_Z}{\sqrt{n}}\tilT_\dagger\bigg).\label{eq:tilG1}
\end{align}
Note that $K_1$ is a zero mean, unit variance random variable that is conditionally independent of $\tilG_1$ given $\tilT_\dagger$. 
Using these definitions, we can rewrite the right hand side of \eqref{eq:before B sqrtn A} as
\begin{align}
\label{eq:rewrite nG1}
\Pr\bigg\{\tilG_1\geq\frac{K_1}{\sqrt{n}}\bigg|\tilT_\dagger,G_2,\ldots,G_{N_\rma}\bigg\}.
\end{align}
Let $f_{\tilG_1}$ be the pdf of $\tilG_1$ and $f'_{\tilG_1}$ be the derivative of $f_{\tilG_1}$.\footnote{The definitions of the pdfs, joint pdfs, marginal pdfs, and conditional pdfs for random variables are similar as $f_{\tilG_1}$ and thus omitted below.}
From Lemma \ref{lemma:B sqrtn A ff'}, there exists an $\iota>0$ such that
\begin{align}
\label{eq:tilG1>0}
\Pr\bigg\{\tilG_1\geq\frac{K_1}{\sqrt{n}}\bigg|\tilT_\dagger,G_2,\ldots,G_{N_\rma}\bigg\}
\geq \Pr\Big\{\tilG_1\geq0\Big|\tilT_\dagger,G_2,\ldots,G_{N_\rma}\Big\}-\frac{1}{n}\bigg(\frac{2}{\iota^2}+\frac{f_{\tilG_1}}{\iota}+\frac{f'_{\tilG_1}}{2}\bigg).
\end{align}
Combining \eqref{eq:bbE of tilT}, \eqref{eq:rewrite nG1} and \eqref{eq:tilG1>0} leads to
\begin{align}
\Pr\big\{T_1\leq\tilT,\tilT< 1\big\}
&\geq \bbE_{\tilT,G_2,\ldots,G_{N_\rma}}\Bigg[\bbo(\tilT<1)\times \bigg(\Pr\Big\{\tilG_1\geq0\Big|\tilT,G_2,\ldots,G_{N_\rma}\Big\}-\frac{1}{n}\bigg(\frac{2}{\iota^2}+\frac{f_{\tilG_1}}{\iota}+\frac{f'_{\tilG_1}}{2}\bigg)\bigg)\Bigg]\label{eq:combine bbE 1}\\*
&\geq \Pr\bigg\{\frac{1}{1+\frac{nG_1}{\mu^2_Z}}\leq \tilT,\tilT<1\bigg\}-\frac{1}{n}\bigg(\frac{2}{\iota^2}+\frac{f_{\tilG_1}}{\iota}+\frac{f'_{\tilG_1}}{2}\bigg)\label{eq:combine bbE 2}\\*
&\geq \Pr\bigg\{\frac{1}{1+G_1}\leq \tilT,\tilT<1\bigg\}-\frac{1}{n}\bigg(\frac{2}{\iota^2}+\frac{f_{\tilG_1}}{\iota}+\frac{f'_{\tilG_1}}{2}\bigg),\label{eq:combine bbE 3}
\end{align}
where \eqref{eq:combine bbE 2} follows from \eqref{eq:tilG1} and the fact that $\tilT_\dagger=\sqrt{\tilT^{-1}-1}$ and $\bbo(\tilT<1)\leq 1$, and \eqref{eq:combine bbE 3} follows from \cite[Eq. (18.14)]{johnson1995continuous} that $\mu_Z=\frac{\Gamma(n+\frac{1}{2})}{\Gamma(n)}\leq\sqrt{n}$.
Substituting \eqref{eq:combine bbE 3} into \eqref{eq:+tilT 1}, we obtain
\begin{align}
\Pr\bigg\{\prod_{j\in[N_\rma]}T_j\leq\xi\bigg\}
&\geq \Pr\big\{\tilT\geq 1\big\}+\Pr\bigg\{\frac{1}{1+G_1}\leq \tilT,\tilT<1\bigg\}-\frac{1}{n}\bigg(\frac{2}{\iota^2}+\frac{f_{\tilG_1}}{\iota}+\frac{f'_{\tilG_1}}{2}\bigg)\label{eq:Tj 1+Gj final-1}\\*
&=\Pr\bigg\{\frac{1}{1+G_1}\leq \tilT,\tilT\geq 1\bigg\}+\Pr\bigg\{\frac{1}{1+G_1}\leq \tilT,\tilT<1\bigg\}-\frac{1}{n}\bigg(\frac{2}{\iota^2}+\frac{f_{\tilG_1}}{\iota}+\frac{f'_{\tilG_1}}{2}\bigg)\label{eq:Tj 1+Gj final-2}\\*
&=\Pr\bigg\{\frac{1}{1+G_1}\leq \tilT\bigg\}-\frac{1}{n}\bigg(\frac{2}{\iota^2}+\frac{f_{\tilG_1}}{\iota}+\frac{f'_{\tilG_1}}{2}\bigg)\label{eq:Tj 1+Gj final-3}\\*
&=\Pr\bigg\{\frac{1}{1+G_1}\prod_{j\in[2,N_\rma]}T_j\leq\xi\bigg\}-\frac{1}{n}\bigg(\frac{2}{\iota^2}+\frac{f_{\tilG_1}}{\iota}+\frac{f'_{\tilG_1}}{2}\bigg),\label{eq:Tj 1+Gj final-4}
\end{align}
where \eqref{eq:Tj 1+Gj final-2} follows from the fact that $\frac{1}{1+G_1}\leq 1$, and \eqref{eq:Tj 1+Gj final-4} follows from the fact that $\tilT=\frac{\xi}{\prod_{j\in[2,N_\rma]}T_j}$.
Note that $f_{\tilG_1}=O(n^\frac{1}{4})$ and $f_{\tilG_1}'=O(\sqrt{n})$, whose proofs are deferred to Appendices \ref{subapx:f tilG 1} and \ref{subapx:f tilG 1 prime}, respectively. Reiterating above steps for $\{T_j\}_{j\in[2,N_\rma]}$, we conclude that
\begin{align}
\Pr\bigg\{\prod_{j\in[N_\rma]}T_j\leq\xi\bigg\}
&\geq \Pr\bigg\{\prod_{j\in[N_\rma]}\frac{1}{1+G_j}\leq\xi\bigg\}+O\Big(\frac{1}{\sqrt{n}}\Big).
\end{align}
The proof of Lemma \ref{lemma:18ach} is completed.
To conclude the proof, it remains to calculate $f_{\tilG_1}$ and $f_{\tilG_1}'$.

\subsubsection{Proof of $f_{\tilG_1}=O(n^\frac{1}{4})$}
\label{subapx:f tilG 1}
Recall that $T_j$ is defined in \eqref{eq:def of Tj}, $\tilT=\frac{\xi}{\prod_{j\in[2,N_\rma]}T_j}$, and $\tilT_\dagger=\sqrt{\tilT^{-1}-1}$.
By the definition of $\tilG_1$ in \eqref{eq:tilG1}, we have
\begin{align}
G_1=\bigg(\sqrt{\frac{1}{2}+\tilT_\dagger^2\sigma_Z^2}\tilG_1+\frac{\mu_Z}{\sqrt{n}}\tilT_\dagger\bigg)^2.
\end{align}
Given $\tilT_\dagger=\tilt_\dagger$ and $\{G_j=g_j\}_{j\in[2,N_\rma]}$. 
Let $\tilg_1$ be the realization of $\tilG_1$, and $g_1:=\big(\sqrt{\frac{1}{2}+\tilt_\dagger^2\sigma_Z^2}\tilg_1+\frac{\mu_Z}{\sqrt{n}}\tilt_\dagger\big)^2$. Since $G_1$ increases monotonically with respect to $\tilG_1$, it follows that
\begin{align}
f_{\tilG_1}(\tilg_1|\tilt_\dagger,g_2,\ldots,g_{N_\rma})
&=f_{G_1}(g_1|g_2,\ldots,g_{N_\rma})\cdot\frac{\rmd g_1}{\rmd \tilg_1} \label{eq:tilG1 conditional-1}\\*
&=f_{G_1}(g_1|g_2,\ldots,g_{N_\rma})\cdot2\bigg(\sqrt{\frac{1}{2}+\tilt_\dagger^2\sigma_Z^2}\tilg_1+\frac{\mu_Z}{\sqrt{n}}\tilt_\dagger\bigg)\sqrt{\frac{1}{2}+\tilt_\dagger^2\sigma_Z^2}\label{eq:tilG1 conditional-2}\\*
&=\frac{f_{G_1,\ldots,G_{N_\rma}}(g_1,\ldots,g_{N_\rma})}{f_{G_2,\ldots,G_{N_\rma}}(g_2,\ldots,g_{N_\rma})}\cdot2\bigg(\sqrt{\frac{1}{2}+\tilt_\dagger^2\sigma_Z^2}\tilg_1+\frac{\mu_Z}{\sqrt{n}}\tilt_\dagger\bigg)\sqrt{\frac{1}{2}+\tilt_\dagger^2\sigma_Z^2}.\label{eq:tilG1 conditional-3}
\end{align}
Recall that $G_j=\rho\Psi(n)\Lambda_{\rmb,j}$, and $f_{\Lambda_{\rmb,1},\ldots,\Lambda_{\rmb,N_\rma}}$ is the joint pdf of $\Lambda_{\rmb,1},\ldots,\Lambda_{\rmb,N_\rma}$, it follows that
\begin{align}
\label{eq:g1 joint}
f_{G_1,\ldots,G_{N_\rma}}(g_1,\ldots,g_{N_\rma})=\frac{1}{(\rho\Psi(n))^{N_\rma}}\cdot f_{\Lambda_{\rmb,1},\ldots,\Lambda_{\rmb,N_\rma}}(\lambda_{\rmb,1},\ldots,\lambda_{\rmb,N_\rma}).
\end{align}
For the marginal pdf $f_{G_2,\ldots,G_{N_\rma}}$, integrating over the entire domain of $G_1$ leads to
\begin{align}
f_{G_2,\ldots,G_{N_\rma}}(g_2,\ldots,g_{N_\rma})
&=\int_{G_1\in\bbR_+}f_{G_1,\ldots,G_{N_\rma}}(g_1,\ldots,g_{N_\rma})\rmd g_1\label{eq:g2-1}\\*
&=\int_{\Lambda_{\rmb,1}\in\bbR_+}\frac{f_{\Lambda_{\rmb,1},\ldots,\Lambda_{\rmb,N_\rma}}(\lambda_{\rmb,1},\ldots,\lambda_{\rmb,N_\rma})}{(\rho\Psi(n))^{N_\rma}}\cdot \rho\Psi(n)\rmd\lambda_{\rmb,1}\label{eq:g2-2}\\*
&=\frac{1}{(\rho\Psi(n))^{N_\rma-1}}\int_{\Lambda_{\rmb,1}\in\bbR_+}f_{\Lambda_{\rmb,1},\ldots,\Lambda_{\rmb,N_\rma}}(\lambda_{\rmb,1},\ldots,\lambda_{\rmb,N_\rma})\rmd\lambda_{\rmb,1}\label{eq:g2-3}\\*
&=\frac{1}{(\rho\Psi(n))^{N_\rma-1}}\cdot f_{\Lambda_{\rmb,2},\ldots,\Lambda_{\rmb,N_\rma}}(\lambda_{\rmb,2},\ldots,\lambda_{\rmb,N_\rma}).\label{eq:g2-4}
\end{align}
Substituting \eqref{eq:g1 joint} and \eqref{eq:g2-4} into \eqref{eq:tilG1 conditional-3} leads to
\begin{align}
f_{\tilG_1}(\tilg_1|\tilt_\dagger,g_2,\ldots,g_{N_\rma})
&=\frac{f_{\Lambda_{\rmb,1}|\Lambda_{\rmb,2},\ldots,\Lambda_{\rmb,N_\rma}}(\lambda_{\rmb,1}|\lambda_{\rmb,2},\ldots,\lambda_{\rmb,N_\rma})}{\rho\Psi(n)}\cdot 2\bigg(\sqrt{\frac{1}{2}+\tilt_\dagger^2\sigma_Z^2}\tilg_1+\frac{\mu_Z}{\sqrt{n}}\tilt_\dagger\bigg)\sqrt{\frac{1}{2}+\tilt_\dagger^2\sigma_Z^2}\label{eq:tilG_1 final-1}\\*
&=O(n^{\frac{1}{4}}),\label{eq:tilG_1 final-2}
\end{align}
where \eqref{eq:tilG_1 final-2} follows since $f_{\Lambda_{\rmb,1}|\Lambda_{\rmb,2},\ldots,\Lambda_{\rmb,N_\rma}}$ is independent with $n$, $\Psi(n)=O(\frac{1}{\sqrt{n}})$, $\sqrt{g_1}=\sqrt{\frac{1}{2}+\tilt_\dagger^2\sigma_Z^2}\tilg_1+\frac{\mu_Z}{\sqrt{n}}\tilt_\dagger=O(n^{-\frac{1}{4}})$, and $\tilt_\dagger^2\sigma_Z^2$ is bounded by $O(1)$ by \cite[Eqs. (315)--(319)]{yang2014quasi}.

\subsubsection{Proof of $f_{\tilG_1}'=O(\sqrt{n})$}
\label{subapx:f tilG 1 prime}
Define two functions 
\begin{align}
F_1&:=\frac{1}{\rho\Psi(n)}\cdot f_{\Lambda_{\rmb,1}|\Lambda_{\rmb,2},\ldots,\Lambda_{\rmb,N_\rma}}(\lambda_{\rmb,1}|\lambda_{\rmb,2},\ldots,\lambda_{\rmb,N_\rma}),\label{eq:F1 def}\\*
F_2&:=2\bigg(\sqrt{\frac{1}{2}+\tilt_\dagger^2\sigma_Z^2}\tilg_1+\frac{\mu_Z}{\sqrt{n}}\tilt_\dagger\bigg)\sqrt{\frac{1}{2}+\tilt_\dagger^2\sigma_Z^2},\label{eq:F2 def}
\end{align}
and by \eqref{eq:tilG_1 final-1}, the derivative of $f_{\tilG_1}$ equals
\begin{align}
\label{eq:tilg1 derivative}
\frac{\rmd}{\rmd \tilg_1}f_{\tilG_1}(\tilg_1|\tilt_\dagger,g_2,\ldots,g_{N_\rma})=F'_1F_2+F_1F'_2.
\end{align}
By calculating, we obtain
\begin{align}
F'_1
&=\frac{f'_{\Lambda_{\rmb,1}|\Lambda_{\rmb,2},\ldots,\Lambda_{\rmb,N_\rma}}(\lambda_{\rmb,1}|\lambda_{\rmb,2},\ldots,\lambda_{\rmb,N_\rma})}{(\rho\Psi(n))^2}\cdot \frac{\rmd g_1}{\rmd\tilg_1}\\*
&=\frac{f'_{\Lambda_{\rmb,1}|\Lambda_{\rmb,2},\ldots,\Lambda_{\rmb,N_\rma}}(\lambda_{\rmb,1}|\lambda_{\rmb,2},\ldots,\lambda_{\rmb,N_\rma})}{(\rho\Psi(n))^2}\cdot 2\sqrt{g_1}\sqrt{\frac{1}{2}+\tilt_\dagger^2\sigma_Z^2},\label{eq:F1'-2}
\end{align}
and
\begin{align}
F'_2
&=2\cdot\bigg(\sqrt{\frac{1}{2}+\tilt_\dagger^2\sigma_Z^2}\bigg)^2\\*
&=1+2\tilt_\dagger^2\sigma_Z^2.\label{eq:F2'-2}
\end{align}
Substituting \eqref{eq:F1 def}, \eqref{eq:F2 def}, \eqref{eq:F1'-2} and \eqref{eq:F2'-2} into \eqref{eq:tilg1 derivative} leads to
\begin{align}
f'_{\tilG_1}(\tilg_1|\tilt_\dagger,g_2,\ldots,g_{N_\rma})
&=\frac{\Big(2\sqrt{g_1}\sqrt{\frac{1}{2}+\tilt_\dagger^2\sigma_Z^2}\Big)^2}{(\rho\Psi(n))^2}\cdot f'_{\Lambda_{\rmb,1}|\Lambda_{\rmb,2},\ldots,\Lambda_{\rmb,N_\rma}}(\lambda_{\rmb,1}|\lambda_{\rmb,2},\ldots,\lambda_{\rmb,N_\rma})\nn\\*
&\quad + \frac{1+2\tilt_\dagger^2\sigma_Z^2}{\rho\Psi(n)}\cdot f_{\Lambda_{\rmb,1}|\Lambda_{\rmb,2},\ldots,\Lambda_{\rmb,N_\rma}}(\lambda_{\rmb,1}|\lambda_{\rmb,2},\ldots,\lambda_{\rmb,N_\rma})\label{eq:f'tilG1 final-1}\\*
&=O(\sqrt{n})\label{eq:f'tilG1 final-2},
\end{align}
where \eqref{eq:f'tilG1 final-2} follows since $f'_{\Lambda_{\rmb,1}|\Lambda_{\rmb,2},\ldots,\Lambda_{\rmb,N_\rma}}$ is independent with $n$, and all other terms are of the same order as those in \eqref{eq:tilG_1 final-1}, which ensures that both terms in \eqref{eq:f'tilG1 final-1} are of order $O(\sqrt{n})$.

\subsection{Proof of Lemma \ref{lemma:converse power} (Covertness Analysis for Converse)}
\label{appendix:covert con}

For the converse part, we consider the worst covert communication case channel $\bh_{\rmw,0}$ for the warden Willie defined in \eqref{eq:def of hw0}.
Therefore, Willie's observation at the $i$-th channel use is now
\begin{align}
\bY_{\rmw,i}=\sqrt{\lambda_0}\bX_i+\bZ_{\rmw,i},
\label{eqB:Willie_single_use}
\end{align}
where $\bZ_{\rmw,i}\sim\calC\calN(\mathbf 0,\bI_{N_\rma})$.
Let $Q_{1,i}$ denote the marginal distribution of $\bY_{\rmw,i}$ under $\rmS_1$ induced by $\bX_i$ with $\Psi_{i}(n)$.
Let $Q_{0,i}:=\calC\calN(\mathbf0,\bI_{N_\rma})$ denote the noise distribution under $\rmS_0$, which implies that
\begin{align}
Q_0^{(n)}=\prod_{i=1}^n Q_{0,i}.
\label{eqB:Q0_product}
\end{align}
It follows from \cite[Eq. (13)]{wang_fundamental_2016} that
\begin{align}
\bbD(Q_1^{(n)}\|Q_0^{(n)})
\ge
\sum_{i=1}^n
\bbD(Q_{1,i}\|Q_{0,i}).
\label{eqB:chain_rule_lower}
\end{align}
We next lower bound $\bbD(Q_{1,i}\|Q_{0,i})$:
\begin{align}
\bbD(Q_{1,i}\|Q_{0,i})
&=-h(\bY_{\rmw,i})+\bbE_{Q_{1,i}}\bigg[\log \frac{1}{Q_{0,i}(\bY_{\rmw,i})}\bigg]\label{eq:KLD h 1}\\*
&=-h(\bY_{\rmw,i})+\sum_{j\in[N_\rma]}\bbE_{Q_{1,i}}\bigg[\log \frac{1}{Q_{0,i}(Y_{\rmw,ij})}\bigg]\label{eq:KLD h 2}\\*
&=-h(\bY_{\rmw,i})+N_\rma\log \pi+\sum_{j\in[N_\rma]}\big(1+\lambda_0\Psi_{i}(n)\big)\label{eq:KLD h 3}\\*
&\geq-\sum_{j\in[m]}\log\pi e\big(1+\lambda_0\Psi_{i}(n)\big)+m\log \pi+\sum_{j\in[m]}\big(1+\lambda_0\Psi_{i}(n)\big)\label{eq:KLD h 4}\\
&=N_\rma\cdot\Big( \lambda_0\Psi_{i}(n)-\log\big(1+\lambda_0\Psi_{i}(n)\big) \Big),\label{eq:KLD h 5}
\end{align}
where \eqref{eq:KLD h 1} follows from the definition of KLD and $h(\cdot)$ is differential entropy, \eqref{eq:KLD h 2} follows since the MIMO channel can be decomposed into $N_\rma$ independent parallel sub-channels via SVD, \eqref{eq:KLD h 3} follws from $Q_{0,i}=\calC\calN(\mathbf{0},\bI_{N_\rma})$ and $\bbE_{Q_{1,i}}[|Y_{\rmw,ij}|^2]=1+\lambda_0\Psi_{i}(n)$, \eqref{eq:KLD h 4} follows from zero-mean multivariate Gaussian distribution is the maximizer of the entropy functional under the second moment constraint \cite[Chapter 12]{cover1999elements}.

Fix any $\omega>1$. Combining \eqref{eqB:chain_rule_lower} and \eqref{eq:KLD h 5} leads to
\begin{align}
\bbD(Q_1^{(n)}\|Q_0^{(n)})&\ge N_\rma\cdot\sum_{i\in[n]}\Big(\lambda_0\Psi_{i}(n)-\log\big(1+\lambda_0\Psi_{i}(n)\big)\Big)\label{eqB:sum_vector_KLD_lower-1}\\
&\ge N_\rma\cdot n\bigg(\lambda_0\cdot\frac1n\sum_{i=1}^n\Psi_{i}(n)-\log\bigg(1+\lambda_0\frac1n\sum_{i=1}^n\Psi_{i}(n)\bigg)\bigg)\label{eqB:sum_vector_KLD_lower-2}\\
&=N_\rma\cdot n\Big(\lambda_0\Psi_{\rma}(n)-\log\big(1+\lambda_0\Psi_{\rma}(n)\big)\Big)\label{eqB:sum_vector_KLD_lower-3}\\
&> N_\rma\cdot\frac{n}{2\omega}\lambda_0^2\Psi_{\rma}^2(n),\label{eqB:sum_vector_KLD_lower-4}
\end{align}
where \eqref{eqB:sum_vector_KLD_lower-2} follows from the Jensen inequality which implies that $\bbE[f(X)]\ge f(\bbE[X])$ for a convex function $f(\cdot)$ and a random variable $X$, \eqref{eqB:sum_vector_KLD_lower-3} follows from the definition of $\Psi_{\rma}$ in \eqref{eq:def of avg power}, and \eqref{eqB:sum_vector_KLD_lower-4} follows from the Taylor series expansion of $\log(1+x)$ around $x=0$ which implies $\log(1+x)<x-\frac{x^2}{2\omega}$ for any $0<x<\frac{3(\omega-1)}{2\omega}$.

To satisfy the covertness constraint \eqref{eq2:both constraints}, the average power for any $(n,M)_\star$-code satisfies
\begin{align}\label{eq:converse power bound}
\Psi_\rma(n)\leq\sqrt{\frac{2\delta\omega}{n\tr\big(\blambda_0^4\big)}}.
\end{align}
The proof of Lemma \ref{lemma:converse power} is now completed.

\subsection{Proof of Lemma \ref{lemma:R* con non} (Non-Asymptotic Converse Bound for CSIRT)}
\label{appendix:csirt con non}

To obtain the converse bound, we shall assume that CSI is available at both transmitter and receiver.
Recall from \eqref{eq:yb channel} that Bob receives the noisy channel output $\bY_\rmb^n=\bX^n(W)\bH_\rmb+\bZ_\rmb^n$. Using the SVD procedure in Section \ref{section_main_results}, when CSI is available at the receiver, Bob post-processes the observation $\bY_\rmb^n$ by unitary transformation via $\bV_\rmb$ to obtain ${\bY'}_\rmb^n:=\bY_\rmb^n\bV_\rmb\in\calY_\rmb^{n\times N_\rma}$. With a slight abuse of notation, we reuse ${\bY}_\rmb^n$ to represent ${\bY'}_\rmb^n$, and similarly for ${\bZ}_\rmb^n$. 
Recall that $\bLambda_\rmb=\diag\{\sqrt{\Lambda_{\rmb,j}}\}_{j\in[N_\rma]}$ contains singular values of $\bH_\rmb$, and $\{Z^n_{\rmb,j}\}_{j\in[N_\rma]}$ are i.i.d. noise vectors distributed as $\calCN(\mathbf{0},\bI_n)$.
As a result, it follows that the MIMO channel \eqref{eq:yb channel} can be transformed into the set of $N_\rma$ parallel quasi-static sub-channels such that for $j\in[N_\rma]$,
\begin{align}\label{eq:parallel channel}
Y^n_{\rmb,j}=x^n_j\sqrt{\Lambda_{\rmb,j}}+Z^n_{\rmb,j},
\end{align}
via SVD.
We shall analyze from the perspective of sub-channels in \eqref{eq:parallel channel}.
Let $Z_{\rmb,ij}^n$ be the $(i,j)$-th entry of $\bZ_\rmb^n$.
Recall that the codeword is now $\bX^n=(X_1^n,\ldots,X_{N_\rma}^n)\in\calX^{n\times N_\rma}$, and $\calF^{n\times N_\rma}$ is a codeword set defined in \eqref{eq:def fset}. Let $\bX^n\in\calF^{n\times N_\rma}$ to satisfy the power constraint.

Let $P_{\bY_\rmb^n,\bLambda_\rmb|\bX^n}\in\calP(\calY_\rmb^{n\times N_\rma}\times \bbR_+^{N_\rma\times N_\rma}|\calX^{n\times N_\rma})$ be the output distribution induced by the codeword $\bX^n$, and $\tilP_{\bY_\rmb^n,\bLambda_\rmb}\in\calP(\calY_\rmb^{n\times N_\rma}\times \bbR_+^{N_\rma\times N_\rma})$ be any auxiliary distribution. Recall that $\beta_{(\cdot)}(\cdot,\cdot)$ is a performance metric defined in \eqref{eq:def beta}. 
Let
\begin{align}
P_{\bY_\rmb^n,\bLambda_\rmb|\bX^n}=P_{\bLambda_\rmb}P_{\bY_\rmb^n|\bLambda_\rmb,\bX^n},
\end{align}
and take as auxiliary channel $\tilP_{\bY_\rmb^n,\bLambda_\rmb}=P_{\bLambda_\rmb}\tilP_{\bY_\rmb^n|\bLambda_\rmb}$, where $\tilP_{\bY_\rmb^n|\bLambda_\rmb=\blambda_\rmb}=\prod_{j\in[N_\rma]}\tilP_{Y_{\rmb,j}^n|\Lambda_{\rmb,j}=\lambda_{\rmb,j}}$ and
\begin{align}
\label{eq:def tilP CN}
\tilP_{Y_{\rmb,j}^n|\Lambda_{\rmb,j}=\lambda_{\rmb,j}}:=\calCN\Big(\mathbf{0},\big(1+\lambda_{\rmb,j}\Psi(n)\big)\bI_n\Big).
\end{align}
By the meta-converse theorem \cite[Theorem 30]{polyanskiy_channel_2010}, we have that for any $\varepsilon\in(0,1)$ and any $\tilP_{\bY_\rmb^n,\bLambda_\rmb}$, every $(n,M)_\rmrt$-code with the maximal error probability $\varepsilon$ and codewords chosen from $\calF^{n\times N_\rma}$ satisfying
\begin{align}
\label{eq:con tilM leq}
M\leq\frac{1}{\inf_{\bx^n\in\calF^{n\times N_\rma}}\beta_{1-\varepsilon}(P_{\bY_\rmb^n,\bLambda_\rmb|\bX^n=\bx^n},\tilP_{\bY_\rmb^n,\bLambda_\rmb})}.
\end{align}
We next focus on the denominator on the right hand side of \eqref{eq:con tilM leq}.
Define a log-likelihood ratio function
\begin{align}
\label{eq:def r log}
r\big(\bX^n;\bY_\rmb^n,\bLambda_\rmb\big):=\log\frac{\rmd P_{\bY_\rmb^n,\bLambda_\rmb|\bX^n}}{\rmd\tilP_{\bY_\rmb^n,\bLambda_\rmb}}.
\end{align}
Let $\bx_2^n:=(x_{ij})_{n\times N_\rma}$, where $x_{ij}=\sqrt{\Psi(n)}$ for any $i\in [n]$ and $j\in[N_\rma]$. Given $\bX^n=\bx_2^n$, from the Neyman--Pearson Lemma \cite[Appendix B]{polyanskiy_channel_2010}, it follows that
\begin{align}
\label{eq:NP beta tilP r}
\beta_{1-\varepsilon}(P_{\bY_\rmb^n,\bLambda_\rmb|\bX^n=\bx_2^n},\tilP_{\bY_\rmb^n,\bLambda_\rmb})=\Pr_{\tilP_{\bY_\rmb^n,\bLambda_\rmb}} \Big\{r\big(\bx_2^n;\bY_\rmb^n,\bLambda_\rmb\big)\geq n\gamma_n\Big\},
\end{align}
where $\gamma_n$ is chosen so that
\begin{align}
\label{eq:r Y|X epsilon}
\Pr_{P_{\bY_\rmb^n,\bLambda_\rmb|\bX^n=\bx_2^n}} \Big\{r\big(\bx_2^n;\bY_\rmb^n,\bLambda_\rmb\big)\leq n\gamma_n\Big\}=\varepsilon.
\end{align}
Note that the random variable $r(\bx_2^n;\bY_\rmb^n,\bLambda_\rmb)$ has the same distribution as $L_n(\bLambda_\rmb)$ in \eqref{eq:def of Lnrt} under $\tilP_{\bY_\rmb^n,\bLambda_\rmb}$, and it has the same distribution as $S_n(\bLambda_\rmb)$ in \eqref{eq:def of Snrt} under $P_{\bY_\rmb^n,\bLambda_\rmb|\bX^n=\bx_2^n}$, whose proofs are deferred to Appendix \ref{subapx:distributions r}.
With the equal power constraint, each codeword leads to the same $\beta_{1-\varepsilon}$, i.e.,
\begin{align}
\label{eq:inf beta equal}
\inf_{\bx^n\in\calF^{n\times N_\rma}}\beta_{1-\varepsilon}(P_{\bY_\rmb^n,\bLambda_\rmb|\bX^n=\bx^n},\tilP_{\bY_\rmb^n,\bLambda_\rmb})=\beta_{1-\varepsilon}(P_{\bY_\rmb^n,\bLambda_\rmb|\bX^n=\bx_2^n},\tilP_{\bY_\rmb^n,\bLambda_\rmb}).
\end{align}
Based on the above results, taking the logarithm of both sides of \eqref{eq:con tilM leq}, dividing by the blocklength $n$, and combining with \eqref{eq:NP beta tilP r} and \eqref{eq:inf beta equal} lead to \eqref{eq:R*rt non}, and \eqref{eq:r Y|X epsilon} leads to \eqref{eq:gamma S=epsilon}.
The proof of the non-asymptotic converse bound for CSIRT is completed. 

\subsubsection{Distributions of $r(\bx_2^n;\bY_\rmb^n,\mathbf{\Lambda}_\rmb)$}
\label{subapx:distributions r}

Let $r_1$ be the distribution of $r(\bx_2^n;\bY_\rmb^n,\bLambda_\rmb)$ under $\tilP_{\bY_\rmb^n,\bLambda_\rmb}$. Let $\tilx^n$ be any column of $\bx_2^n$ as all columns are identical, and $Z_{\rmb,ij}^n$ be the $(i,j)$-th entry of $\bZ_\rmb^n$. Recall that $P_{\bY_\rmb^n,\bLambda_\rmb|\bX^n}=P_{\bLambda_\rmb}P_{\bY_\rmb^n|\bLambda_\rmb,\bX^n}$ and $\tilP_{\bY_\rmb^n,\bLambda_\rmb}=P_{\bLambda_\rmb}\tilP_{\bY_\rmb^n|\bLambda_\rmb}$, from the definition of the log-likelihood ratio function in \eqref{eq:def r log}, it follows that
\begin{align}
r_1
&=\log\frac{\rmd P_{\bY_\rmb^n|\bLambda_\rmb,\bX^n=\bx_2^n}}{\rmd\tilP_{\bY_\rmb^n|\bLambda_\rmb}}\big(\bY_\rmb^n,\bLambda_\rmb\big)\label{eq:calculate r1-1}\\*
&=\sum_{j\in[N_\rma]}\log\frac{\rmd P_{Y_{\rmb,j}^n|\Lambda_{\rmb,j},X_j^n=\tilx^n}}{\rmd \tilP_{Y_{\rmb,j}^n|\Lambda_{\rmb,j}}}\big(Y_{\rmb,j}^n,\Lambda_{\rmb,j}\big)\label{eq:calculate r1-2}\\*
&=\sum_{j\in[N_\rma]}\bigg(n\log\big(1+\Lambda_{\rmb,j}\Psi(n)\big)-\frac{\Lambda_{\rmb,j}\Psi(n)}{1+\Lambda_{\rmb,j}\Psi(n)}\big|Y_{\rmb,j}^n\big|^2+2\sqrt{\Lambda_{\rmb,j}\Psi(n)}\Re\big(Y_{\rmb,j}^n\big)-\Lambda_{\rmb,j}\Psi(n)\bigg)\label{eq:calculate r1-4}\\*
&=\sum_{j\in[N_\rma]}\bigg(n\log\big(1+\Lambda_{\rmb,j}\Psi(n)\big)-\Lambda_{\rmb,j}\Psi(n)\big|Z_{\rmb,j}^n\big|^2+2\sqrt{\Lambda_{\rmb,j}\Psi(n)}\sqrt{1+\Lambda_{\rmb,j}\Psi(n)}\Re\big(Z_{\rmb,j}^n\big)-\Lambda_{\rmb,j}\Psi(n)\bigg)\label{eq:calculate r1-5}\\*
&=\sum_{i\in[n]}\sum_{j\in[N_\rma]}\bigg(\log\big(1+\Lambda_{\rmb,j}\Psi(n)\big)-\Lambda_{\rmb,j}\Psi(n)\big|Z_{\rmb,ij}\big|^2+2\sqrt{\Lambda_{\rmb,j}\Psi(n)}\sqrt{1+\Lambda_{\rmb,j}\Psi(n)}\Re\big(Z_{\rmb,ij}\big)-\Lambda_{\rmb,j}\Psi(n)\bigg)\label{eq:calculate r1-6}\\*
&=\sum_{i\in[n]}\sum_{j\in[N_\rma]}\bigg(\log\big(1+\Lambda_{\rmb,j}\Psi(n)\big)+1-\Big|\sqrt{\Lambda_{\rmb,j} \Psi(n)}Z_{\rmb,ij}-\sqrt{1+\Lambda_{\rmb,j} \Psi(n)}\Big|^2\bigg),\label{eq:calculate r1-7}
\end{align}
where \eqref{eq:calculate r1-2} follows from the fact that $P_{\bY_\rmb^n|\bLambda_\rmb,\bX^n=\bx_2^n}=\prod_{j\in[N_\rma]}P_{Y_{\rmb,j}^n|\Lambda_{\rmb,j},X_j^n=\tilx^n}$ and $\tilP_{\bY_\rmb^n|\bLambda_\rmb}=\prod_{j\in[N_\rma]}\tilP_{Y_{\rmb,j}^n|\Lambda_{\rmb,j}}$, \eqref{eq:calculate r1-4} follows from the fact that $P_{Y_{\rmb,j}^n|\Lambda_{\rmb,j},X_j^n=\tilx^n}=\calCN(\sqrt{\Lambda_{\rmb,j}}\tilx^n,\bI_n)$, $\tilP_{Y_{\rmb,j}^n|\Lambda_{\rmb,j}}=\calCN(\mathbf{0},(1+\Lambda_{\rmb,j}\Psi(n))\bI_n)$ and $\tilx^n=(\sqrt{\Psi(n)},\ldots,\sqrt{\Psi(n)})$, \eqref{eq:calculate r1-5} follows from the fact that $Y_{\rmb,j}^n=\sqrt{1+\Lambda_{\rmb,j}\Psi(n)}Z_{\rmb,j}^n$ under $\tilP_{\bY_\rmb^n,\bLambda_\rmb}$ and $Z_{\rmb,j}^n\sim\calCN(\mathbf{0},\bI_n)$, \eqref{eq:calculate r1-6} follows since the entries of $\bZ_\rmb^n$ are generated i.i.d. from $\calCN(0,1)$, and \eqref{eq:calculate r1-7} follows from the perfect square trinomial. Note that the right hand side of \eqref{eq:calculate r1-7} is indeed the definition of $L_n(\bLambda_\rmb)$ in \eqref{eq:def of Lnrt}.

Under $P_{\bY_\rmb^n,\bLambda_\rmb|\bX^n=\bx_2^n}$, let $r_2$ be the distribution of $r(\bx_2^n;\bY_\rmb^n,\bLambda_\rmb)$, and now $Y_{\rmb,j}^n=\sqrt{\Lambda_{\rmb,j}}\tilx^n+Z_{\rmb,j}^n$. Recall that the entries of $\bZ_\rmb^n$ are generated i.i.d. from $\calCN(0,1)$ and $\tilx^n=(\sqrt{\Psi(n)},\ldots,\sqrt{\Psi(n)})$. By calculating, it follows from \eqref{eq:calculate r1-4} that
\begin{align}
r_2
&=\sum_{i\in[n]}\sum_{j\in[N_\rma]}\bigg(\log\big(1+\Lambda_{\rmb,j}\Psi(n)\big)-\frac{\Lambda_{\rmb,j}\Psi(n)}{1+\Lambda_{\rmb,j}\Psi(n)}\Big|\sqrt{\Lambda_{\rmb,j}\Psi(n)}+Z_{\rmb,ij}\Big|^2+2\sqrt{\Lambda_{\rmb,j}\Psi(n)}\Re\big(Z_{\rmb,ij}\big)+\Lambda_{\rmb,j}\Psi(n)\bigg)\label{eq:calculate r2-1}\\*
&=\sum_{i\in[n]}\sum_{j\in[N_\rma]}\bigg(\log\big(1+\Lambda_{\rmb,j}\Psi(n)\big)-\frac{\Lambda_{\rmb,j}\Psi(n)\big|Z_{\rmb,ij}\big|^2-2\sqrt{\Lambda_{\rmb,j}\Psi(n)}\Re\big(Z_{\rmb,ij}\big)-\Lambda_{\rmb,j}\Psi(n)}{1+\Lambda_{\rmb,j}\Psi(n)}\bigg)\label{eq:calculate r2-2}\\*
&=\sum_{i\in[n]}\sum_{j\in[N_\rma]}\bigg(\log\big(1+\Lambda_{\rmb,j}\Psi(n)\big)+1-\frac{\big|\sqrt{\Lambda_{\rmb,j} \Psi(n)}Z_{\rmb,ij}-1\big|^2}{1+\Lambda_{\rmb,j} \Psi(n)}\bigg).\label{eq:calculate r2-3}
\end{align}
Note that the right hand side of \eqref{eq:calculate r2-3} is indeed the definition of $S_n(\bLambda_\rmb)$ in \eqref{eq:def of Snrt}.

\subsection{Proof of Lemma \ref{lemma:csirt con reliability} (Asymptotic Converse Bound for CSIRT)}
\label{appendix:csirt con asy}

Recall that $P_{\bY_\rmb^n,\bLambda_\rmb|\bX^n}$ is the output distribution induced by the codeword $\bX^n$, $\tilP_{\bY_\rmb^n,\bLambda_\rmb}$ is any auxiliary distribution, and $\bx_2^n=(x_{ij})_{n\times N_\rma}$, where $x_{ij}=\sqrt{\Psi(n)}$ for any $i\in [n]$ and $j\in[N_\rma]$.
From Appendix \ref{appendix:csirt con non}, we obtain
\begin{align}
\Pr\big\{L_n(\bLambda_\rmb)\geq n\gamma_n\big\}
&=\beta_{1-\varepsilon}\big(P_{\bY_\rmb^n,\bLambda_\rmb|\bX^n=\bx_2^n},\tilP_{\bY_\rmb^n,\bLambda_\rmb}\big),\label{eq:L = beta}\\*
\Pr\big\{S_n(\bLambda_\rmb)\leq n\gamma_n\big\}
&=1-\Pr_{P_{\bY_\rmb^n,\bLambda_\rmb|\bX^n=\bx_2^n}}\Bigg\{\log\frac{\rmd P_{\bY_\rmb^n,\bLambda_\rmb|\bX^n=\bx_2^n}}{\rmd \tilP_{\bY_\rmb^n,\bLambda_\rmb}}\geq n\gamma_n\Bigg\},\label{eq:S = 1-beta}
\end{align}
where \eqref{eq:L = beta} is obviously follows from \eqref{eq:NP beta tilP r}, and \eqref{eq:S = 1-beta} follows from \eqref{eq:def r log} and \eqref{eq:r Y|X epsilon}.
We first relax the upper bound of \eqref{eq:R*rt non} by lower-bounding its denominator, it follows from \eqref{eq:L = beta} and \cite[Eq. (106)]{polyanskiy_channel_2010} that
\begin{align}
\Pr\big\{L_n(\bLambda_\rmb)\geq n\gamma_n\big\}
&\geq \frac{1}{e^{n\gamma_n}}\Bigg(1-\Pr_{P_{\bY_\rmb^n,\bLambda_\rmb|\bX^n=\bx_2^n}}\Bigg\{\log\frac{\rmd P_{\bY_\rmb^n,\bLambda_\rmb|\bX^n=\bx_2^n}}{\rmd \tilP_{\bY_\rmb^n,\bLambda_\rmb}}\geq n\gamma_n\Bigg\}-\varepsilon\Bigg) \label{eq:beta geq S}\\*
&= e^{-n\gamma_n}\Big(\Pr\big\{S_n(\bLambda_\rmb)\leq n\gamma_n\big\}-\varepsilon\Big),\label{eq:beta geq S-2}
\end{align}
where \eqref{eq:beta geq S-2} follows from \eqref{eq:S = 1-beta}.
Substituting \eqref{eq:beta geq S-2} into \eqref{eq:R*rt non} leads to
\begin{align}
\label{eq:R* original S}
R_{\rmrt,\equ}^*(n,\varepsilon)\leq \gamma_n-\frac{1}{n}\log\Big(\Pr\big\{S_n(\bLambda_\rmb)\leq n\gamma_n\big\}-\varepsilon\Big).
\end{align}
On the right hand side of \eqref{eq:R* original S}, we shall focus on
\begin{align}
\label{eq:pr S leq ngamma}
\Pr\big\{S_n(\bLambda_\rmb)\leq n\gamma_n\big\}.
\end{align} 
From the definition in \eqref{eq:def of Snrt}, the random variable $S_n(\bLambda_\rmb)$ is the sum of $n$ i.i.d. random variables with mean and variance:
\begin{align}
\mu(\bLambda_\rmb)&:=\sum_{j\in[N_\rma]}\log\big(1+\Lambda_{\rmb,j}\Psi(n)\big),\label{eq:S sum mean}\\*
\sigma^2(\bLambda_\rmb)&:=N_\rma-\sum_{j\in[N_\rma]}\frac{1}{\big(1+\Lambda_{\rmb,j}\Psi(n)\big)^2},\label{eq:S sum variance}
\end{align}
respectively. Let $\blambda_\rmb\in\bbR^{N_\rma\times N_\rma}$. Given $\bLambda_\rmb=\blambda_\rmb$, define
\begin{align}
\label{eq:S sum normalized}
u(\blambda_\rmb):=\frac{\gamma_n-\mu(\blambda_\rmb)}{\sigma(\blambda_\rmb)}.
\end{align}
Recall that $Q(\cdot)$ is the complementary Gaussian cumulative distribution function.
Using the Cram\'er--Esseen Theorem \cite[Theorem 15]{yang2014quasi}, we shall show in Appendix \ref{appendix:Proof of CLT S} that
\begin{align}
\label{eq:CLT S}
\Pr\big\{S_n(\bLambda_\rmb)\leq n\gamma_n\big|\bLambda_\rmb=\blambda_\rmb\big\}\geq Q(-\sqrt{n}u(\blambda_\rmb))+O\Big(\frac{1}{n}\Big).
\end{align}
For any $R\in\bbR_+$, it follows from \eqref{eq:con C final def} that the covert outage probability is given by
\begin{align}
\label{eq:def of Fout_con}
F_\rmout(R):=\Pr\Big\{ \log\det \big( \bI_{N_\rmb}+ \Psi(n)\bH_\rmb^\rmH\bH_\rmb \big)<R \Big\}.
\end{align}
Furthermore, we average \eqref{eq:CLT S} over $\bLambda_\rmb$ and use Lemma \ref{lemma:B sqrtn A ff'} to show in Appendix \ref{appendix:Proof of eliminate Q in S} that
\begin{align}
\label{eq:eliminate Q in S}
\Pr\big\{S_n(\bLambda_\rmb)\leq n\gamma_n\big\}\geq F_\rmout(\gamma_n)+O\Big(\frac{1}{\sqrt{n}}\Big).
\end{align}
Note that under the covert constraint, the $O(\frac{1}{n})$ term in \eqref{eq:CLT S} is enlarged to $O(\frac{1}{\sqrt{n}})$ in \eqref{eq:eliminate Q in S}, but it will be compensated by the covert outage probability in the following analysis.
We choose $\gamma_n=C_{\varepsilon}(\Psi(n))+O(\frac{1}{n})$. Substituting \eqref{eq:eliminate Q in S} into \eqref{eq:R* original S} leads to
\begin{align}
R_{\rmrt,\equ}^*(n,\varepsilon)
&\leq \gamma_n-\frac{1}{n}\log\bigg(F_\rmout(\gamma_n)-\varepsilon+O\Big(\frac{1}{\sqrt{n}}\Big)\bigg)\label{eq:con final-1}\\*
&=C_{\varepsilon}(\Psi(n))+O\Big(\frac{1}{n}\Big)-\frac{1}{n}\log\bigg(O\Big(\frac{1}{\sqrt{n}}\Big)+O\Big(\frac{1}{n^2}\Big)\bigg)\label{eq:con final-2}\\*
&=C_{\varepsilon}(\Psi(n))+O\Big(\frac{\log n}{n}\Big),\label{eq:con final-3}
\end{align}
where \eqref{eq:con final-2} follows from the Taylor Series of expansion of $F_\rmout(\gamma_n)$ around the covert outage rate $C_{\varepsilon}(\Psi(n))$ by \eqref{eq:Ft taylor}, and the fact that $\gamma_n=C_{\varepsilon}(\Psi(n))+O(\frac{1}{n})$, $F_\rmout(C_{\varepsilon}(\Psi(n)))=\varepsilon$, $O((\gamma_n-C_{\varepsilon}(\Psi(n)))^2)=O(n^{-2})$, and $F'_\rmout(C_{\varepsilon}(\Psi(n)))=O(\sqrt{n})$ by \eqref{eq:F' final-3}. Recall that $R_{\rmrt,\max}^*$ defined in \eqref{eq:def M max} is the maximal achievable rate for quasi-static MIMO fading channels with the maximal power constraint under CSIRT case. It follows from \cite[Lemma 39]{polyanskiy_channel_2010} that
\begin{align}
R_{\rmrt,\max}^*(n,\varepsilon)
&\leq(1+\frac{1}{n})R_{\rmrt,\equ}^*(n+1,\varepsilon)\\*
&\leq C_{\varepsilon}(\Psi(n))+O\Big(\frac{\log n}{n}\Big),\label{eq:Re to Rm}
\end{align}
where \eqref{eq:Re to Rm} follows from the fact that $\Psi(n)=O(\frac{1}{\sqrt{n}})$ and the remaining terms can be collected into $O(\frac{\log n}{n})$.
The proof of Lemma \ref{lemma:csirt con reliability} is now completed.

\subsection{Proof of \eqref{eq:CLT S}}
\label{appendix:Proof of CLT S}
We use the Cram\'er--Esseen Theorem \cite[Theorem 15]{yang2014quasi} to obtain the second-order asymptotics, and here we restate it as a lemma. Let $\{X_i\}_{i\in[n]}$ be i.i.d. real random variables having zero mean and unit variance, $(t,\xi_1)\in\bbR^2$. Define two functions
\begin{align}
\varphi(t)&:=\bbE\big[e^{\imath tX_1}\big],\\*
\tilF_n(\xi_1)&:=\Pr\Bigg\{\frac{1}{\sqrt{n}}\sum_{i\in[n]}X_i\leq\xi_1\Bigg\}.
\end{align}
Let $k_0<1$, $k_1:=\frac{\bbE[X_1^3]}{6\sqrt{2\pi}}$, and $k_2\in\bbR_+$ be a constant independent of $\{X_i\}_{i\in[n]}$ and $\xi_1$. Define $\xi_2:=\frac{1}{12\bbE[|X_1|^3]}$.
\begin{lemma}
\label{lemma:cramer-esseen}
For any $\xi_1$ and $n$, and some $k_0<1$, if $\bbE[|X_1|^4]<\infty$ and $\sup_{|t|\geq\xi_2}|\varphi(t)|\leq k_0$, then
\begin{align}
\bigg|\tilF_n(\xi_1)-Q(-\xi_1)-k_1(1-\xi_1^2)e^{-\frac{\xi_1^2}{2}}\frac{1}{\sqrt{n}}\bigg|\leq k_2\bigg(\frac{\bbE\big[|X_1|^4\big]}{n}+n^6\Big(k_0+\frac{1}{2n}\Big)^n\bigg).
\end{align}
\end{lemma}
The proof of Lemma \ref{lemma:cramer-esseen} can be found in \cite[Theorem 15]{yang2014quasi} and \cite[Theorem VI.1]{petrov2012sums}.

Based on the definitions of $S_n(\bLambda_\rmb)$ in \eqref{eq:def of Snrt} and $\sigma^2(\bLambda_\rmb)$ in \eqref{eq:S sum variance}, we construct the form of $\{X_i\}_{i\in[n]}$ in Lemma \ref{lemma:cramer-esseen}. Recall that $\{Z_{\rmb,ij}\}_{i\in[n],j\in[N_\rma]}$ are i.i.d. $\calCN(0,1)$ distributed. For $i\in[n]$, define a random variable
\begin{align}
\label{eq:def Di}
D_i(\bLambda_\rmb):=\frac{1}{\sigma(\bLambda_\rmb)}\sum_{j\in[N_\rma]}\bigg(1-\frac{\big|\sqrt{\Lambda_{\rmb,j} \Psi(n)}Z_{\rmb,ij}-1\big|^2}{1+\Lambda_{\rmb,j} \Psi(n)}\bigg).
\end{align}
Note that $\{D_i(\bLambda_\rmb)\}_{i\in[n]}$ have zero mean and unit variance, and are conditionally i.i.d. given $\bLambda_\rmb$. Recall that $u(\bLambda_\rmb)$ is defined in \eqref{eq:S sum normalized}, and using these definitions we can rewrite \eqref{eq:pr S leq ngamma} as
\begin{align}
\label{eq:S=Di}
\Pr\big\{S_n(\bLambda_\rmb)\leq n\gamma_n\big\}=\Pr\Bigg\{\frac{1}{\sqrt{n}}\sum_{i\in[n]}D_i(\bLambda_\rmb)\leq\sqrt{n}u(\bLambda_\rmb)\Bigg\}.
\end{align}
For $\{D_i(\bLambda_\rmb)\}_{i\in[n]}$, define
\begin{align}
\varphi_{D_i}(t)&:=\bbE\big[e^{\imath tD_i(\bLambda_\rmb)}\big|\bLambda_\rmb=\blambda_\rmb\big],\label{eq:def phi Di}\\*
\tilde{\xi}_2&:=\frac{1}{12\bbE\big[|D_i(\bLambda_\rmb)|^3\big|\bLambda_\rmb=\blambda_\rmb\big]}.\label{eq:def til xi 2}
\end{align}
We shall verify at the end of this appendix that $\{D_i(\bLambda_\rmb)\}_{i\in[n]}$ defined in \eqref{eq:def Di} indeed satisfy the conditions of Lemma \ref{lemma:cramer-esseen}, i.e., for $i\in[n]$, \emph{1)} $\bbE[|D_i(\bLambda_\rmb)|^4|\bLambda_\rmb=\blambda_\rmb]=O(1)$ showed in Appendix \ref{subapx:clt condition1}, and \emph{2)} there exists a $\tilk_0<1$ such that $\sup_{|t|\geq\tilde{\xi}_2}|\varphi_{D_i}(t)|\leq\tilk_0$ showed in Appendix \ref{subapx:clt condition2}.
For any $\blambda_\rmb\in\bbR^{N_\rma\times N_\rma}$ and $n\in\bbN$, and a constant $\tilk_2\in\bbR_+$ which is independent of $\{D_i(\bLambda_\rmb)\}_{i\in[n]}$ and $u(\bLambda_\rmb)$, it follows from Lemma \ref{lemma:cramer-esseen} that
\begin{align}
\label{eq:clt 1result}
&\Pr\Bigg\{\frac{1}{\sqrt{n}}\sum_{i\in[n]}D_i(\bLambda_\rmb)\leq\sqrt{n}u(\bLambda_\rmb)\Bigg|\bLambda_\rmb=\blambda_\rmb\Bigg\}-Q\big(-\sqrt{n}u(\blambda_\rmb)\big)\nn\\*
&\quad\geq \frac{\bbE\big[D_i(\bLambda_\rmb)^3\big|\bLambda_\rmb=\blambda_\rmb\big]}{6\sqrt{2\pi}\sqrt{n}}\big(1-nu(\blambda_\rmb)^2\big)e^{-\frac{nu(\blambda_\rmb)^2}{2}}-\frac{\tilk_2}{n}\bbE\big[|D_i(\bLambda_\rmb)|^4\big|\bLambda_\rmb=\blambda_\rmb\big]-\tilk_2n^6\Big(\tilk_0+\frac{1}{2n}\Big)^n.
\end{align}
We shall show in Appendix \ref{subapx:clt rhs order 1/n} that the right hand side of \eqref{eq:clt 1result} is of order $O(\frac{1}{n})$, then we obtain
\begin{align}
\label{eq:clt 2result}
\Pr\Bigg\{\frac{1}{\sqrt{n}}\sum_{i\in[n]}D_i(\bLambda_\rmb)\leq\sqrt{n}u(\bLambda_\rmb)\Bigg|\bLambda_\rmb=\blambda_\rmb\Bigg\}
\geq Q\big(-\sqrt{n}u(\blambda_\rmb)\big)+O\Big(\frac{1}{n}\Big).
\end{align}
Combining \eqref{eq:S=Di} and \eqref{eq:clt 2result} leads to \eqref{eq:CLT S}.
The proof is completed.

\subsubsection{Proof of Condition 1 of Lemma \ref{lemma:cramer-esseen}}
\label{subapx:clt condition1}
Recall that $D_i(\bLambda_\rmb)$ is defined in \eqref{eq:def Di}, and we shall calculate $\bbE[|D_i(\bLambda_\rmb)|^4|\bLambda_\rmb=\blambda_\rmb]$. Recall that $\sigma^2(\bLambda_\rmb)$ is defined in \eqref{eq:S sum variance}, and $\bLambda_\rmb=\diag\big(\sqrt{\Lambda_{\rmb,1}},\ldots,\sqrt{\Lambda_{\rmb,N_\rma}}\big)$. Given $\Lambda_{\rmb,j}=\lambda_{\rmb,j}$ for $j\in[N_\rma]$, let
\begin{align}
\label{eq:def tilD_ij}
\tilD_{ij}:=1-\frac{\big|\sqrt{\lambda_{\rmb,j} \Psi(n)}Z_{\rmb,ij}-1\big|^2}{1+\lambda_{\rmb,j} \Psi(n)},
\end{align}
then we have
\begin{align}
\bbE\big[|D_i(\bLambda_\rmb)|^4\big|\bLambda_\rmb=\blambda_\rmb\big]
&=\frac{1}{\sigma^4(\blambda_\rmb)}\cdot\bbE\Bigg[\bigg|\sum_{j\in[N_\rma]}\tilD_{ij}\bigg|^4\Bigg]\label{eq:D to tilD-1}\\*
&=\frac{1}{\sigma^4(\blambda_\rmb)}\cdot\Bigg(\sum_{j\in[N_\rma]}\bbE\Big[\big|\tilD_{ij}\big|^4\Big]+6\sum_{j\in[N_\rma],j'\in(j,N_\rma]}\bbE\Big[\big|\tilD_{ij}\big|^2\Big]\Big[\big|\tilD_{ij'}\big|^2\Big]\Bigg),\label{eq:D to tilD-2}
\end{align}
where \eqref{eq:D to tilD-2} follows from the fact that $\bbE[|\tilD_{ij}|]=0$.
By calculating, the quartic term in \eqref{eq:D to tilD-2} follows that
\begin{align}
\big|\tilD_{ij}\big|^4
&=\Big(\frac{1}{1+\lambda_{\rmb,j}\Psi(n)}\Big)^4 \cdot \Big|\lambda_{\rmb,j}\Psi(n)\big|Z_{\rmb,ij}\big|^2-2\sqrt{\lambda_{\rmb,j}\Psi(n)}\Re\big(Z_{\rmb,ij}\big)-\lambda_{\rmb,j}\Psi(n)\Big|^4\label{eq:calculate 4th tilD-1}\\*
&\leq \Big(\frac{1}{1+\lambda_{\rmb,j}\Psi(n)}\Big)^4 \cdot \Big(\lambda_{\rmb,j}\Psi(n)\big|Z_{\rmb,ij}\big|^2+2\sqrt{\lambda_{\rmb,j}\Psi(n)}\big|\Re\big(Z_{\rmb,ij}\big)\big|+\lambda_{\rmb,j}\Psi(n)\Big)^4\label{eq:calculate 4th tilD-2}\\*
&=\Big(\frac{1}{1+\lambda_{\rmb,j}\Psi(n)}\Big)^4 \cdot \Big( \big(\lambda_{\rmb,j}\Psi(n)\big)^4 + \big(\lambda_{\rmb,j}\Psi(n)\big)^4 \big|Z_{\rmb,ij}\big|^8 + 16\big(\lambda_{\rmb,j}\Psi(n)\big)^2 \big|\Re\big(Z_{\rmb,ij}\big)\big|^4 \nn\\*
&\quad +8\big(\lambda_{\rmb,j}\Psi(n)\big)^\frac{7}{2}\big|\Re\big(Z_{\rmb,ij}\big)\big| + 32\big(\lambda_{\rmb,j}\Psi(n)\big)^\frac{5}{2}\big|\Re\big(Z_{\rmb,ij}\big)\big|^3 + 48\big(\lambda_{\rmb,j}\Psi(n)\big)^3 \big|Z_{\rmb,ij}\big|^2 \big|\Re\big(Z_{\rmb,ij}\big)\big|^2 \Big) \nn\\*
&\quad +32\big(\lambda_{\rmb,j}\Psi(n)\big)^\frac{5}{2} \big|Z_{\rmb,ij}\big|^2 \big|\Re\big(Z_{\rmb,ij}\big)\big|^3 + 4\big(\lambda_{\rmb,j}\Psi(n)\big)^4 \big|Z_{\rmb,ij}\big|^6 + 8\big(\lambda_{\rmb,j}\Psi(n)\big)^\frac{7}{2} \big|Z_{\rmb,ij}\big|^6 \big|\Re\big(Z_{\rmb,ij}\big)\big| \nn\\*
&\quad +6\big(\lambda_{\rmb,j}\Psi(n)\big)^4 \big|Z_{\rmb,ij}\big|^4 + 24\big(\lambda_{\rmb,j}\Psi(n)\big)^3 \big|\Re\big(Z_{\rmb,ij}\big)\big|^2 + 24\big(\lambda_{\rmb,j}\Psi(n)\big)^3 \big|Z_{\rmb,ij}\big|^4 \big|\Re\big(Z_{\rmb,ij}\big)\big|^2 \nn\\*
&\quad + 4\big(\lambda_{\rmb,j}\Psi(n)\big)^4 \big|Z_{\rmb,ij}\big|^2 + 24\big(\lambda_{\rmb,j}\Psi(n)\big)^\frac{7}{2} \big|Z_{\rmb,ij}\big|^4 \big|\Re\big(Z_{\rmb,ij}\big)\big| + 24\big(\lambda_{\rmb,j}\Psi(n)\big)^\frac{7}{2} \big|Z_{\rmb,ij}\big|^2 \big|\Re\big(Z_{\rmb,ij}\big)\big|,\label{eq:calculate 4th tilD-3}
\end{align}
where \eqref{eq:calculate 4th tilD-2} follows from the triangle inequality and ensures that each term is nonnegative, and \eqref{eq:calculate 4th tilD-3} follows by expanding the quartic term.
Define the following nonnegative constants:
\begin{align}
C_1&:=\bbE\big[\big|Z_{\rmb,ij}\big|^8\big] + 4\bbE\big[\big|Z_{\rmb,ij}\big|^6\big] + 6\bbE\big[\big|Z_{\rmb,ij}\big|^4\big] + 4\bbE\big[\big|Z_{\rmb,ij}\big|^2\big] + 1,\\*
C_2&:=8\bbE\big[\big|Z_{\rmb,ij}\big|^6 \big|\Re\big(Z_{\rmb,ij}\big)\big|\big] + 24\bbE\big[\big|Z_{\rmb,ij}\big|^4 \big|\Re\big(Z_{\rmb,ij}\big)\big|\big] + 24\bbE\big[\big|Z_{\rmb,ij}\big|^2 \big|\Re\big(Z_{\rmb,ij}\big)\big|\big] + 8\bbE\big[\big|\Re\big(Z_{\rmb,ij}\big)\big|\big],\\*
C_3&:=24\bbE\big[\big|Z_{\rmb,ij}\big|^4 \big|\Re\big(Z_{\rmb,ij}\big)\big|^2\big] + 48\bbE\big[\big|Z_{\rmb,ij}\big|^2 \big|\Re\big(Z_{\rmb,ij}\big)\big|^2\big] + 24\bbE\big[\big|\Re\big(Z_{\rmb,ij}\big)\big|^2\big],\\*
C_4&:=32\bbE\big[\big|Z_{\rmb,ij}\big|^2 \big|\Re\big(Z_{\rmb,ij}\big)\big|^3\big] + 32\bbE\big[\big|\Re\big(Z_{\rmb,ij}\big)\big|^3\big],\\*
C_5&:=16\bbE\big[\big|\Re\big(Z_{\rmb,ij}\big)\big|^4\big].
\end{align}
As $\Psi(n)=O(\frac{1}{\sqrt{n}})$, it follows that
\begin{align}
&\bbE\Big[\big|\tilD_{ij}\big|^4\Big]\nn\\*
&\quad \leq \Big(\frac{1}{1+\lambda_{\rmb,j}\Psi(n)}\Big)^4 \cdot \Big( C_1\big(\lambda_{\rmb,j}\Psi(n)\big)^4 + C_2\big(\lambda_{\rmb,j}\Psi(n)\big)^\frac{7}{2} + C_3\big(\lambda_{\rmb,j}\Psi(n)\big)^3 + C_4\big(\lambda_{\rmb,j}\Psi(n)\big)^\frac{5}{2} + C_5\big(\lambda_{\rmb,j}\Psi(n)\big)^2 \Big)\label{eq:bbE tilD 4th-1}\\*
&\quad =O(1) \cdot \Big( O(n^{-2}) + O(n^{-\frac{7}{4}}) + O(n^{-\frac{3}{2}}) + O(n^{-\frac{5}{4}}) +O(n^{-1}) \Big)\label{eq:bbE tilD 4th-2}\\*
&\quad =O(n^{-1}),\label{eq:bbE tilD 4th-3}
\end{align}
where \eqref{eq:bbE tilD 4th-1} follows from \eqref{eq:calculate 4th tilD-3}, and the first multiplicative term is $O(1)$, the latter term is dominated by $C_5(\lambda_{\rmb,j}\Psi(n))^2=O(n^{-1})$. Similarly, the square term in \eqref{eq:D to tilD-2} follows that
\begin{align}
\bbE\Big[\big|\tilD_{ij}\big|^2\Big]
&=\Big(\frac{1}{1+\lambda_{\rmb,j}\Psi(n)}\Big)^2 \cdot \Big( 5\big(\lambda_{\rmb,j}\Psi(n)\big)^2 - 4\big(\lambda_{\rmb,j}\Psi(n)\big)^\frac{3}{2} + 2\lambda_{\rmb,j}\Psi(n) \Big)\label{eq:bbE tilD 2nd-1}\\*
&=O(1) \cdot \Big( O(n^{-1}) + O(n^{-\frac{3}{4}}) + O(n^{-\frac{1}{2}}) \Big)\label{eq:bbE tilD 2nd-2}\\*
&=O(n^{-\frac{1}{2}}),\label{eq:bbE tilD 2nd-3}
\end{align}
where \eqref{eq:bbE tilD 2nd-1} follows from the fact that $\bbE[|Z_{\rmb,ij}|^4]=2$, $\bbE[|Z_{\rmb,ij}|^2]=1$, $\bbE[|\Re(Z_{\rmb,ij})|^2]=\frac{1}{2}$ and $\bbE[|\Re(Z_{\rmb,ij})|]=0$.
Substituting \eqref{eq:bbE tilD 4th-3} and \eqref{eq:bbE tilD 2nd-3} into \eqref{eq:D to tilD-2} leads to
\begin{align}
\bbE\big[|D_i(\bLambda_\rmb)|^4\big|\bLambda_\rmb=\blambda_\rmb\big]
&= \frac{1}{\sigma^4(\blambda_\rmb)}\cdot O(n^{-1})\label{eq:bbE tilD 4th=O(1)-1}\\*
&= O(1),\label{eq:bbE tilD 4th=O(1)-2}
\end{align}
where \eqref{eq:bbE tilD 4th=O(1)-2} follows from the fact that $\sigma^2(\blambda_\rmb)=O(n^{-\frac{1}{2}})$. 
Note that the order only depends on $\Psi(n)$, and in the covertness analysis, the power scaling is actually constrained to be of order $\Theta(\frac{1}{\sqrt{n}})$ in \eqref{eq:con power level}, which ensures that \eqref{eq:bbE tilD 4th=O(1)-2} holds strictly.
From \eqref{eq:bbE tilD 4th=O(1)-2}, \emph{Condition 1} of Lemma \ref{lemma:cramer-esseen} is satisfied for $\{D_i(\bLambda_\rmb)\}_{i\in[n]}$. 

\subsubsection{Proof of Condition 2 of Lemma \ref{lemma:cramer-esseen}}
\label{subapx:clt condition2}
To prove $\sup_{|t|\geq\tilde{\xi}_2}|\varphi_{D_i}(t)|\leq\tilk_0$, we first calculate $|\varphi_{D_i}(t)|$ to obtain an upper bound.
Define $\tila:=\lambda_{\rmb,j}\Psi(n)$ and $\tilde{\sigma}^2:=1-\frac{1}{(1+\tila)^2}$. Given $\bLambda_\rmb=\blambda_\rmb$, from the definition of $\varphi_{D_i}(t)$ in \eqref{eq:def phi Di} and $\tilD_{ij}$ in \eqref{eq:def tilD_ij}, it follows that
\begin{align}
\big|\varphi_{D_i}(t)\big|
&=\Bigg|\bbE\Bigg[\exp\Bigg({\frac{\imath t}{\sigma(\blambda_\rmb)}\sum_{j\in[N_\rma]}\tilD_{ij}}\Bigg)\Bigg]\Bigg|\label{eq:calculate phi-1}\\*
&=\prod_{j\in[N_\rma]}\Bigg|\bbE\bigg[\exp\bigg({\frac{\imath t}{\tilde{\sigma}}\tilD_{ij}}\bigg)\bigg]\Bigg|\label{eq:calculate phi-2}\\*
&=\prod_{j\in[N_\rma]}\Bigg| \exp\bigg(\frac{\imath t}{\tilde{\sigma}}\bigg) \bbE\bigg[ \exp\bigg( -\frac{\imath t}{\tilde{\sigma}}\cdot\frac{\big| \sqrt{\tila}Z_{\rmb,ij}-1 \big|^2}{1+\tila} \bigg) \bigg] \Bigg|\label{eq:calculate phi-3}\\*
&=\prod_{j\in[N_\rma]}\Bigg| \exp\bigg(\frac{\imath t}{\tilde{\sigma}}\bigg) \exp\bigg( -\frac{\imath t}{\tilde{\sigma}(1+\tila)} - \frac{t^2 \tila}{\tilde{\sigma}^2(1+\tila)^2} \bigg) \bigg( 1+\frac{\imath t \tila}{\tilde{\sigma}(1+\tila)} \bigg)^{-1} \Bigg| \label{eq:calculate phi-4}\\*
&=\prod_{j\in[N_\rma]}\Bigg( \exp\bigg( -\frac{t^2\tila}{\tilde{\sigma}^2(1+\tila)^2 } \bigg) \bigg( 1 + \frac{t^2 \tila^2}{\tilde{\sigma}^2(1+\tila)^2} \bigg)^{-\frac{1}{2}}  \Bigg)\label{eq:calculate phi-5}\\*
&=\prod_{j\in[N_\rma]} \Bigg( \bigg( 1-\frac{t^2}{\tila+2}+o\Big(\frac{t^2}{\tila+2}\Big) \bigg) \sqrt{\frac{\tila +2}{\tila t^2+\tila+2}} \Bigg),\label{eq:calculate phi-6}
\end{align}
where \eqref{eq:calculate phi-4} follows since that for $(a_1,a_2)\in\bbR^2$ and a random variable $A\sim\calCN(-1,a_1)$, the characteristic function of $a_2|A|^2$ is given as $\bbE[\exp(\imath ta_2|A|^2)]=\frac{\exp(\imath ta_2-t^2a_1a_2^2)}{1-\imath ta_1a_2}$, \eqref{eq:calculate phi-5} follows from taking the modulus of complex terms, and \eqref{eq:calculate phi-6} follows from the Taylor Series of expansion of $e^{-x}$ which implies $e^x=1-x+o(x)$ and the fact that $\tilde{\sigma}^2=1-\frac{1}{(1+\tila)^2}$.
From \eqref{eq:calculate phi-6}, $|\varphi_{D_i}(t)|$ decreases monotonically with $t$, implying that
\begin{align}
\label{eq:def tilk0 sup phi}
\sup_{|t|\geq\tilde{\xi}_2}|\varphi_{D_i}(t)|\leq \prod_{j\in[N_\rma]} \Bigg( \bigg( 1-\frac{\tilde{\xi}_2^2}{\tila+2}+o\Big(\frac{\tilde{\xi}_2^2}{\tila+2}\Big) \bigg) \sqrt{\frac{\tila +2}{\tila \tilde{\xi}_2^2+\tila+2}} \Bigg):=\tilk_0.
\end{align}
By Lyapunov's inequality \cite[Page 18]{petrov2012sums} and \eqref{eq:bbE tilD 4th=O(1)-2}, it follows that
\begin{align}
\label{eq:Di 3 O1}
\bbE\big[|D_i(\bLambda_\rmb)|^3\big|\bLambda_\rmb=\blambda_\rmb\big]\leq\Big(\bbE\big[|D_i(\bLambda_\rmb)|^4\big|\bLambda_\rmb=\blambda_\rmb\big]\Big)^\frac{3}{4}=O(1).
\end{align}
Likewise, the power scaling, on which the order of each parameter uniquely depends, is actually constrained to be of order $\Theta(\frac{1}{\sqrt{n}})$ in \eqref{eq:con power level}, thereby excluding arbitrarily small power levels, and hence from the definition of $\tilde{\xi}_2$ in \eqref{eq:def til xi 2} and \eqref{eq:Di 3 O1} we have $\tilde{\xi}_2=O(1)$ strictly.
Note that $\tila=O(\frac{1}{\sqrt{n}})$. Substituting the orders of $\tilde{\xi}_2$ and $\tila$ into $\tilk_0$ defined in \eqref{eq:def tilk0 sup phi} leads to $\tilk_0<1$.
From the fact that $\tilk_0<1$ and \eqref{eq:def tilk0 sup phi}, \emph{Condition 2} of Lemma \ref{lemma:cramer-esseen} is satisfied for $\{D_i(\bLambda_\rmb)\}_{i\in[n]}$.

\subsubsection{Proof of \eqref{eq:clt 2result}}
\label{subapx:clt rhs order 1/n}

In this part, we shall show that the right hand side of \eqref{eq:clt 1result} is of order $O(\frac{1}{n})$ to obtain \eqref{eq:clt 2result}, and we restate it here as
\begin{align}
\label{eq:def E123}
\underbrace{\frac{\bbE\big[D_i(\bLambda_\rmb)^3\big|\bLambda_\rmb=\blambda_\rmb\big]}{6\sqrt{2\pi}\sqrt{n}}\big(1-nu(\blambda_\rmb)^2\big)e^{-\frac{nu(\blambda_\rmb)^2}{2}}}_{:=\calE_1}-\underbrace{\frac{\tilk_2}{n}\bbE\big[|D_i(\bLambda_\rmb)|^4\big|\bLambda_\rmb=\blambda_\rmb\big]}_{:=\calE_2}-\underbrace{\tilk_2n^6\Big(\tilk_0+\frac{1}{2n}\Big)^n}_{:=\calE_3}.
\end{align}
For the first term $\calE_1$ in \eqref{eq:def E123}, we first calculate $\bbE[D_i(\bLambda_\rmb)^3|\bLambda_\rmb=\blambda_\rmb]$. Recall that $\tilD_{ij}$ is defined in \eqref{eq:def tilD_ij}, it follows from the definition of $D_i(\bLambda_\rmb)$ in \eqref{eq:def Di} that
\begin{align}
\bbE\big[D_i(\bLambda_\rmb)^3\big|\bLambda_\rmb=\blambda_\rmb\big]
&=\frac{1}{\sigma^3(\blambda_\rmb)}\cdot\bbE\Bigg[\bigg(\sum_{j\in[N_\rma]}\tilD_{ij}\bigg)^3\Bigg]\label{eq:D3 to tilD-1}\\*
&=\frac{1}{\sigma^3(\blambda_\rmb)}\cdot\sum_{j\in[N_\rma]}\bbE\Big[\tilD_{ij}^3\Big],\label{eq:D3 to tilD-2}
\end{align}
where \eqref{eq:D3 to tilD-2} follows from the fact that $\bbE[|\tilD_{ij}|]=0$. Similar to \eqref{eq:bbE tilD 4th-3} and \eqref{eq:bbE tilD 2nd-3}, by calculating we obtain $\bbE[\tilD_{ij}^3]=O(n^{-\frac{3}{2}})$. Combining with the fact that $\sigma^2(\blambda_\rmb)=O(n^{-\frac{1}{2}})$ leads to
\begin{align}
\frac{\bbE\big[D_i(\bLambda_\rmb)^3\big|\bLambda_\rmb=\blambda_\rmb\big]}{6\sqrt{2\pi}\sqrt{n}}=O(n^{-\frac{5}{4}}).
\end{align}
Recall that $u(\bLambda_\rmb)$ is defined in \eqref{eq:S sum normalized}. From the fact that $\gamma_n=C_{\varepsilon}(\Psi(n))+O(\frac{1}{n})$, $C_{\varepsilon}(\Psi(n))=O(n^{-\frac{1}{2}})$, $\mu(\bLambda_\rmb)=O(n^{-\frac{1}{2}})$ and $\sigma(\bLambda_\rmb)=O(n^{-\frac{1}{4}})$, it follows that $u(\bLambda_\rmb)=O(n^{-\frac{1}{4}})$, and then
\begin{align}
\calE_1
&=O(n^{-\frac{5}{4}})\cdot\big(1-O(n^\frac{1}{2})\big)\cdot O\big(\exp(-n^\frac{1}{2})\big)\\*
&=O\big( n^{-\frac{3}{4}}\exp(-n^{\frac{1}{2}}) \big).
\end{align}
For the second term $\calE_2$ in \eqref{eq:def E123}, combining \eqref{eq:bbE tilD 4th=O(1)-2} with the fact from Lemma \ref{lemma:cramer-esseen} that $\tilk_2$ is a constant, we obtain $\calE_2=O(n^{-1})$.
For the third term $\calE_3$ in \eqref{eq:def E123}, as $n\to\infty$, it follows that
\begin{align}
\calE_3
&=\tilk_2n^6\tilk_0^n\Big(1+\frac{1}{2\tilk_0 n}\Big)^n\\*
&=\tilk_2n^6\tilk_0^n e^{\frac{1}{2\tilk_0}}\label{eq:E3 order-2}\\*
&=O(\tilk_0^n),\label{eq:E3 order-3}
\end{align}
where \eqref{eq:E3 order-2} follows from the second important limit which implies $\lim_{x\to\infty}(1+\frac{1}{x})^x=e$, and \eqref{eq:E3 order-3} follows from the fact that $\tilk_2$ is a constant, $\tilk_0\in(0,1)$ and the exponential decay term $\tilk_0^n$ overwhelmingly dominates the polynomial factor $n^6$.
Consequently, the order of \eqref{eq:def E123} is dominated by the term $\calE_2$, i.e., the right hand side of \eqref{eq:clt 1result} is of order $O(n^{-1})$, and \eqref{eq:clt 2result} follows obviously.

\subsection{Proof of \eqref{eq:eliminate Q in S}}
\label{appendix:Proof of eliminate Q in S}

To obtain \eqref{eq:eliminate Q in S}, we shall average \eqref{eq:CLT S} over $\bLambda_\rmb$. For clarity, we restate \eqref{eq:CLT S} here as
\begin{align}
\label{eq:CLT S restate}
\Pr\big\{S_n(\bLambda_\rmb)\leq n\gamma_n\big|\bLambda_\rmb=\blambda_\rmb\big\}\geq Q(-\sqrt{n}u(\blambda_\rmb))+O\Big(\frac{1}{n}\Big).
\end{align}
Let $\tilU\sim\calN(0,1)$ be a real random variable independent of $\bLambda_\rmb$.
Taking expectations on both sides of \eqref{eq:CLT S restate} leads to
\begin{align}
\Pr\big\{S_n(\bLambda_\rmb)\leq n\gamma_n\big\}
&\geq \bbE\big[ Q(-\sqrt{n}u(\bLambda_\rmb))\big]+O\Big(\frac{1}{n}\Big)\label{eq:E Q to u-1}\\*
&=\bbE\bigg[ \Pr\bigg\{ u(\bLambda_\rmb)\geq\frac{\tilU}{\sqrt{n}} \bigg\} \bigg]+O\Big(\frac{1}{n}\Big)\label{eq:E Q to u-2}\\*
&=\bbE\bigg[ \Pr\bigg\{ u(\bLambda_\rmb)\geq\frac{\tilU}{\sqrt{n}} \bigg\}\bigg|\bLambda_\rmb \bigg]+O\Big(\frac{1}{n}\Big)\label{eq:E Q to u-3}\\*
&=\Pr\bigg\{ u(\bLambda_\rmb)\geq\frac{\tilU}{\sqrt{n}} \bigg\} +O\Big(\frac{1}{n}\Big),\label{eq:E Q to u-4}
\end{align}
where \eqref{eq:E Q to u-2} follows from the definition of the complementary Gaussian cumulative distribution function $Q(\cdot)$, \eqref{eq:E Q to u-3} follows from the fact that the unconditional probability coincides with the conditional probability when $\tilU$ and $\bLambda_\rmb$ are independent, and \eqref{eq:E Q to u-4} follows by the law of total expectation.
Let $f_U$ be the pdf of $u(\bLambda_\rmb)$ and $f'_U$ be the derivative of $f_U$.
From Lemma \ref{lemma:B sqrtn A ff'}, there exists an $\iota>0$ such that
\begin{align}
\Pr\bigg\{ u(\bLambda_\rmb)\geq\frac{\tilU}{\sqrt{n}} \bigg\}
&\geq \Pr\big\{ u(\bLambda_\rmb)\geq0 \big\}-\frac{1}{n}\bigg( \frac{2}{\iota^2}+\frac{f_U}{\iota}+\frac{f'_U}{2} \bigg)\label{eq:u to 0-0}\\*
&= \Pr\big\{ u(\bLambda_\rmb)\geq0 \big\}+O\Big(\frac{1}{\sqrt{n}}\Big)\label{eq:u to 0-1}\\*
&=\Pr\big\{ \mu(\bLambda_\rmb)\leq \gamma_n \big\}+O\Big(\frac{1}{\sqrt{n}}\Big)\label{eq:u to 0-2}\\*
&=F_\rmout(\gamma_n)+O\Big(\frac{1}{\sqrt{n}}\Big),\label{eq:u to 0-3}
\end{align}
where \eqref{eq:u to 0-1} follows from the fact that $f_U=O(n^{\frac{1}{4}})$ and $f'_U=O(n^{\frac{1}{2}})$, whose proofs are deferred to Appendices \ref{subapx:f U} and \ref{subapx:f' U}, respectively, \eqref{eq:u to 0-2} follows from \eqref{eq:S sum normalized}, and \eqref{eq:u to 0-3} follows from \eqref{eq:S sum mean} and the definition of the covert outage probability $F_\rmout(\cdot)$ in \eqref{eq:def of Fout_con}.
Substituting \eqref{eq:u to 0-3} into \eqref{eq:E Q to u-4} leads to \eqref{eq:eliminate Q in S}.
The proof is completed. To conclude the proof, it remains to calculate $f_U$ and $f'_U$.

\subsubsection{Proof of $f_U=O(n^{\frac{1}{4}})$}
\label{subapx:f U}
For clarity, let $U:=u(\bLambda_\rmb)$ and $u:=u(\blambda_\rmb)$. Note that we only focus on the order with respect to $n$, and assume that for each $j\in[N_\rma]$, $\Lambda_{\rmb,j}=\Lambda_\rmb$, which implies
\begin{align}
u(\bLambda_\rmb)=\frac{\gamma_n-N_\rma\log\big(1+\Lambda_\rmb\Psi(n)\big)}{\sqrt{N_\rma-N_\rma\big(1+\Lambda_\rmb\Psi(n)\big)^{-2}}}.
\end{align}
Given $\Lambda_\rmb=\lambda_\rmb$, define $\tilc:=\lambda_\rmb\Psi(n)$ and $\sigma_\dagger:=\sqrt{N_\rma-N_\rma(1+\tilc)^{-2}}$.
Let $f_{\Lambda_\rmb}$ be the pdf of $\Lambda_\rmb$. Since $U$ decreases monotonically with respect to $\Lambda_\rmb$, it follows that
\begin{align}
f_U(u)
&=f_{\Lambda_\rmb}(\lambda_\rmb)\cdot\bigg|\frac{\rmd u}{\rmd \lambda_\rmb}\bigg|^{-1}\label{eq:calculate U-1}\\*
&=f_{\Lambda_\rmb}(\lambda_\rmb)\cdot\bigg| \frac{1}{\sigma_\dagger^2}\cdot\Big( -\frac{N_\rma\Psi(n)\sigma_\dagger}{1+\tilc} - \big(\gamma_n-N_\rma\log(1+\tilc)\big)\cdot\frac{\rmd\sigma_\dagger}{\rmd\lambda_\rmb} \Big) \bigg|^{-1}\label{eq:calculate U-2}\\*
&=f_{\Lambda_\rmb}(\lambda_\rmb)\cdot\frac{\sigma_\dagger^3}{N_\rma\Psi(n)}\cdot\bigg| \frac{\sigma_\dagger^2}{1+\tilc} + \frac{\gamma_n-N_\rma\log(1+\tilc)}{(1+\tilc)^3} \bigg|^{-1}\label{eq:calculate U-3}\\*
&=f_{\Lambda_\rmb}(\lambda_\rmb)\cdot\frac{\sqrt{N_\rma}(\tilc^2+2\tilc)^\frac{3}{2}}{\Psi(n)\big| \big( \tilc^2+2\tilc-\log(1+\tilc) \big)N_\rma+\gamma_n \big|}\label{eq:calculate U-4}\\*
&=f_{\Lambda_\rmb}(\lambda_\rmb)\cdot\frac{\sqrt{N_\rma}(\tilc^2+2\tilc)^\frac{3}{2}}{\Psi(n)\big(\big( \tilc^2+\tilc+o(\tilc) \big)N_\rma+\gamma_n \big)}\label{eq:calculate U-5}\\*
&=f_{\Lambda_\rmb}(\lambda_\rmb)\cdot\frac{O(n^{-\frac{3}{4}})}{O(n^{-1})}\label{eq:calculate U-6}\\*
&=O(n^\frac{1}{4}),\label{eq:calculate U-7}
\end{align}
where \eqref{eq:calculate U-3} follows from the calculation that $\frac{\rmd\sigma_\dagger}{\rmd\lambda_\rmb}=\frac{N_\rma\Psi(n)}{\sigma_\dagger(1+\tilc)^3}$, \eqref{eq:calculate U-4} follows from the fact that $\sigma_\dagger:=\sqrt{N_\rma-N_\rma(1+\tilc)^{-2}}$, \eqref{eq:calculate U-5} follows from the Taylor Series of expansion of $\log(1+x)$ around $x=0$ which implies $\log(1+x)=x+o(x)$, \eqref{eq:calculate U-6} follows from the fact that $\Psi(n)=O(n^{-\frac{1}{2}})$, $\tilc=O(n^{-\frac{1}{2}})$, $\gamma_n=C_{\varepsilon}(\Psi(n))+O(\frac{1}{n})$ and $C_{\varepsilon}(\Psi(n))=O(n^{-\frac{1}{2}})$, and \eqref{eq:calculate U-7} follows from the fact that $f_{\Lambda_\rmb}(\lambda_\rmb)=O(1)$ and each parameter depends only on terms of order $\Theta(\frac{1}{\sqrt{n}})$ in \eqref{eq:con power level}, avoiding excessively small denominators.

\subsubsection{Proof of $f'_U=O(n^{\frac{1}{2}})$}
\label{subapx:f' U}
Since $U$ is monotonic and invertible with respect to $\Lambda_\rmb$, it follows that
\begin{align}
\label{eq:fU'=}
f'_U(u)=\frac{\rmd f_U}{\rmd u}=\frac{\rmd f_U}{\rmd\lambda_\rmb}\cdot\frac{\rmd\lambda_\rmb}{\rmd u}.
\end{align}
We have shown in \eqref{eq:calculate U-7} that $\frac{\rmd\lambda_\rmb}{\rmd u}=O(n^\frac{1}{4})$. In the following, we shall calculate $\frac{\rmd f_U}{\rmd\lambda_\rmb}$. Define two fuctions
\begin{align}
F_3&:=\sqrt{N_\rma}(\tilc^2+2\tilc)^\frac{3}{2},\label{eq:def F3}\\*
F_4&:=\Psi(n)\big(\big( \tilc^2+\tilc+o(\tilc) \big)N_\rma+\gamma_n \big),\label{eq:def F4}
\end{align}
and by \eqref{eq:calculate U-5}, it follows that
\begin{align}
\label{eq:F3F4''}
\frac{\rmd f_U}{\rmd\lambda_\rmb}=f'_{\Lambda_\rmb}(\lambda_\rmb)\cdot\frac{F_3}{F_4}+f_{\Lambda_\rmb}(\lambda_\rmb)\cdot\frac{F_3'F_4-F_3F_4'}{F_4^2}.
\end{align}
Recall that $\tilc=\lambda_\rmb\Psi(n)=O(n^{-\frac{1}{2}})$.
By calculating, we obtain
\begin{align}
F_3'
&=\sqrt{N_\rma}\cdot\frac{3}{2}(\tilc^2+2\tilc)^\frac{1}{2}\cdot2\Psi(n)(\tilc+1)\\*
&=3\sqrt{N_\rma}\Psi(n)(\tilc+1)\sqrt{\tilc^2+2\tilc}\\*
&=O(n^{-\frac{3}{4}}),\label{eq:F3'}
\end{align}
and
\begin{align}
F_4'
&=2N_\rma\Psi(n)^2\tilc+N_\rma\Psi(n)^2+o\big(N_\rma\Psi(n)\tilc\big)\\*
&=O(n^{-1}),\label{eq:F4'}
\end{align}
Substituting \eqref{eq:def F3}, \eqref{eq:def F4}, \eqref{eq:F3'} and \eqref{eq:F4'} into \eqref{eq:F3F4''}, and using the fact that $f_{\Lambda_\rmb}$ and $f'_{\Lambda_\rmb}$ are independent with $n$, and $F_3=O(n^{-\frac{3}{4}})$ and $F_4=O(n^{-1})$ by \eqref{eq:calculate U-6} lead to $\frac{\rmd f_U}{\rmd\lambda_\rmb}=O(n^{\frac{1}{4}})$.
Combining with \eqref{eq:calculate U-7} and \eqref{eq:fU'=} leads to $f'_U(u)=O(n^{\frac{1}{2}})$.

\bibliographystyle{ieeetr}
\bibliography{reference}
\end{document}

%% file: preamble.tex
\usepackage{soul}
\usepackage[mathscr]{eucal}
\usepackage[cmex10]{amsmath}
\usepackage{epsfig,epsf,psfrag}
\usepackage{amssymb,amsmath,amsthm,amsfonts,latexsym}
\usepackage{amsmath,graphicx,bm,xcolor,url,overpic}
\usepackage{fixltx2e}
\usepackage{array}
\usepackage{verbatim}
\usepackage{bm}
\usepackage{algorithmic}
\usepackage{algorithm}
\usepackage{verbatim}
\usepackage{textcomp}
\usepackage{mathrsfs}
\usepackage{epstopdf}

\newcommand{\openone}{\leavevmode\hbox{\small1\normalsize\kern-.33em1}}

\catcode`~=11 \def\UrlSpecials{\do\~{\kern -.15em\lower .7ex\hbox{~}\kern .04em}} \catcode`~=13

\allowdisplaybreaks[4]

\newcommand{\nn}{\nonumber}

\newcommand{\equ}{\mathrm{equ}}
\newcommand{\avg}{\mathrm{avg}}
\newcommand{\rmtG}{\mathrm{tG}}
\newcommand{\rmrt}{\mathrm{rt}}
\newcommand{\rmno}{\mathrm{no}}
\newcommand{\rmout}{\mathrm{out}}

\newcommand{\bbo}{\mathbbm{1}}

\newcommand{\calCN}{\mathcal{CN}}
\newcommand{\calA}{\mathcal{A}}
\newcommand{\calB}{\mathcal{B}}
\newcommand{\calC}{\mathcal{C}}
\newcommand{\calD}{\mathcal{D}}
\newcommand{\calE}{\mathcal{E}}
\newcommand{\calF}{\mathcal{F}}
\newcommand{\calG}{\mathcal{G}}
\newcommand{\calH}{\mathcal{H}}

\newcommand{\calN}{\mathcal{N}}

\newcommand{\calP}{\mathcal{P}}
\newcommand{\calQ}{\mathcal{Q}}

\newcommand{\calV}{\mathcal{V}}

\newcommand{\calX}{\mathcal{X}}
\newcommand{\calY}{\mathcal{Y}}
\newcommand{\calZ}{\mathcal{Z}}

\newcommand{\ba}{\mathbf{a}}
\newcommand{\bA}{\mathbf{A}}
\newcommand{\bb}{\mathbf{b}}
\newcommand{\bB}{\mathbf{B}}

\newcommand{\be}{\mathbf{e}}

\newcommand{\bh}{\mathbf{h}}
\newcommand{\bH}{\mathbf{H}}

\newcommand{\bI}{\mathbf{I}}

\newcommand{\bL}{\mathbf{L}}

\newcommand{\bP}{\mathbf{P}}

\newcommand{\bV}{\mathbf{V}}

\newcommand{\bx}{\mathbf{x}}
\newcommand{\bX}{\mathbf{X}}
\newcommand{\by}{\mathbf{y}}
\newcommand{\bY}{\mathbf{Y}}
\newcommand{\bz}{\mathbf{z}}
\newcommand{\bZ}{\mathbf{Z}}

\newcommand{\rma}{\mathrm{a}}

\newcommand{\rmb}{\mathrm{b}}

\newcommand{\rmc}{\mathrm{c}}

\newcommand{\rmd}{\mathrm{d}}

\newcommand{\rme}{\mathrm{e}}

\newcommand{\rmF}{\mathrm{F}}

\newcommand{\rmG}{\mathrm{G}}

\newcommand{\rmH}{\mathrm{H}}

\newcommand{\rmm}{\mathrm{m}}

\newcommand{\rmP}{\mathrm{P}}

\newcommand{\rmr}{\mathrm{r}}

\newcommand{\rmS}{\mathrm{S}}
\newcommand{\rmt}{\mathrm{t}}

\newcommand{\rmu}{\mathrm{u}}

\newcommand{\rmw}{\mathrm{w}}

\newcommand{\bbC}{\mathbb{C}}
\newcommand{\bbD}{\mathbb{D}}
\newcommand{\bbE}{\mathbb{E}}

\newcommand{\bbN}{\mathbb{N}}

\newcommand{\bbR}{\mathbb{R}}

\newcommand{\bbV}{\mathbb{V}}

\DeclareMathAlphabet{\mathbsf}{OT1}{cmss}{bx}{n}
\DeclareMathAlphabet{\mathssf}{OT1}{cmss}{m}{sl}

\newcommand{\rvQ}{\mathsf{Q}}

\DeclareSymbolFont{bsfletters}{OT1}{cmss}{bx}{n}
\DeclareSymbolFont{ssfletters}{OT1}{cmss}{m}{n}
\DeclareMathSymbol{\bsfGamma}{0}{bsfletters}{'000}
\DeclareMathSymbol{\ssfGamma}{0}{ssfletters}{'000}
\DeclareMathSymbol{\bsfDelta}{0}{bsfletters}{'001}
\DeclareMathSymbol{\ssfDelta}{0}{ssfletters}{'001}
\DeclareMathSymbol{\bsfTheta}{0}{bsfletters}{'002}
\DeclareMathSymbol{\ssfTheta}{0}{ssfletters}{'002}
\DeclareMathSymbol{\bsfLambda}{0}{bsfletters}{'003}
\DeclareMathSymbol{\ssfLambda}{0}{ssfletters}{'003}
\DeclareMathSymbol{\bsfXi}{0}{bsfletters}{'004}
\DeclareMathSymbol{\ssfXi}{0}{ssfletters}{'004}
\DeclareMathSymbol{\bsfPi}{0}{bsfletters}{'005}
\DeclareMathSymbol{\ssfPi}{0}{ssfletters}{'005}
\DeclareMathSymbol{\bsfSigma}{0}{bsfletters}{'006}
\DeclareMathSymbol{\ssfSigma}{0}{ssfletters}{'006}
\DeclareMathSymbol{\bsfUpsilon}{0}{bsfletters}{'007}
\DeclareMathSymbol{\ssfUpsilon}{0}{ssfletters}{'007}
\DeclareMathSymbol{\bsfPhi}{0}{bsfletters}{'010}
\DeclareMathSymbol{\ssfPhi}{0}{ssfletters}{'010}
\DeclareMathSymbol{\bsfPsi}{0}{bsfletters}{'011}
\DeclareMathSymbol{\ssfPsi}{0}{ssfletters}{'011}
\DeclareMathSymbol{\bsfOmega}{0}{bsfletters}{'012}
\DeclareMathSymbol{\ssfOmega}{0}{ssfletters}{'012}

\newcommand{\tila}{\tilde{a}}

\newcommand{\tilc}{\tilde{c}}

\newcommand{\tilD}{\tilde{D}}

\newcommand{\tilF}{\tilde{F}}

\newcommand{\tilg}{\tilde{g}}
\newcommand{\tilG}{\tilde{G}}

\newcommand{\tilk}{\tilde{k}}

\newcommand{\tilP}{\tilde{P}}

\newcommand{\tilt}{\tilde{t}}
\newcommand{\tilT}{\tilde{T}}

\newcommand{\tilU}{\tilde{U}}

\newcommand{\hatW}{\hat{W}}

\newcommand{\tilx}{\tilde{x}}

\newcommand{\blambda}{\bm{\lambda}}

\newcommand{\bLambda}{\bm{\Lambda}}
\newcommand{\bSigma	}{\bm{\Sigma}}
\newcommand{\bPsi}{\bm{\Psi}}

\newcommand{\iid}{i.i.d.\ }

\DeclareMathOperator{\diag}{diag}

\DeclareMathOperator{\tr}{tr}
\DeclareMathOperator{\spn}{span}

\DeclareMathOperator{\rank}{rank}

\newcommand{\bzero}{\mathbf{0}}

\newtheorem{theorem}{Theorem}
\newtheorem{lemma}[theorem]{Lemma}

\newtheorem{corollary}[theorem]{Corollary}
\newtheorem{definition}{Definition}
\newtheorem{example}{Example}

